\pdfoutput=1

\documentclass[11pt,letterpaper,margin=1in]{article}
\usepackage{graphicx}
\usepackage{placeins}
\usepackage{float}
\usepackage{fullpage}
\usepackage{amsmath, amssymb, amsthm, bbm}
\usepackage{thmtools}
\usepackage{mdframed}
\usepackage{mathtools}
\usepackage[dvipsnames]{xcolor}
\usepackage[shortlabels]{enumitem}
\usepackage{hyperref}
\usepackage{stmaryrd}
\usepackage{tikz}
\usepackage{subfig}
\usepackage{nicefrac}
\hypersetup{breaklinks, urlcolor=blue, colorlinks, citecolor=green!50!black, linkcolor=blue}
\usepackage[capitalize,noabbrev,nameinlink]{cleveref}
\usepackage[ruled,vlined,linesnumbered,algo2e]{algorithm2e}
\usepackage{todonotes}
\usetikzlibrary{decorations.pathreplacing, arrows.meta}
\usepackage{makecell}

\Crefname{algocf}{Algorithm}{Algorithms}
\Crefname{claim}{Claim}{Claims} 
\Crefname{problem}{Problem}{Problems}
\Crefname{fact}{Fact}{Facts}
\Crefname{observation}{Observation}{Observations}

\title{Matching Composition and \\Efficient Weight Reduction in Dynamic Matching}
\author{
Aaron Bernstein\thanks{
        New York University,
        \texttt{bernstei@gmail.com}. Supported by Sloan Fellowship, Google Research Fellowship,  NSF Grant 1942010, and Charles S. Baylis endowment at NYU.
    } \and
Jiale Chen\thanks{
  Stanford University,
  \texttt{jialec@stanford.edu}. Supported by a Lawrence Tang Graduate Fellowship, a Microsoft Research Faculty Fellowship, and NSF CAREER Award CCF-1844855.
} \and
Aditi Dudeja\thanks{
  University of Salzburg,
  \texttt{aditi.dudeja@plus.ac.at}. This work is supported by the Austrian Science Fund (FWF): P 32863-N. This project has received funding from the European Research Council (ERC) under the European Union’s Horizon 2020 research and innovation programme (grant agreement No 947702).
} \and
Zachary Langley\thanks{
  Rutgers University,
  \texttt{zach.langley@rutgers.edu}.
} \and
Aaron Sidford\thanks{
  Stanford University,
  \texttt{sidford@stanford.edu}.
    Supported by a Microsoft Research Faculty Fellowship, NSF CAREER Grant CCF-1844855, NSF Grant CCF-1955039, and a PayPal research award.
} \and
Ta-Wei Tu\thanks{
        Stanford University,
        \texttt{taweitu@stanford.edu}.
        Supported by a Stanford School of Engineering Fellowship, a Microsoft Research Faculty Fellowship, and NSF CAREER Award CCF-1844855.
    }}
\date{}
\declaretheorem[numberwithin=section,refname={Theorem,Theorems},Refname={Theorem,Theorems}]{theorem}
\declaretheorem[numberlike=theorem]{lemma}

\declaretheorem[numberlike=theorem]{corollary}
\declaretheorem[numberlike=theorem]{claim}
\declaretheorem[numberlike=theorem]{fact}
\declaretheorem[numberlike=theorem,style=definition]{definition}

\declaretheorem[numberlike=theorem]{observation}

\newmdtheoremenv{theo}{Theorem}

\newcommand{\bx}{\boldsymbol{x}}
\newcommand{\bw}{\boldsymbol{w}}

\newcommand{\eps}{\varepsilon}
\renewcommand{\epsilon}{\varepsilon}

\newcommand{\A}{\mathcal{A}}

\newcommand{\I}{\mathcal{I}}

\newcommand{\M}{\mathcal{M}}
\newcommand{\N}{\mathbb{N}}

\newcommand{\R}{\mathbb{R}}

\newcommand{\U}{\mathcal{U}}

\newcommand{\Z}{\mathbb{Z}}
\newcommand{\MWM}{\textrm{MWM}}

\newcommand{\defeq}{\stackrel{\mathrm{{\scriptscriptstyle def}}}{=}}

\DeclareMathOperator{\poly}{poly}

\DeclareMathOperator{\supp}{supp}

\begin{document}

\begin{titlepage}
  \maketitle \pagenumbering{roman}
  \setcounter{tocdepth}{3}
  
  \begin{abstract}		
  We consider the foundational problem of maintaining a $(1-\varepsilon)$-approximate maximum weight matching (MWM) in an $n$-node dynamic graph undergoing edge insertions and deletions. We provide a general reduction that reduces the problem on graphs with a weight range of $\mathrm{poly}(n)$ to $\mathrm{poly}(1/\varepsilon)$ at the cost of just an additive $\mathrm{poly}(1/\varepsilon)$ in update time. This improves upon the prior reduction of Gupta-Peng (FOCS 2013) which reduces the problem to a weight range of $\varepsilon^{-O(1/\varepsilon)}$ with a multiplicative cost of $O(\log n)$.

When combined with a reduction of Bernstein-Dudeja-Langley (STOC 2021) this yields a reduction from dynamic $(1-\varepsilon)$-approximate MWM in bipartite graphs with a weight range of $\mathrm{poly}(n)$ to dynamic $(1-\varepsilon)$-approximate maximum cardinality matching in bipartite graphs at the cost of a multiplicative $\mathrm{poly}(1/\varepsilon)$ in update time, thereby resolving an open problem in [GP'13; BDL'21]. Additionally, we show that our approach is amenable to MWM problems in streaming, shared-memory work-depth, and massively parallel computation models. We also apply our techniques to obtain an efficient dynamic algorithm for rounding weighted fractional matchings in general graphs. Underlying our framework is a new structural result about MWM that we call the ``matching composition lemma'' and new dynamic matching subroutines that may be of independent interest.

  \end{abstract}

  \newpage
  \tableofcontents
  \newpage
\end{titlepage}
\newpage
\pagenumbering{arabic}

\section{Introduction}

The \emph{maximum matching problem} is foundational in graph algorithms and has numerous applications. A \emph{matching} is a set of vertex-disjoint edges in an (undirected) graph. In \emph{unweighted graphs}, $G= (V,E)$, the problem, known as \emph{maximum cardinality matching (MCM)},  is to find a matching with the maximum number of edges (also known as the matching's \emph{size}). More generally, in \emph{weighted graphs}, $G = (V,E,w)$, where each $e \in E$ has weight $w(e) > 0$, the problem, known as \emph{maximum weighted matching (MWM)}, is to find a matching $M$ of maximum \emph{weight}, i.e., $\sum_{e \in M} w(e)$. %
For simplicity, throughout the paper, we assume that each $w(e) \in [1,W]$ for $W = \poly(n)$.

In the standard \emph{static} or \emph{offline} version of the maximum matching problem, it was recently shown how to compute maximum matchings in unweighted and weighted bipartite graphs in almost-linear time \cite{ChenKLPGS22,BrandCLPGSS23} (when the edge weights are integer). 
Though this almost resolves the complexity of the problem in the standard static, full-memory, sequential model of computation, the complexities of the problem in alternative models of computation such as dynamic, streaming, and parallel models are yet to be determined. The problem has been studied extensively in these models and there are conditional lower bounds that rule out efficient algorithms for exact maximum matching in certain settings (see, for example, \cite{HenzingerKNS15} for dynamic and \cite{GuruswamiO16} for streaming). 

Consequently, there has been work on efficiently computing \emph{approximately} maximum matchings. There is a range of approximation quality versus efficiency trade-offs studied (see e.g., \cite{BernsteinS15,BernsteinS16,GrandoniLSSS19,Wajc20,BhattacharyaKS23dynamic1}). We focus on the gold standard of computing (multiplicative) $(1-\eps)$-approximate matching, that is a matching of size or weight that is at least $(1-\epsilon)$ times the maximum. While there are algorithms that compute $(1-\eps)$-approximate maximum matchings in many computational models, there is often a significant gap between the bounds for weighted and unweighted graphs; many state-of-the-art results are limited to unweighted graphs. However, in dynamic, streaming, and parallel computational models, there is no clear indication of a fundamental separation between the computational complexity of the two cases. It is plausible that the existing gap could be due to a relative lack of techniques for working with weighted matchings.

Closing the gap between the
state of the art for weighted and unweighted matchings is an important open problem.
Several works have made progress on this problem by developing meta-algorithms that convert any algorithm for unweighted matching into an algorithm for weighted matching in a black-box fashion (albeit with potential loss in approximation quality and efficiency). 
The central focus of this paper is to provide improved meta-algorithms and new tools that more efficiently reduce weighted matching to unweighted matching. We motivate, develop, and introduce our results through the prominent dynamic matching problem (which we introduce next), though we also obtain results for streaming and parallel settings.

\paragraph{Dynamic Weighted Matching}
In the dynamic matching problem, a graph undergoes a sequence of adversarial updates, and the algorithm must
(explicitly) maintain an (approximate) maximum matching in the graph after each update.\footnote{Some papers instead maintain a data structure that can answer queries about a maximum matching, e.g.,~\cite{CharikarS18}. The low-recourse transformation we propose in~\cref{sec:low-recourse} can convert certain algorithms that maintain the matchings implicitly to ones that maintain them explicitly.
For the simplicity of the statement, unless stated otherwise, the dynamic matching algorithms in the paper should maintain a matching explicitly.
} The goal is to minimize the update time of the algorithm, which is the time needed to process a single update. In the most general \textit{fully dynamic} model, each update either inserts an edge into or deletes an edge from the graph. We also consider two natural, previously studied, \emph{partially dynamic} models, including the \emph{incremental} model, where each update can only insert an edge, and the \emph{decremental} model, where each update can only delete an edge. 

Over the past decades,
meta-algorithms for reducing dynamic matching on weighted to unweighted graphs have been developed (with different approximation and update time trade-offs). The first general reduction is by Stubbs and Williams~\cite{StubbsW17}, who show that any dynamic $\alpha$-approximate MCM algorithm can be converted to a dynamic $(1/2-\eps)\alpha$-approximate MWM algorithm with (multiplicative) $\poly(\log n,\eps^{-1})$ overhead in the update time.

More recently, the
state of the art was achieved by a result of Bernstein, Dudeja, and Langley~\cite{BernsteinDL21}.
This paper 
reduces the approximation error for weighted matching to $(1-\eps) \alpha$ in bipartite graphs and $(2/3-\eps)\alpha$ in non-bipartite graphs with $\eps^{-\Theta(1/\eps)}\log n$ multiplicative overhead in the update time. \cite{BernsteinDL21} crucially relies on different general reduction of Gupta and Peng~\cite{GuptaP13}, which reduces weighted matching in a general (not necessarily bipartite) graph with large weights to one with small weights---concretely, from real values in $[1,W]$ to integers in $\{1,\dots, \eps^{-O(1/\eps)}\}$---at the cost of an extra $(1-\eps)$-approximation factor and $O(\log n)$ multiplicative overhead. %

All of these reductions mentioned incur a multiplicative overhead of only $O_{\eps}(\poly(\log{n}))$ to the update time, where we use $O_{\eps}(\cdot)$ to hide factors depending on $\epsilon$. However, the dependence of $1/\eps$ in update time overhead in previous reductions for $(1-\eps)$-approximate MWM~\cite{GuptaP13,BernsteinDL21} are all exponential.
Consequently, even for $\eps = 1/O(\log n)$, an accuracy decaying slowly with increases in the graph's size, the algorithms may no longer achieve non-trivial update times.

\paragraph{Our Contribution}
Our first contribution is the following weighted-to-unweighted reduction in bipartite graphs, which settles the open problem of \cite{GuptaP13, BernsteinDL21} for bipartite graphs.
Interestingly, this reduction, along with the prior reductions in \cite{GuptaP13,StubbsW17,BernsteinDL21}, are \emph{partially dynamic preserving}, i.e., if the input unweighted matching algorithm is incremental or decremental, then the resulting weighted matching algorithm is also incremental or decremental respectively.

\begin{theorem}[Informal version of \Cref{red:WeightedtoUnweighted:new}]
\label{thm:main-unweighted-informal}
Given any dynamic $(1-\eps)$-approximate MCM algorithm in $n$-node $m$-edge \textbf{bipartite} graphs with update time $\U(n,m,\eps)$, there is a transformation which produces a dynamic $(1-O(\eps))$-approximate MWM algorithm for $n$-node
bipartite graphs with amortized update time $\U(\poly(1/\eps)\cdot n,\poly(1/\eps)\cdot m,\eps) \cdot \poly(1/\eps)$.
This transformation is partially dynamic preserving.
In non-bipartite graphs, the approximation ratio for weighted matching becomes $2/3-O(\eps)$. Moreover, if the unweighted algorithm is deterministic, then so is the weighted algorithm.
\end{theorem}

Beyond improving the exponential dependence on $1/\eps$ to polynomial, \Cref{thm:main-unweighted-informal} also eliminates the $O(\log{n})$ factors in the update-time
of \cite{BernsteinDL21}.
Therefore, \Cref{thm:main-unweighted-informal} implies that in dynamic bipartite graphs, $(1-\eps)$-approximate MWM shares the same complexity as $(1-\eps)$-approximate MCM, up to $\poly(1/\eps)$ factors.
\vspace{1.5ex}

In the case of general graphs, we make substantial progress towards reducing weighted matching to unweighted matching. Indeed, a crucial ingredient of our algorithm is a reduction from large weights to small ones, which applies to non-bipartite graphs as well.

\begin{theorem}[Informal version of \Cref{thm:partial reduction:ultimate}]
\label{thm:main-small-weights-informal}
Given any dynamic $(1-\eps)$-approximate MWM algorithm in $n$-node $m$-edge general (possibly non-bipartite) graphs with weights in $[1, W]$ with update time $\U(n,m,W,\eps)$, there is a transformation which produces a dynamic $(1-O(\eps))$-approximate MWM algorithm for $n$-node
general graphs with amortized update time $\U(\poly(1/\eps)\cdot n,\poly(1/\eps)\cdot m, \poly(1/\eps), \eps) \cdot \poly(\log(1/\eps))+\poly(1/\eps)$.
This transformation is partially dynamic preserving.
\end{theorem}

\Cref{thm:main-small-weights-informal} removes the exponential dependence on $1/\eps$ in \cite{GuptaP13} and again incurs no $\log n$ factors in the update-time overhead (whereas \cite{GuptaP13} incurs $\log n$). %

Additionally, our framework leads to a dynamic weighted \emph{rounding} algorithm with a polylogarithmic dependence on $W$, improving over that of \cite{ChenST23} (which depends linearly on $W$).
Here, a weighted rounding algorithm maintains an \emph{integral} matching supported on a dynamically changing \emph{fractional} one while approximately preserving its weight (see \cref{def:rounding}).
This
shows that dynamic \emph{integral} matching, weighted or not, is equivalent to dynamic \emph{fractional} matching (up to $\poly(\log n, 1/\eps)$ terms).
For example, as discussed in \cite{ChenST23}, this leads to a decremental $(1-\eps)$-approximate MWM algorithm in weighted general graphs with update time $\poly(\log n, 1/\eps)$.
Previously, rounding algorithms with $\poly(\log n,1/\eps)$ update times were only known for \emph{unweighted graphs}, bipartite~\cite{ArarCCSW18,Wajc20,BhattacharyaK21,BhattacharyaKSW24} and general~\cite{ChenST23,Dudeja24}.

\begin{theorem}[Informal version of~\cref{thm:weighted-rounding}]
  There is a dynamic weighted rounding algorithm with $\widetilde O(\poly(1/\eps))$ update time.
\end{theorem}

\paragraph{Techniques} Our key technical contribution is \cref{thm:main-small-weights-informal}, which removes the $\eps^{-O(1/\eps)}$ factor in an analogous result of \cite{GuptaP13}.
To reduce edge weights in weighted graph $G = (V,E, w)$ with vertices $V$, edges $E$, and edge weights $w$, both our reduction and the one in \cite{GuptaP13} define edge sets $E_i \subseteq E$, where each edge set is defined solely as a function of weights.
The algorithm then computes an \emph{arbitrary} $(1-\eps)$-approximate MWM $M_i$ in each $E_i$ and shows that the $M_i$ can be combined to compute an approximate MWM $M$ for the entire graph.
Consequently, the edge sets, $E_i$, are chosen to satisfy the following two properties: 
\begin{enumerate} 
\item Within each $E_i$, the ratio $\rho$ of the maximum to minimum edge weight is small (we call this the \emph{width} of the interval). This is the crux of our weight-reduction because a simple scaling and rounding approach yields the requested $M_i$ using only an algorithm for integer weights in $\{1, \ldots, \lceil\rho/\eps\rceil\}$.%
\item $\mu_w(M_1 \cup ... \cup M_k) \geq (1-\eps)\mu_w(G)$, where $\mu_w(\cdot)$ is the weight of the MWM in the input graph or edge set. This property is to ensure that the $M_i$ can be combined to obtain a $(1-\epsilon)$-approximate matching.  %
\end{enumerate}

In the algorithm of \cite{GuptaP13}, the intervals were disjoint; in fact, they were well-separated. This made it easy to prove property 2 above via a simple greedy combination. We show, however, that any set of disjoint intervals $E_i$ that satisfies property 2 must have, in the worst case, width $\rho \geq \exp(1/\eps)$. (See~\cref{sec:appendix:counterexample} for more details.)

To bypass this barrier, we allow for overlapping intervals. This forces us to use a different and more involved analysis, as the analysis of \cite{GuptaP13} crucially relied on disjointedness.
At the outset, it is not clear that this is even enough.
In fact, as we discuss below, it is only narrowly suffices in that our analysis crucially relies on the matchings $M_i$ being $(1-\eps)$-approximate, rather than $\alpha$-approximate for some constant $\alpha < 1$.

Our key technical contribution consists of two new structural properties of weighted matching, which we call matching \emph{composition} and \emph{substitution} lemmas (see \cref{lemma:matching substitution lemma,lemma:matching composition lemma}). On a high level, these two lemmas show that as long as the weight classes $E_i$ overlap slightly, a width of $\poly(1/\eps)$ is sufficient to ensure the second property above.

Using the above idea inside the framework of \cite{BernsteinDL21} with a simple greedy aggregation algorithm in~\cite{GuptaP13} immediately yields a weaker version of \cref{thm:main-unweighted-informal} (with $\poly(1/\eps) \cdot O(\log n)$ multiplicative overhead and a slightly worse dependence on $1/\eps$).
We further optimize the greedy aggregation in~\cite{GuptaP13} to remove the $\log n$ factor, improve the analysis of \cite{BernsteinDL21}, and propose a low-recourse transformation to remove several $\poly(1/\eps)$ factors.

\paragraph{A Limitation of Our Reductions} Although our results successfully remove the exponential dependence on $1/\eps$ in prior work, they also have a limitation. The work of \cite{GuptaP13} reduces $\alpha(1-\eps)$-approximate matching with general weights to $\alpha$-approximate matching with small edges weights; crucially, it works for any $\alpha \leq 1$. By contrast, our \cref{thm:main-small-weights-informal} requires $\alpha = (1-\eps)$. As a result, our \cref{thm:main-unweighted-informal} also requires $\alpha = (1-\eps)$, while the analogous result of \cite{BernsteinDL21} works for any $\alpha$.
Consequently, there are several applications related to smaller $\alpha$, for example, $\alpha = 1/(2-\sqrt{2})$ (see \cite{Behnezhad23,BhattacharyaKSW23dynamic2,AzarmehrBR24}) or $\alpha = 2/3$ (see \cite{BernsteinS15,BernsteinS16}) that benefit from \cite{BernsteinDL21} but not from our reduction. Notably, there are also conditional lower bounds ruling out efficient $(1-\eps)$-approximate matching algorithms in various models, specifically the fully dynamic~\cite{Liu24} and the single-pass streaming model~\cite{Kapralov21}.
Nonetheless, $(1-\eps)$-approximate matching is
a well-studied regime, and, as we show in \Cref{sec:applications}, our reduction improves the state of the art in multiple computational models. 

Perhaps surprisingly, this limitation is inherent to the general framework discussed above. More precisely, consider a scheme which picks edge sets $E_i = \{e \mid \ell_i \leq w(e) \leq r_i\}$, computes an arbitrary $\alpha$-approximate matching in each $M_i$, and then shows that $\mu_w(M_1 \cup \ldots \cup M_k) \geq (\alpha-\eps)\cdot \mu_w(G)$. Our key contribution is to show that for
$\alpha = (1-O(\eps))$, there exist suitable edge sets with width $\poly(1/\eps)$. However, our sets do not necessarily work for a fixed constant $\alpha < 1$; in fact, we show that for $\alpha < 1$, there are graphs for which \emph{any} suitable edge sets necessarily have width $\exp(1/\eps)$ 
(as what was done in \cite{GuptaP13}). 
Generalizing our result to all work for all $\alpha$ would thus require a different approach. See \cref{appendix:lower-bound} for more details.

\paragraph{Overview of the Paper}
In the remainder of the introduction, in \Cref{sec:intro-other-models}  we give an informal overview of our results in other computational models and in \Cref{sec:intro-applications} we discuss concrete applications. Thereafter, we state preliminaries in \Cref{sec:preliminaries}. In \Cref{sec:overview}, we give a technical overview including our main structural lemma that implies our reduction framework. We then show our main technical lemma in \Cref{sec:structural lemma} which helps to prove the main structural lemma, and state our framework in \Cref{sec:framework} in detail. We state applications of our framework in \Cref{sec:applications}. We conclude with open problems in \cref{sec:intro-open-problems}.

\subsection{Additional Computational Models}
\label{sec:intro-other-models}
\paragraph{Our Results}
Although the above theorems were written for dynamic models, our techniques are general and apply to a variety of different models. In \cref{sec:applications} we state formal reductions in a variety of models; here, we simply state the main takeaways. We show that 
analogs of the results in \cref{thm:main-unweighted-informal} and \cref{thm:main-small-weights-informal} also apply to the semi-streaming, massively parallel computing (MPC) with $O(n\log n)$ space per machine, and the parallel work-depth models.
We suspect they apply to other models as well, but in this paper, we focus on these four.

In the case of semi-streaming and MPC, the reduction when applied to existing algorithms, leaves the number of passes/rounds the same (up to a constant factor), but increases the space requirement by $\log n\cdot\poly(1/\eps)$. On the other hand, in the parallel work-depth model, the work increases by a factor of $\log n\cdot\poly(1/\eps)$, and the depth increases by an additive $\log^2 n$ factor. Despite these overheads, we can improve many of the state of the arts in these models.

Later in this section, we discuss these improvements and our contributions in more detail.

\paragraph{Contrast to Previous Work}
As in the case of dynamic algorithms, when we turn to other models our \cref{thm:main-unweighted-informal} and \cref{thm:main-small-weights-informal} achieve the same reductions as \cite{BernsteinDL21} and \cite{GuptaP13} respectively, except that we reduce their multiplicative overhead of $\eps^{-O(1/\eps)}$ to $\poly(1/\eps)$. There is also a different reduction of Gamlath, Kale, Mitrovic, and Svensson \cite{GamlathKMS19}, which works in both the streaming and MPC models, but not in the dynamic model. \cite{GamlathKMS19} has the advantage of reducing the most general case of weighted non-bipartite matching to the simplest case of unweighted bipartite matching, but it has exponential dependence on $1/\eps$ and it increases the number of \emph{passes/rounds} by a $\eps^{-O(1/\eps)}$ factor (which is generally considered a bigger drawback than the space increase). 

Similar to the dynamic models, our reductions have a somewhat narrower range of application than those of \cite{BernsteinDL21} and \cite{GuptaP13} even in MPC, parallel, and streaming models because ours do not work for general approximation factors: they only reduce a $(1-\eps)$-approximation to a $(1-\Theta(\eps))$-approximation. There is also a second, more minor drawback, which is that our reduction in \cref{thm:main-small-weights-informal} works in a slightly narrower range of models than the corresponding reduction of \cite{GuptaP13}. For example, the reduction of \cite{GuptaP13} applies to algorithms that only maintain the approximate \emph{size} of the maximum matching (see \cite{BhattacharyaKSW23dynamic2,Behnezhad23}), whereas our reduction only applies to algorithms that maintain the actual matching. But for the most part, our reduction and that of \cite{GuptaP13} apply to the same set of models.

\subsection{Applications}
\label{sec:intro-applications}

In this subsection, we give an informal overview of some of the implications of our reductions. For a more formal statement of the results, we refer the reader to \Cref{sec:applications}.

\paragraph{Applications to Bipartite Graphs}
Since our reduction from weighted to unweighted matching is black-box, it immediately improves upon the state of the art for weighted matching in a wide variety of computational models. Many of these results which achieve the state of the art were obtained by plugging existing unweighted algorithms into the reduction of \cite{BernsteinDL21}, and hence incur a multiplicative overhead of $\eps^{-O(1/\eps)}$. Plugging in our \Cref{red:WeightedtoUnweighted:new} reduces the multiplicative overhead to $\poly(1/\eps)$. In particular, our reduction obtains weighted analogs of the following \emph{unweighted bipartite} results. %
\begin{enumerate}
\item\label{item:GP13} A fully dynamic algorithm for maintaining a $(1-\eps)$-approximate
MCM in $O(\sqrt{m}\poly(1/\eps))$ time per update \cite{GuptaP13}.
\item\label{item:cec} A decremental $(1-\eps)$-approximate MCM algorithm with update time $\poly(\log n,\eps^{-1})$ \cite{BernsteinGS20,JambulapatiJST22}.
\item\label{item:Liu} A fully dynamic algorithm for maintaining a $(1-\eps)$-approximate MCM with update time $O(\poly(\eps^{-1})\cdot \frac{n}{2^{\Omega(\sqrt{\log n})}})$ \cite{Liu24}.
\item\label{item:Liu-offline} A fully dynamic \emph{offline} algorithm for maintaining a $(1-\eps)$-approximate MCM with update time $O(n^{0.58}\poly(\eps^{-1}))$; in the offline model, the entire sequence is known to the algorithm in advance \cite{Liu24}.

\item\label{item:inc} An incremental $(1-\eps)$-approximate MCM algorithm with update time $\poly(\eps^{-1})$ \cite{BlikstadK23}.
\item\label{item:streaming} An $O(\eps^{-2})$-pass, $O(n)$ space streaming algorithm for $(1-\eps)$-approximate MCM \cite{AssadiLT21}
\item\label{item:mpc} An $O(\eps^{-2}\cdot \log \log n)$-round, $O(n)$ space per machine, MPC algorithm for $(1-\eps)$-approximate MCM \cite{AssadiLT21}.

\end{enumerate}
Our reduction extends all of the above results to work in weighted graphs: the multiplicative overhead is only $\poly(1/\eps)$ in the dynamic model, as well as a $O(\log n)$ factor in some of the other models. Before our work, the weighted versions of \ref{item:GP13}, \ref{item:Liu}, and \ref{item:Liu-offline} had a multiplicative overhead of $\eps^{-O(1/\eps)}\cdot \log W$. For others mentioned on the list, separate weighted versions were known (see \cite{BhattacharyaKS23,LiuKK23}), but had worse dependence on either $\log n$ factors or $\eps^{-1}$ factors. Our main contribution here is to remove these overheads and, equally importantly, to streamline existing research by removing the need for a separate weighted algorithm.

Note that there are additional results on streaming approximate matching algorithm which obtain improved pass dependencies on $\epsilon$ at the cost of $\poly(\log n)$ factors \cite{AhnGuha11a,AhnG18,AssadiJJST22,Assadi24}. %
Our reduction does not improve the state of the art here for such methods. Therefore, we focus on algorithms that have pass complexities that only depend on $\epsilon$.

\paragraph{Applications to Non-Bipartite Graphs} Similar to \Cref{red:WeightedtoUnweighted:new}, our aspect-ratio reduction in \Cref{thm:partial reduction:ultimate} also works as a black-box. Each of the results below was initially obtained by applying the reduction of \cite{GuptaP13} to a weighted matching algorithm with a large dependence on $W$. By plugging in our reduction, we reduce the $\eps$-dependence in all of them from $\eps^{-O(1/\eps)}$ to $\poly(1/\eps)$.  
\begin{enumerate}
    \item A fully dynamic algorithm for maintaining a $(1-\eps)$-approximate MWM in general graphs in $\sqrt{m}\cdot \eps^{-O(1/\eps)}\cdot \log W$ time \cite{GuptaP13}.
    \item A decremental algorithm for maintaining a $(1-\eps)$-approximate MWM in general graphs in $\poly(\log n) \cdot \eps^{-O(1/\eps)}$ update time \cite{ChenST23}.
    \item A $\poly(\log n) \cdot \eps^{-O(1/\eps)}$ update time algorithm for rounding $(1-\eps)$-approximate weighted fractional matchings in general graphs to $(1-\Theta(\eps))$-approximate integral matchings \cite{ChenST23}.
\end{enumerate}

\newcommand{\Gint}[1]{G_I}

\section{Preliminaries}
\label{sec:preliminaries}

\paragraph{General Notation}
For positive integer $k$, we let $[k] \defeq \{1, \dots, k\}$.
For sets $S$ and $T$, we let $S \oplus T \defeq (S \setminus T) \cup (T \setminus S)$ denote their symmetric difference.

\paragraph{Graphs and Matchings}
Throughout this work, $G = (V, E)$ denotes an undirected graph, and $w: E \to \mathbb{R}_{>0}$ is an edge weight function.
The \emph{weight ratio} of $G$ is $\max_e w(e) / \min_{f} w(f)$.
A \emph{matching} $M \subseteq E$ is a set of vertex-disjoint edges.
The weight of a matching $M$, denoted $w(M)$, is the sum of the weights of the edges in the matching: $w(M) \defeq \sum_{e \in M} w(e)$.
We denote the maximum value of $w(M)$ achieved by any matching $M$ on $G$ by $\mu_w(G)$.
For $\alpha \in [0, 1]$, an $\alpha$-approximate MWM of $G$ is a matching $M$ such that $w(M) \ge \alpha \cdot \mu_w(G)$. The following result states that we can compute an $(1 - \eps)$-approximate MWM very efficiently.

\begin{theorem}[{\cite[Theorem 3.12]{DuanP14}}]\label{thm:DP14}
    On an $m$-edge general weighted graph, a $(1-\eps)$-approximate MWM can be computed in time $O(m\log(\eps^{-1})\eps^{-1})$.
\end{theorem}

\paragraph{Weight Intervals}
For $I \subseteq \R$, we denote by $\Gint{I}$ the subgraph of $G$ restricted to edges $e$ such that $w(e) \in I$. A set of (disjoint) weight intervals $[\ell_1,r_1),\dots,[\ell_k,r_k)\subseteq \R$ has
\emph{weight gap} $\delta$ if $\ell_{i+1} \geq \delta\cdot r_i$ for all $i \in [k - 1]$ and we call such a set of weight classes \emph{$\delta$-spread}. We also say that the set of intervals is \emph{$\delta$-wide} if $r_i \ge \delta \cdot \ell_i$ for all $i \in [k]$.
If $\ell_{i+1} = r_i$ and the intervals cover $[1, W]$, then we say the intervals are a \emph{weight partition}.

\paragraph{Computational Model}
We work in the standard Word-RAM model in which arithmetic operations over $\Theta(\log n)$-bit words can be performed in constant time.

\section{Technical Overview}\label{sec:overview}

Here we give an overview of our framework. For comparison and motivation, in \Cref{sec:overview:gp13} we first introduce a result of \cite{GuptaP13}, which provides a deterministic framework that is partially
dynamic preserving for
dynamic approximate MWM algorithms to reduce the weight range to $\eps^{-O(1/\eps)}$ with an overhead of $O(\log W)$. 
In~\cref{sec:overview:gaps-to-overlaps},
we explain the difficulty of reducing to $\poly(1/\eps)$ weight range using~\cite{GuptaP13}.
Motivated by this, in \cref{sec:overview:composition}, we introduce our key technical innovation, the
matching composition lemma~(\cref{lemma:matching composition lemma}), that allows us to bypass the barrier.
We then show that this lemma naturally induces our algorithmic framework that reduces the weight range down to $\poly(1/\eps)$.
Finally, in \cref{sec:overview:further-improvements}, we overview several further improvements that we made to shave off 
$\log n$ and $1/\eps$ factors from the running time.

\subsection{Weight Reduction Framework of Gupta--Peng}
\label{sec:overview:gp13}

The reduction framework of \cite{GuptaP13} works as follows. First, it groups edges geometrically by their weights so that the weight ratio of each group is $\Theta(1/\eps)$. It then deletes one group for every $\Theta(1/\eps)$ consecutive groups and merges the remaining consecutive groups. We refer to the merged groups as \emph{weight classes}; note that they are $\Theta(1/\eps)$-spread and have weight ratio $\eps^{-\Theta(1/\eps)}$. For each of these weight classes, a $(1-\eps)$-approximate MWM is maintained. Because the weight classes are $\Theta(1/\eps)$-spread, a simple greedy aggregation~\cite{AnandBGS12,GuptaP13,StubbsW17} of the $(1-\eps)$-approximate MWMs on each weight class leads to a $(1-O(\eps))$-approximation of $\mu_w(G)$.
As a result, this reduction reduces general approximate MWM to the problem of maintaining an approximate MWM inside a weight class with weight ratio $\eps^{-\Theta(1/\eps)}$.

More formally, the algorithm of \cite{GuptaP13} assigns all edges $e$ with weight $w_e\in[\eps^{-i},\eps^{-(i+1)})$ to the $i$th group. Let $G^{(j)}$ denote the graph obtained by deleting all groups $i$ such that $i\equiv j\bmod{\lceil \eps^{-1}\rceil}$. An averaging argument shows that $\max_{j} \mu_w(G^{(j)})\geq (1-O(\eps))\mu_w(G)$. So it suffices to maintain
a $(1-O(\epsilon))$-approximate MWM on each $G^{(j)}$ and return the one with maximum weight.

To do so, \cite{GuptaP13} merges all groups between neighboring deletions to form weight classes in $G^{(j)}$.
Those weight classes have weight ratio $\eps^{-\Theta(1/\eps)}$ and are $\Theta(1/\eps)$-spread.~\cite{GuptaP13} proceed by maintaining a $(1-\eps)$-approximate MWM $M^{(j)}_k$ on each weight class $[\ell^{(j)}_k,r^{(j)}_k)$; by scaling down appropriately, maintaining each $M^{(j)}_k$ requires maintaining a $(1-\eps)$-approximate MWM in a graph with edges in range $[1,\eps^{-\Theta(1/\eps)}]$, as desired. The authors of \cite{GuptaP13} then greedily aggregate the $M_i$ into a single matching $M$ by checking the $M_i$ in descending order of weight range and including in $M$ any edge that is not adjacent to existing edges in $M$. The $\Theta(1/\eps)$ weight gap between weight classes ensures that for each edge $e$ in the final matching $M$, the total weight of edges in $\bigcup M_i$ that were not included in $M$ because of $e$ is at most $O(\eps)\cdot w_e$. Since $(1-O(\eps))\mu_w(G)\leq (1-O(\eps))\mu_w(G^{(j)})\leq \sum_{i}w(M_i)$,
the greedy aggregation keeps a $1-O(\eps)$ fraction of the weight in $\sum_i w(M_i)$ thus is a $(1-O(\eps))$-approximate MWM.

\subsection{Disjoint Weight Classes Require Exponential Width}\label{sec:overview:gaps-to-overlaps}

The $\Theta(1/\eps)$ weight gap plays an important role in~\cite{GuptaP13} because it ensures
that $\mu_w(\bigcup M_i) \geq (1-O(\eps))\mu_w(G)$,
while also allowing for efficient greedy aggregation.
But as long as we try to maintain weight classes that are $1/\eps$-spread, it seems hard to reduce the weight ratio all the way down to $\poly(1/\eps)$. This is because it would require deleting a constant fraction of the initial weight groups (the ones of weight ratio $\Theta(1/\eps)$), so the MWM on the remaining graph $G^{(j)}$ would have a much smaller weight than $\mu_w(G)$. Indeed, we rule out the possibility of a broader family of methods that works with non-overlapping weight classes (which includes all methods that create weight gaps) by answering the following question in the negative.

\begin{restatable}{question}{WeightPartitionConjecture}\label{conjecture:weight partition}
    For any graph $G$, is there a weight partition $[\ell_1,r_1),\dots,[\ell_k,r_k)$ 
    such that $r_i/\ell_i\leq \poly(1/\eps)$ holds for all $i$ and given any set of $(1-\eps)$-approximate MWM $M_i$ on each $G_{[\ell_i,r_i)}$ we have
    \[\mu_w(M_1\cup M_2\cup\dots\cup M_k)\geq (1-O(\eps))\cdot \mu_w(G)?\]
    \label{question:conjecture}
\end{restatable}

To see why methods creating $\Theta(1/\eps)$ weight gaps are a special case of the weight partition allowed in \Cref{question:conjecture}, note that given any partition with gaps, we can naturally define a weight partition by letting each weight gap be its own weight class; if a large matching exists after deleting the edges in the gaps, it still exists when we keep those edges.

We give a counterexample (see~\cref{claim:counterexample of conjecture}) that answers~\cref{conjecture:weight partition} in the negative, even when the weight partition can be chosen depending on the input graph (recall that \cite{GuptaP13} chose the weight partition up front, oblivious to the structure of the input graph).
To see why this is the case, consider the gadget shown in \cref{fig:no overlapping} below.
Fix a partition $\mathcal{P}$ of $[1, W]$ into weight classes.
Observe that if this partition ``separates'' the gadget, i.e., some weight class $i$ in $\mathcal{P}$ contains only the weight $1$ but not $1.5$ (and the other weight class $j$ contains $1.5$), then it leads to an overall loss larger than $(1-\eps)$.
More concretely, if the matching $M_i$ in class $i$ restricted to the gadget contains the edge $bc$ (and not $ab$), then $\MWM(M_i \cup M_j) = 1.5$ while the entire gadget contains a matching $\{ab, cd\}$ of weight $2.5$.
The final counterexample we construct then contains multiple copies of the gadget with different weights and argues that any weight partition $\mathcal{P}$ with $r_i/\ell_i \leq \poly(1/\eps)$ for all $i$ must ``separate'' sufficiently many gadgets.
Consequently, it cannot preserve $(1-\eps)$-approximation.

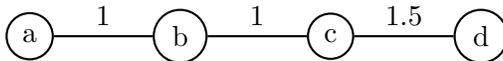
\begin{figure}[htbp]
  \centering
    \begin{tikzpicture}[auto,node distance=2cm, thick,main node/.style={circle,draw}]
      \node[main node] (a) {a};
      \node[main node] (b) [right of=a] {b};
      \node[main node] (c) [right of=b] {c};
      \node[main node] (d) [right of=c] {d};
      \draw (a) -- node {$1$} (b);
      \draw (b) -- node {$1$} (c);
      \draw (c) -- node {$1.5$} (d);
    \end{tikzpicture}
    \caption{Gadget for answering \cref{conjecture:weight partition} in the negative}
    \label{fig:no overlapping}
\end{figure}

\subsection{Leveraging Weight Overlaps: the Matching Composition Lemma}\label{sec:overview:composition}

How can we bypass this barrier?
Let us take a closer look at \cref{fig:no overlapping}.
In the gadget above, there are two possible $M_i$ for the weight class $i$ that contains the weight $1$: either it contains the edge $ab$ or $bc$.
As we have discussed above, the ``bad'' case is when $M_i$ contains $bc$ instead of $ab$, in which case the edge $bc$ will be ``kicked out'' by the edge $cd$ in $M_j$ and results in a weight loss.
Notice also that the weight loss affects the final approximation ratio \emph{when the weights of $bc$ and $cd$ are relatively close} (as depicted in \cref{fig:no overlapping})---if instead of $1.5$, the weight of $cd$ is changed to at least $1/\eps$, then the weights of $ab$ and $bc$ are negligible compared to $cd$ (up to an $\eps$ fraction), and it is okay if somehow $M_i$ contains $bc$ and it is ``kicked out'' by $cd$.
Therefore, to fix the issue, for any two weight values $w_1$ and $w_2$ that are close enough (in particular, $\eps \lesssim w_1/w_2 \lesssim 1/\eps$), we should have a weight class that contains both of them.
On the other hand, it is fine for $w_1$ and $w_2$ to not be in any weight class together if they are far apart.

Based on the observation, we see that if the weight classes are disjoint, then there will always be two close-enough weight values that are separated by the partition.
As a result, instead of creating weight gaps, we should leverage \emph{overlaps} between adjacent weight classes.
More concretely, we compute the approximate MWM for each weight class based on the information within the class and \emph{also} edges of neighboring weight classes. In other words, we enlarge the intervals in which we compute the matchings slightly. Perhaps surprisingly, we show that having an overlap of $\poly(1/\eps)$ allows us to bypass the above barrier completely, which we prove the following key technical lemma.

\begin{restatable}[Matching Composition Lemma]{lemma}{MatchingComposition}\label{lemma:matching composition lemma}
    Let $\eps\leq 1/6$ and $G$ be a graph, and consider a $\delta$-wide weight partition $[\ell_1,r_1),[\ell_2,r_2), \dots, [\ell_k,r_k)$. If $M_i$ is an arbitrary $(1-\eps)$-approximate MWM on $G_{[\eps\ell_i,r_i \eps^{-1})}$ for all $i \in [k]$, then
    \[\mu_w(M_1\cup M_2\cup\dots\cup M_k)\geq (1-O(\eps \cdot \log_{\delta}(1/\eps)))\cdot \mu_w(G).\]
\end{restatable}

On a high level, the matching composition lemma states that if we ``pad'' the weight classes $[\ell_i, r_i)$ slightly by a factor of $1/\eps$ in both directions, causing overlap, then we can effectively ``sparsify'' the graph by only considering MWM's on each $G_{[\eps \ell_i, r_i\eps^{-1})}$.
This readily leads to an algorithmic framework for weight reduction: fix a $\delta$-wide weight partition of the graph, maintain a $(1-\eps)$-approximate MWM $M_i$ on each ``padded'' weight class using the given dynamic algorithm, and then somehow aggregate them together to form the final output matching (i.e., maintain a $(1-\eps)$-approximate MWM on the union of $M_i$'s).
As long as the aggregation can be done efficiently, the output matching can be as well.
We give a more detailed overview in \cref{sec:overview:framework}.

\paragraph{The Matching Substitution Lemma}
The moment we introduce weight overlaps, we need a completely different analysis from that of \cite{GuptaP13} to prove that $\mu_w(\bigcup M_i) \geq (1-O(\eps))\mu_w(G)$. By forcing disjoint (and in fact spread) weighted classes, \cite{GuptaP13} ensured that any conflict between matchings $M_i$ and $M_j$ could always be resolved in favor of the higher weight class (hence greedy aggregation). But once there is weight overlap, it is not clear how conflicts should be resolved. Our new analysis thus requires a new structural understanding of weighted matching.

In particular,
the matching composition lemma is proved via a structural \emph{matching substitution lemma} formally stated below.
It asserts that one can effectively ``substitute'' parts of a matching $S$ with matchings $T_1, \ldots, T_k$ that come from certain weight classes.

\begin{restatable}[Matching Substitution Lemma]{lemma}{Hybrid}\label{lemma:matching substitution lemma}
 Let $G$ be a graph and $[\ell_1,r_1),\dots,[\ell_k,r_k)\subseteq\mathbb R$ be $(1/\eps)$-spread.
 For $\eps \leq 1/2$, given any
 matching $S\subseteq G$, and a batch of target matchings $\{T_i\subseteq G_{[\eps \ell_i,r_i \eps^{-1})}\mid i \in [k]\}$,
    there exists a matching $M\subseteq S\cup T_1\cup\cdots\cup T_k$ of weight
    \[w(M)\geq (1-4\eps)w(S) - \sum_{i \in [k]} \left(\mu_w(G_{[\eps\ell_i,r_i \eps^{-1})})-w(T_i)\right)\]
    such that $M \cap G_{[\ell_i,r_i)}\subseteq T_i$ for all $i \in [k]$.
\end{restatable}

The matching substitution lemma starts with an arbitrary source matching $S$, and a set of target matchings $T_1,\dots, T_k$ on the ``padded'' weight classes $[\eps\ell_1,r_1/ \eps),\dots,[\eps \ell_k,r_k / \eps)$.
It allows us to build a matching starting from $S$ and substitute all edges of $S$ in each weight class $[\ell_i,r_i)$ with edges contained $T_i$; this incurs some additive approximation error, but it is easy to check that the error is small as long as each $T_i$ is a near-maximum matching for the corresponding padded weight class.
For example, if we take $S$ to be a maximum weight matching on $G$ and set $T_i$ to be the $(1-\eps)$-approximate MWM $M_i$ on $G_{[\eps \ell_i, r_i\eps^{-1})}$ from~\cref{lemma:matching composition lemma}, then we can substitute each weight range of $S$ with edges from $M_i$ at minimal loss.
We defer the full proof to \cref{sec:structural lemma}.

\paragraph{Arbitrary Approximation Ratio}
As discussed in the introduction, our reductions only work for $(1-\eps)$-approximations, and not for arbitrary $\alpha$-approximations. In particular, in the matching composition lemma, if each $M_i$ is instead an arbitrary $\alpha$-approximate MWM (for some fixed $\alpha < 1$), then it is \textbf{\emph{not}} the case that $\mu_w(M_1 \cup \ldots \cup M_k) \geq (\alpha - O(\eps))\mu_w(G)$. Somewhat surprisingly, this limitation is not an artifact of our particular choice of weight classes, and turns out to be inherent to the general framework of composing approximate matchings between weight classes: for such a framework to work with any $\alpha$-approxmation (as does \cite{GuptaP13}), $\exp(1/\eps)$-wide weight classes are required. See \cref{appendix:lower-bound} for more details.

\subsubsection{Algorithmic Framework} \label{sec:overview:framework}

The matching composition lemma suggests the following algorithmic framework:
\begin{enumerate}
  \item Fix a $\delta$-wide weight partition of the input graph and maintain a $(1-\eps)$-approximate MWM $M_i$ on each padded weight class.
  \item\label{item:aggregate} Maintain a $(1-O(\eps))$-approximate MWM on the union of $M_i$'s as the output matching. By the matching composition lemma this is $(1-O(\eps\log_\delta(1/\eps)))$-approximate in the input graph.
\end{enumerate}

Scaling $\eps$ down by a factor of $O(\log_\delta(1/\eps))$ thus ensures that the matching we output is $(1-\eps)$-approximate in the input graph.
We now describe how we implement Step \labelcref{item:aggregate} efficiently.
For this we set $\delta = \Theta(\eps^{-3})$.
With this choice of $\delta$, even though the neighboring ``padded'' weight classes overlap, the sets of ``odd'' and ``even'' intervals are still each $\Theta(1/\eps)$-spread.
\begin{enumerate}[label=\textup{(2.\arabic*)},widest=(2.1),itemindent=*]
\item\label{item:greedy} As such, similar to \cite{GuptaP13}, these matchings can be separately aggregated using a greedy algorithm in $O(\log n)$ update time. More specifically, let $M_i$ be the matching in the $i$-th weight class. Then, we compute a $(1-\eps)$-approximate MWM $M_{\mathrm{odd}}$ (respectively, $M_{\mathrm{even}}$) on the union $M_1 \cup M_3 \cup \cdots$ (respectively, $M_2 \cup M_4 \cup \cdots$) greedily.
\item\label{item:degree-two} At this point, we are left with two matchings $M_{\mathrm{odd}}$ and $M_{\mathrm{even}}$ that we need to combine together.
This can be relatively easily handled in $O(1/\eps)$ time per change to $M_{\mathrm{odd}}$ and $M_{\mathrm{even}}$ since the union of these two matchings consists of only paths and cycles, and MWM can be computed and maintained very efficiently on them by splitting each connected component into paths of length $O(1/\eps)$ and solving each path individually via a dynamic program.
\end{enumerate}

As a result, with this choice of $\delta$ we arrive at a deterministic framework that reduces the aspect ratio from $\poly(n)$ to $\Theta(\eps^{-5})$ for any dynamic algorithm (note that the $\Theta(\eps^{-5})$ term comes from padding the $\Theta(\eps^{-3})$-wide intervals in each direction).

\begin{theorem}\label{thm:informal-reduction}
Given a dynamic $(1-\eps)$-approximate MWM algorithm in general (possibly non-bipartite) graphs with maximum weight in $[1, \poly(1/\eps)]$, there is a transformation that produces a dynamic $(1-O(\eps))$-approximate MWM algorithm in graphs with maximum weight $[1, W]$.
The reduction is partially dynamic preserving and has a multiplicative update time overhead 
of $\log n\cdot \poly(1/\eps)$.
The new weighted algorithm is deterministic if the initial algorithm is deterministic.
\end{theorem}

The weight reduction framework described above works for both bipartite and non-bipartite graphs. Moreover, combined with the unfolding framework of~\cite{BernsteinDL21}, it reduces weighted matching algorithms directly to unweighted matching algorithms in bipartite graphs with $\log n\cdot \poly(1/\eps)$ multiplicative overhead.

\subsection{Further Improvements} \label{sec:overview:further-improvements}

On top of the framework \cref{thm:informal-reduction}, we made the following additional improvements to decrease $\log n$ and $1/\eps$ factors in the final running time which may be of independent interest.

\paragraph{More Efficient Aggregation of Spread Matchings}
In Step~\labelcref{item:greedy} of our framework described above, we need to maintain a $(1-\eps)$-approximate MWM over matchings $M_1, \ldots, M_k$ whose weights are sufficiently spread apart (by a gap of $\Theta(1/\eps)$).

\begin{restatable*}[$(1-\eps)$-Approximate MWM over Matchings in $(1/\eps)$-Spread Weight Classes]{problem}{WellSeparated}\label{problem:well-separated weight classes}
    Given a set of $(1/\eps)$-spread weight classes $[\ell_1,r_1),\dots,[\ell_k,r_k)\subseteq \R$, and a set of $k$ matchings $M_1,\dots,M_k\subseteq G$ undergoing adversarial edge deletions/insertions satisfying $M_i\subseteq G_{[\ell_i,r_i)}$ for all $i \in [k]$.
    The task is to dynamically maintain a matching $M$ satisfying 
    \[w(M)\geq (1-O(\eps))\sum_{i\in[k]} w(M_i).\]
\end{restatable*}

The work of \cite{GuptaP13} solved \cref{problem:well-separated weight classes} with update time $O(k)$ using a \emph{greedy census matching} algorithm that was also used in \cite{AnandBGS12,StubbsW17}.
To improve upon this, we propose a different notion of \emph{locally} greedy census matching.
We show that the new notion suffices for maintaining a $(1-\eps)$-approximation and since, on a high level, the local version allows us to consider fewer edges in each update, we get a faster update time of $O(k/\log n)$.
Note that for Step~\labelcref{item:greedy}, the value of $k$ is\footnote{Note that this is because we assume the input graph has weights in $[1, \ldots, \poly(n)]$. %
} $O(\log n)$
and thus this shaves off the $O(\log n)$ factor in the update time that would have been there if we used the subroutine of \cite{GuptaP13}.
See~\cref{sec:framework:well-separated} for more details.

\paragraph{Low-Recourse Transformation}
Note that the overall update time of our framework also depends on the \emph{recourse} $\sigma$ of the given dynamic algorithm $\mathcal{A}$, i.e., the number of changes to the matching $M_i$ that it generates per update to the input graph.
This is because each such changes propagate to the internal dynamic subroutines, and for our case it will first correspond to an update to our algorithm for \cref{problem:well-separated weight classes}, and then be propagated to Step~\labelcref{item:degree-two} which has an update time of $O(1/\eps)$.
Our overall update time is thus $\mathcal{U} + \sigma/\eps$, where $\mathcal{U}$ is the update time of the dynamic algorithm $\mathcal{A}$.
Similar scenarios also occur in previous reductions of \cite{GuptaP13,BernsteinDL21}, and they both implicitly used the fact that $\sigma \leq \mathcal{U}$ (this is for algorithms that explicitly maintain a matching) and therefore their reductions incur a \emph{multiplicative} overhead in the update time of $\mathcal{A}$.

However, the output recourse can be much smaller than the update time.
For instance, for unweighted dynamic matching algorithms, the recourse can always be made $O(1/\eps)$ by a simple lazy update trick (the work of \cite{SolomonS21} further achieved a \emph{worst-case} recourse bound by ``smoothing'' the lazy update), while all known dynamic matching algorithms have update time much larger than this.
To address this disparity and make the overhead of our reduction \emph{additive}, we design a generic low-recourse transformation that converts, in a black-box way, \emph{any} $(1-\eps)$-approximate dynamic MWM algorithm to one with amortized recourse $O(\poly(\log W)/\eps)$.
This improves the na\"ive lazy update approach that would have a recourse bound of $O(W/\eps)$.
As our framework reduces the weight range to $W = \Theta(\eps^{-5})$, we use this new low-recourse transformation on the input algorithm $\mathcal{A}$ to decrease the additive overhead from $O(\eps^{-6})$ (with the na\"ive lazy update) to $O(\log(\eps^{-1})/\eps)$.

To improve the aspect ratio further than $\Theta(\eps^{-5})$, we continue to apply the framework on each $\Theta(\eps^{-5})$ intervals. Combined with the low-recourse transformation, we provide a trade-off between the aspect ratio and the efficiency of aggregation (see~\cref{corollary:partial reduction:low-recourse}). We use it to improve the fully dynamic low-degree algorithm in~\cite{GuptaP13} which then serves as another aggregation method that finally allows us to reduce the aspect ratio to $\Theta(\eps^{-2})$, the best we can get using~\cref{lemma:matching composition lemma}.
See~\cref{sec:framework} for more details.

\paragraph{The Final Transformation}
In the end, applying the improvements we discussed in this section, we obtain our final main theorem.

\begin{restatable}{theorem}{UltimateReduction}\label{thm:partial reduction:ultimate}
  Given a dynamic $(1-\eps)$-approximate MWM algorithm $\A$ that, on input $n$-vertex $m$-edge graph with aspect ratio $W$, has initialization time $\I(n, m, W, \varepsilon)$, and update time $\U(n, m, W, \varepsilon)$, there is a transformation which produces a dynamic $(1-\eps\log(\eps^{-1}))$-approximate MWM dynamic algorithm that has initialization time
  \[\log(\eps^{-1})\cdot O(\I(n, m, \Theta(\eps^{-2}), \Theta(\varepsilon))+m\varepsilon^{-1}),\]
  amortized update time
  \[\poly(\log(\eps^{-1}))\cdot O(\U(n, m, \Theta(\eps^{-2}), \Theta(\varepsilon))+\eps^{-6}),\]
  and amortized recourse
  \[\poly(\log(\eps^{-1}))\cdot O(\eps^{-6}).\]
  The transformation is partially dynamic preserving.
\end{restatable}

\section{Matching Composition and Substitution Lemmas}\label{sec:structural lemma}

We now turn to the proof of our key technical lemmas, the matching composition and substitution lemmas.
We first prove the matching substitution lemma, and then use it to deduce the matching composition lemma that ultimately leads to our algorithmic framework.

In the matching substitution lemma, we are given a source matching $S$ and target matchings $T_1, \dots, T_k$ on padded versions of $(1/\eps)$-spread intervals $[\ell_1, r_1), \dots, [\ell_k, r_k)$; more precisely, each $T_i$ is a matching on $[\eps \ell_i, r_i \eps^{-1})$.
The lemma states that we can find a new matching $M$ with $w(M) \approx w(S)$---assuming the matchings $T_i$ are large---such that $M$ only uses edges of $T_i$ on the weight interval $[\ell_i, r_i)$ for each $i \in [k]$.
The key idea in the proof is to identify a set of edges $D$ with small total weight (relative to $S$) to delete such that the edges of every component in $(S \cup T_1\cup \cdots \cup T_k) \setminus D$ are confined to a single weight class.

\Hybrid*

\begin{proof}
    We construct a sequence of matchings $M_0,M_1,\dots,M_k$, such that $M_0=S$ is the source matching, $M_i$ is constructed from $M_{i-1}\cup T_i$, and $M_k=M$ is the desired matching in the lemma.
    
    We first describe how to construct $M_i$ for $i\geq 1$. The components of $M_{i-1}\oplus T_i$ are only paths and cycles.
    Construct a set $D_i\subseteq M_{i-1}$ as follows: For each path or cycle $P\subseteq M_{i-1}\oplus T_i$ and $e\in P\cap M_{i-1}$ such that $w(e)\geq r_i \eps^{-1}$, in both directions of $P$, add the closest edges in $P\cap M_{i-1}$ of weight at most $r_i$ into $D_i$. For each $e\in P\cap M_{i-1}$ such that $\ell_i\leq w(e)< r_i \eps^{-1}$, in both directions of $P$, add the closest edges in $P\cap M_{i-1}$ of weight less than $\eps \ell_i$ into $D_i$.

    Now let $\widetilde M_{i-1} \defeq M_{i-1}\setminus D_i$ and again consider $\widetilde M_{i-1}\oplus T_i$ and any path or cycle $P\subseteq \widetilde M_{i-1}\oplus T_i$. Notice that if there is an $e\in P\cap \widetilde M_{i-1}$ such that $w(e)\in[\ell_i,r_i)$, then it must be the case that $P\subseteq G_{[\eps \ell_i,r_i \eps^{-1})}$. Let $L_i$ be the collection of such paths and cycles. We then construct $M_i\defeq\widetilde M_{i-1}\oplus L_i$, and thus $M_i\subseteq \widetilde M_{i-1}\cup T_i$ and $M_{i} \cap G_{[\ell_i,r_i)} \subseteq T_i$. It follows by induction that $M_k\subseteq M_{k-1}\cup T_k\subseteq\cdots\subseteq S\cup T_1\cup\cdots\cup T_k$ and $M_k\cap G_{[\ell_i,r_i)}\subseteq T_i$ for all $i$.

    We now analyze $w(M_k)$. Starting with $M_0$, two kinds of changes happened to the matching. The first one is the edge deletion $D_1\cup\cdots\cup D_k$, and the second one is the edge substitution through $L_1\cup\cdots\cup L_k$. We analyze the total weight loss in each part respectively.
    \begin{enumerate}
        \item Since $\ell_i\geq r_{i-1} \eps^{-1}$, only edges in $S$ cause deletion. For any edge $e\in S$, it could cause at most 2 edges deletions with respect to every weight class $[\eps \ell_i, r_i \eps^{-1})$. The weight of the deleted edges in the $i$th weight class would be at least $\eps$ smaller than $w_e$ and at most $r_i$. Since $r_i\geq \ell_i\geq r_{i-1} \eps^{-1}$, the total weight of those deleted edges would be at most $w(e)\cdot \left(2\eps+2\eps^2+\cdots\right)\leq 4\varepsilon\cdot w(e)$. Thus 
        \[w(D_1\cup\cdots\cup D_k)\leq 4\eps \cdot w(M).\]
        \item For each of the substitution induced by $L_i$, notice that $L_i\subseteq G_{[\eps \ell_i,r_i \eps^{-1})}$, thus
        \[\sum_{P\in L_i}w(P\cap \widetilde M_{i-1})-w(P\cap T_i)\leq \mu_w(G_{[\eps \ell_i,r_i \eps^{-1})})-w(T_i).\]
        Therefore, the second part leads to a total weight loss of at most \[
        \sum_{i\in[k]}\left(\mu_w(G_{[\eps \ell_i,r_i \eps^{-1})})-w(T_i)\right).
        \]
    \end{enumerate}
\end{proof}

We also need the following helper lemma.

\begin{restatable}{lemma}{WeightCombination}\label{corollary:weight combination with gap}
For $\varepsilon\leq 1/6$, any graph $G$ and any set of $(1/\eps)$-spread weight classes\\$[\ell_1,r_1),\dots,[\ell_k,r_k)\subseteq \R$, we have
\[\sum_{i\in[k]}\mu_w(G_{[\ell_i,r_i)})\leq (1+4\varepsilon)\cdot \mu_w(G).\]
\end{restatable}

\begin{proof}
  Suppose that $\ell_1 < \ell_2 < \cdots < \ell_k$.
  Let $M_i$ be a MWM on $G_{[\ell_i,r_i)}$, and let $H \defeq M_1 \cup \cdots \cup M_k$.
  Let $M$ be the matching obtained by the following greedy process:
  While $H$ is non-empty, we pick an edge $e$ in $H$ with the maximum weight and include it into $M$.
  Then, to ensure that the next edge we pick from $H$ still forms a matching with $M$, we remove all edges in $H$ that are adjacent to $e$ (including $e$ itself).
  Observe that if an edge $f$ is removed from $H$ by $e$, then we must have $w(f) \leq w(e)$.
  Let $i_e$ be such that $e \in M_{i_e}$.
  We also have that for each $j < i_e$, at most two edges from $M_j$ will be removed by $e$ (the two matched edges in $M_j$ for the endpoints of $e$).
  As the weight classes are $(1/\eps)$-spread, we have
  \[
    \sum_{f\;\text{removed by}\;e}w(f) \leq w(e) \cdot (1 + 2 \cdot (\eps + \eps^2 + \cdots)) \leq (1+4\eps) \cdot w(e).
  \]
  At the end of the process, $H$ will become empty.
  In other words, each edge in $H$ is removed by some edge in $M$.
  This shows that
  \[
    \sum_{i \in [k]}\mu_w(G_{[\ell_i, r_i)}) = \sum_{f \in H}w(f) \leq (1+4\eps) \cdot w(M) \leq (1+4\eps) \cdot \mu_w(G).
  \]
\end{proof}

We can now prove the matching composition lemma.

\MatchingComposition*

\begin{proof}
    Let $g = \lceil\log_{\delta}(1/\eps^3)\rceil+1$, and for all $j \in \{0, \dots, g - 1\}$, let $I_{j}=\{i \in [k] : i \equiv j\bmod g\}$. For the weight classes in each $I_j$, the weight gap between neighboring weight classes is at least $\delta^{g-1}\geq 1/\eps^3$. The set of weight classes $\{[\ell_i,r_i):i\in I_j\}$ is $(1/\eps^3)$-spread, and thus the set of padded weight classes $\{[\eps\ell_i,r_i\eps^{-1}):i\in I_j\}$ is $(1/\eps)$-spread. Consider any exact MWM $M^*$ on $G$. We will start with the initial source matching $S_0=M^*$, and for $j=0,1,\dots,g-1$, sequentially apply~\cref{lemma:matching substitution lemma} on the source matching $S_j$ and target matchings $\{M_i\mid i\in I_{j}\}$ and get $S_{j+1}$. For a fixed $j$, since $M_i$ is a $(1-\eps)$-approximate MWM on $G_{[\eps \ell_i,r_i \eps^{-1})}$ for $i\in I_j$, \cref{lemma:matching substitution lemma,corollary:weight combination with gap} give us a matching $S_{j+1}\subseteq S_{j}\cup (\bigcup_{i\in I_{j}} M_i)$ that satisfies
    \begin{align*}
    w(S_{j+1})&\geq (1-4\eps)w(S_{j})-\eps\cdot \sum_{i\in I_{j}}\mu_w(G_{[\eps \ell_i,r_i \eps^{-1})})\\
    &\geq (1-4\eps)w(S_{j})-\eps(1+4\eps)\cdot \mu_w(G)
    \geq (1-4\eps)w(S_{j})-3\eps\cdot \mu_w(G),
    \end{align*}
    and that $S_{j+1} \cap G_{[\ell_i,r_i)} \subseteq M_i$ for all $i \in I_j$.
    By induction, we have
    \[w(S_{j+1})\geq (1-4(j+1)\eps)\cdot w(S_0)-3(j+1)\eps\cdot \mu_w(G)\geq (1-7(j+1)\eps)\mu_w(G),\]
    and that
    \[S_{j+1}\cap G_{[\ell_i,r_i)}\subseteq \left(S_j\cap G_{[\ell_i,r_i)}\right)\cup \left(\bigcup_{t\in I_j} M_t\right)\subseteq\dots\subseteq \bigcup_{l=j^\prime}^j \bigcup_{t\in I_l} M_t\]
    hold for all $j^\prime \leq j$ and $i \in I_{j^\prime}$.
    Thus, we have
    \[w(S_{g})\geq (1-O(g\cdot\eps))\mu_w(G)\geq (1-O(\log_{\delta}(1/\eps)\cdot \eps))\mu_w(G),\]
    and
    \[ S_g\cap G_{[\ell_j,r_j)}\subseteq \bigcup_{i\in[k]}M_i\]
    for all $j \in [k]$.
    Therefore, the matching $S_g$ is contained in the union of all $M_i$'s and consequently
    \[\mu_{w}(M_1\cup M_2\cup\dots\cup M_k)\geq w(S_g)\geq (1-O(\log_{\delta}(1/\eps)\cdot \eps))\mu_w(G).\]
\end{proof}

\newcommand{\Modd}{M_{\textrm{odd}}}
\newcommand{\Meven}{M_{\textrm{even}}}
\newcommand{\Iodd}{I_{\textrm{odd}}}
\newcommand{\Ieven}{I_{\textrm{even}}}

\section{Framework}\label{sec:framework}
In this section, we will describe our framework in detail. As suggested by~\cref{lemma:matching composition lemma}, we first fix a $\Theta(\eps^{-3})$-wide weight partition, and compute a $(1-\eps)$-approximate MWM on each ``padded'' weight classes with aspect ratio $\Theta(\eps^{-5})$.
The choice of width ensures that the set of odd ``padded'' weight classes has $\Theta(1/\eps)$ weight gaps and so does the set of even ones. We use a subroutine~\cref{alg:locally greedy} in~\cref{sec:framework:well-separated} to aggregate odd matchings and even matchings, and maintain a $(1-\eps)$-approximate MWM on the union of them using the second subroutine~\cref{alg:degree-two} in~\cref{sec:framework:degree-two}. In~\cref{sec:framework:reduction} we give the complete framework to reduce the aspect ratio with multiplicative $\poly(1/\eps)$ overhead. In~\cref{sec:low-recourse}, we introduce a low-recourse transformation for $(1-\eps)$-approximate dynamic MWM to change the multiplicative $\poly(1/\eps)$ overhead to an additive $\poly(1/\eps)$ overhead. Finally, in~\cref{sec:putting everything together}, we use the low-recourse transformation to obtain an efficient fully dynamic algorithm on low-degree graphs, which leads to an efficient weighted rounding algorithm and could also serve as an efficient aggregation that allows us to reduce the aspect ratio to $O(\eps^{-2})$, which is the best we can hope for based on~\cref{lemma:matching composition lemma}. Also, combined with~\cite{BernsteinDL21} we achieve a $\poly(1/\eps)$ multiplicative overhead reduction from weighted matching algorithms to unweighted ones in bipartite graphs.

\subsection{Dynamic Approximate MWM on Matchings in \texorpdfstring{$(1/\eps)$-Spread}{(1-e)-Spread} Weight Classes}\label{sec:framework:well-separated}

Our first subroutine is an improved algorithm that combines matchings in weight classes that are sufficiently spread.
In particular, the goal is to solve the following problem.

\WellSeparated

As mentioned in \cref{sec:overview:further-improvements}, our improvement comes from maintaining the following \emph{locally greedy census} matchings.

\begin{restatable}[Locally Greedy Census]{definition}{LocalGreedy}
Consider $k$ matchings $M_1,M_2,\dots,M_k$. A matching $M$ is a \emph{locally greedy census matching} of $M_1,M_2,\dots,M_k\subseteq G$ if for every edge $e \in M_i\setminus M$, there exists an $f \in M_j$ such that $e \cap f \ne \emptyset$ for some $j> i$.   
\end{restatable}

The above local notion should be compared with the standard greedy census matching considered in \cite{AnandBGS12,GuptaP13,StubbsW17}.
In the standard notion, an edge can only be removed if it is incident to some higher-weight edge \emph{that is included into the output matching}.
In contrast to that, in our locally greedy census matching, if an edge is incident to \emph{any} higher-weight edge, regardless of whether that edge is in the output matching we are allowed to remove it.
This allows us to consider potentially much fewer edges when maintaining the local greedy census matching.
Nevertheless, we show that a similar charging argument can be used to prove the following guarantee.

\begin{restatable}{lemma}{LocalGreedyApproximation}\label{lemma:local greedy:approximation}
For $\varepsilon\leq 1/2$, any set of $(1/\eps)$-spread weight classes $[\ell_1,r_1),\dots,[\ell_k,r_k)\subseteq\mathbb R$, and matchings $M_1,\dots, M_k\subseteq G$ satisfying $M_i\subseteq G_{[\ell_i,r_i)}$ for all $i \in [k]$, every locally greedy census matching $M$ over the union of $M_1,\dots,M_k$ satisfies
\[
    w(M)\geq (1-4\eps)\sum_{i\in[k]} w(M_i).
\]
\end{restatable}

\begin{proof}
    The proof idea is similar to \cref{corollary:weight combination with gap}.
    For any edge $e\in M_j$, at most two edges in each lower weight class $i<j$ are 
    not included in $M$ because of $e$, and the total weight of these edges is at most
    \[\sum_{i \in [j-1]}2\cdot w(e)\cdot \eps^{-(i-j)}\leq \frac{2\eps}{1-\eps}\cdot w(e).\]
    Thus,
    \[w(M)\geq \left(1-\frac{2\eps}{1-\eps}\right)\sum_{i\in[k]}w(M_i)\geq (1-4\varepsilon)\sum_{i\in[k]} w(M_i). \]
\end{proof}

\begin{algorithm2e}[!ht]
  \caption{Dynamic Locally Greedy Census Matching} \label{alg:locally greedy}
  
  \SetEndCharOfAlgoLine{}

  \SetKwInput{KwData}{Input}
  \SetKwInput{KwResult}{Output}
  \SetKwInOut{State}{global}
  \SetKwProg{KwProc}{function}{}{}
  \SetKwFunction{Initialize}{Initialize}
  \SetKwFunction{Insert}{Insert}
  \SetKwFunction{Delete}{Delete}

  \KwProc{\Initialize{}} {
    \For{each node $u\in G$} {
      Initialize its neighborhood $N_u\gets\emptyset$.\;
    }
    \For{$j=k,\ldots,1$} {
      \For{$uv \in M_j$} {
        \If{$N_u=\emptyset$ and $N_v=\emptyset$}{
          Add $uv$ to the matching.\;
        }
        Add $uv$ to $N_u$ and $N_v$.\;
      }
    }
  }

  \KwProc{\Insert{$j,uv$}} {
    Add $uv$ to $N_u$ and $N_v$.\;
    \If{$u$ is matched to some vertex $u^\prime$, and $u u^\prime \in M_i$ such that $i<j$}{
      Delete $u u^\prime$ from the matching.
    }
    \If{$v$ is matched to some vertex $v^\prime$, and $vv^\prime \in M_i$ such that $i<j$}{
      Delete $vv^\prime$ from the matching.
    }
    \If{$uv$ is in the highest weight class among $N_u$ and $N_v$}{
      Add $uv$ to the matching.
    }
  }
  
  \KwProc{\Delete{$j,uv$}} {
    Delete $uv$ from $N_u$ and $N_v$.\;
    \If{$N_u$ is not empty}{
      $uu^\prime \gets$ the edge in the highest weight class among $N_u$.\;
      \If{$u u^\prime$ is in the highest weight class among $N_{u^\prime}$.}{
        Add $uu^\prime$ to the matching.\;
      }
    }
    \If{$N_v$ is not empty}{
      $vv^\prime \gets$ the edge in the highest weight class among $N_v$.\;
      \If{$vv^\prime$ is in the highest weight class among $N_{v^\prime}$.}{
        Add $vv^\prime$ to the matching.\;
      }
    }
  }
\end{algorithm2e}

We show that our modified notion allows us to maintain a locally greedy census matching more efficiently than what is achieved in \cite{StubbsW17} for the non-local version.
Remarkably, our algorithm achieves a constant update time when there are only $O(\log n)$ matchings.

\begin{restatable}{theorem}{LocalGreedyRuntime}\label{thm:local greedy:runtime}
    \cref{alg:locally greedy} initializes in $O(m)$ time and solves \cref{problem:well-separated weight classes} by dynamically maintaining a locally greedy census matching with $\min\{O(\log k),O(k/\log n)\}$ worst-case update time and $O(1)$ worst-case recourse.
\end{restatable}

\begin{proof}
    For any edge $uv$, it is contained in the locally greedy census matching if and only if it is in the highest weight class among $N_u$ and $N_v$. By definition, after the initialization, \cref{alg:locally greedy} maintains a locally greedy census matching. And the initialization takes $O(m)$ time.
    
    For an edge update $uv$, the only possible changes in the locally greedy census matching are in $N_u$ and $N_v$. For insertion of $uv$, \cref{alg:locally greedy} checks whether the edges related to $u$ and $v$ in the current matching still satisfies the condition, and whether $uv$ can be added. For deletion of $uv$, only the edges in the highest weight class among $N_u$ or $N_v$ can be added into the matching. \cref{alg:locally greedy} finds those edges and checks whether the condition is met. Therefore, it maintains a locally greedy census matching and the worst-case recourse is $O(1)$.

    For each node $u$, there is at most one edge from each weight class in $N_u$, i.e., $|N_u|\leq k$. To maintain the maximum element in $N_u$,
    we can use a binary search tree which runs in $O(\log k)$ time. Both $\texttt{Insert}$ and $\texttt{Delete}$ have $O(1)$ number of updates and queries to the binary search tree. Therefore, the update time would be $O(\log k)$.
    
    Alternatively, we can use a packed bit-representation of the weight-class information in $N_u$. We set the $i$-th bit to be 1 if and only if there is an edge in $N_u$ in the $(k-i)$-th weight class. Thus to find the edge in the highest weight class among $N_u$, it suffices to look at the lowest bit in the representation, which can be done in $O(k/\log n)$ time in the word RAM model.
\end{proof}

\subsection{Dynamic Approximate MWM on Degree-Two Graphs}\label{sec:framework:degree-two}

After combining the odd and even matchings with our locally greedy census matching algorithm, we are left with a union of two matchings which is a graph with maximum degree at most two.
That is, we need to solve the following problem.

\begin{restatable}[Fully-Dynamic $(1-\eps
)$-Approximate MWM on Degree-Two Graphs]{problem}{DegreeTwo}\label{problem:degree-two}
    Given a graph $G$ undergoing edge
    updates satisfying that its maximum degree is at most two. The task is to dynamically maintain a matching $M$ satisfying the following condition:
    \[w(M)\geq (1-O(\eps))\cdot \mu_w(G).\]
\end{restatable}

Observe that a degree-two graph consists of paths and cycles. Since an exact MWM on a path or cycle
$P$ can be computed in $O(|P|)$ time with dynamic programming, it suffices to maintain a collection of short paths and cycles on which a large-weight matching is supported.
For this, one can delete the minimum weight edge in each $\Theta(1/\eps)$-length neighborhood while keeping a $1-O(\eps)$ fraction of the total weight. We propose \cref{alg:degree-two} to solve \cref{problem:degree-two} by dynamically maintaining this $O(1/\eps)$-length decomposition of the paths and cycles and computing an exact MWM on each piece.

\begin{lemma}[Dynamic Path/Cycle Maintainer]\label{lemma:dynamic path-cycle maintainer}
There is a deterministic data structure $\mathcal D$ that maintains a set of dynamic paths or cycles $\{P_i\}$ under the insertion/deletion of edges and supports the following operations, where all update times and recourse mentioned are worst-case:
\begin{enumerate}
    \item Find the path/cycle $P_u$ that $u$ belongs to in $O(|P_u|)$ time.
    \item \emph{\texttt{FindHeads}$(P)$}: For a path $P$, find its both ends in $O(|P_u|)$ time.
    \item \emph{\texttt{Insert}/\texttt{Delete}$(uv)$}: Insert/delete an edge $uv$ in $O(|P_u|+|P_v|+1)$ time.
    \item \emph{\texttt{FindMin}$(P,h,\ell,r)$}: For a path $P$, find the edge with the minimum weight between the $\ell$-th and the $r$-th edges counting from $h$, one of the end of $P$, in $O(|P|)$ time.
    \item For a path/cycle $P$, explicitly maintain its MWM in $O(|P|)$ time and recourse. 
\end{enumerate}
\end{lemma}
\begin{proof}
    We can check all elements in a path in linear time in its size. Thus the first 4 operations are straightforward to achieve.
    Now we prove that it can output the MWM. Consider the following dynamic programming for computing MWM on paths. For a path $P$, number its edges from $P_1$ to $P_{|P|}$. Denote $f_{i,0/1}$ as the MWM on the path $P_{1}\dots P_i$ when $P_i$ is in the matching or not.
    For any $i\leq |P|$, $f_{i,0/1}$ can be computed by
    \[f_{i,x}=w_i\cdot x+\max_{0\leq y\leq 1-x} f_{i-1,y}.\]
    Therefore, the value can be computed in $O(|P|)$ time and the edge list corresponding to the MWM can be inferred by taking notes of how each state is updated.
\end{proof}

\begin{algorithm2e}[!ht]
  \caption{Fully-Dynamic $(1-\varepsilon)$-Approximate MWM on Degree-Two Graphs}
  
  \SetEndCharOfAlgoLine{}

  \SetKwInput{KwData}{Input}
  \SetKwInput{KwResult}{Output}
  \SetKwInOut{State}{global}
  \SetKwProg{KwProc}{function}{}{}
  \SetKwFunction{Initialize}{Initialize}
  \SetKwFunction{Insert}{Insert}
  \SetKwFunction{Delete}{Delete}
  \SetKwFunction{Maintain}{Maintain}

  \KwProc{\Initialize{}} {
    $\mathcal D\gets$ an instance of the dynamic path/cycle maintainer described in \cref{lemma:dynamic path-cycle maintainer}.\;
    $R\gets \emptyset$.\;
    \lFor{$uv \in E$}{\Insert{$uv$}.}
  }

  \KwProc{\Insert{$u, v$}} {
    $\mathcal D.\texttt{Insert}(uv)$.\;
    \lIf{$P_u$ is a path}{\Maintain{$P_u$}.}
  }
  
  \KwProc{\Delete{$uv$}} {
    \lIf{$uv\in R$}{
      $R\gets R\setminus uv$.
    }{
      $\mathcal D.\texttt{Delete}(uv)$.\;
      $h_u,u\gets \mathcal D.\texttt{FindHeads}(P_u)$.\;
      \If{there is an edge $h_uh_u^\prime\in R$}{
        $R\gets R\setminus h_u h_u^\prime$.\;
        $\mathcal D.\texttt{Insert}(h_u h_u^\prime)$.\;
        \Maintain{$P_u$}.\;
      }
      $h_v,v\gets \mathcal D.\texttt{FindHeads}(P_v)$.\;
      \If{there is an edge $h_v h_v^\prime \in R$}{
        $R\gets R\setminus h_v h_v^\prime$.\;
        $\mathcal D.\texttt{Insert}(h_v h_v^\prime)$.\;
        \Maintain{$P_v$}.\;
      }
    }
  }
  \KwProc{\Maintain{$P$}} {
    $h,t\gets \mathcal D.\texttt{FindHeads}(P)$.\;
    \If{$|P|\geq 3\lceil\eps^{-1}\rceil$}{
      $uv \gets \mathcal D.\texttt{FindMin}(P,h,\lfloor(|P|-\lceil\eps^{-1}\rceil)/2\rfloor,\lfloor(|P|-\lceil\eps^{-1}\rceil)/2\rfloor+\lceil\eps^{-1}\rceil-1)$.\label{line:degree two:split}\;
      $\mathcal D.\texttt{Delete}(uv)$.\;
      $R\gets R\cup \{uv\}$.\;
      \Maintain{$P_u$}.\;
      \Maintain{$P_v$}.\;
    }
  }
\label{alg:degree-two}
\end{algorithm2e}

\begin{lemma}\label{lemma:degree two:length}
    During the execution of \cref{alg:degree-two}, $\mathcal D$ maintains a set of paths/cycles with length at most $3\lceil\eps^{-1}\rceil$.
\end{lemma}
\begin{proof}
    The length of a path only increases after an edge insertion in $\mathcal D$, and \cref{alg:degree-two} calls $\Maintain$ every time which splits the path into two whenever its length is at least $3\lceil\eps^{-1}\rceil$. Thus the paths have lengths at most $3\lceil\eps^{-1}\rceil-1$. The only case $\mathcal D$ keeps a cycle is that before the formation of that cycle, the path has a length at most $3\lceil\eps^{-1}\rceil-1$. Therefore, the cycle has length at most $3\lceil\eps^{-1}\rceil$.
\end{proof}

\begin{restatable}{lemma}{DegreeTwoAlgo}\label{lemma:degree two:approximation and runtime}
    \cref{alg:degree-two} initializes in $O(m\varepsilon^{-1})$ time and solves \cref{problem:degree-two} by explicitly maintaining an matching with $O(\varepsilon^{-1})$ worst-case update time and $O(\varepsilon^{-1})$ worst-case recourse.
\end{restatable}

\begin{proof}
    \cref{lemma:degree two:length} show that $\mathcal D$ maintains a set of paths/cycles with length at most $3\lceil\eps^{-1}\rceil$. By \cref{lemma:dynamic path-cycle maintainer}, we know that each operation of $\mathcal{D}$ takes time $O(1/\eps)$. Now consider the recurrence in $\Maintain$. Any path that appears in $\Maintain$ has length at most $O(1/\eps)$ and will be at least $\lceil\eps^{-1}\rceil$ shorter in line~\ref{line:degree two:split}. Thus there are only $O(1)$ recurrences in $\Maintain$. Therefore, the worst-case update time of the algorithm is $O(1/\eps)$ and the worst-case recourse is $O(1/\eps)$.

    Now we show it maintains a $(1-O(\varepsilon))$-approximated MWM. In a degree-two graph, every connected component is either a path or a cycle, and $\mathcal D$ maintains an exact MWM on each component of $G\setminus R$ according to the last operation in \cref{lemma:dynamic path-cycle maintainer}. We know that $\mu_w(G)\geq \frac{1}{2}\sum_{e\in G} w_e$, since for each component in $G$, the ``odd'' edges and ``even'' edges both form a matching. On the other hand, an edge is added into $R$ only if it is the minimum among a set of $\lceil\eps^{-1}\rceil$ edges, and those sets are disjoint for different edges in $R$ since we only add edges to $R$ when the path is at least $3\lceil\eps^{-1}\rceil$ long. Thus $\sum_{e\in R} w_e\leq \eps\cdot \sum_{e\in G} w_e\leq 2\varepsilon\cdot \mu_w(G)$.
    Denote $M$ as the matching output by $\mathcal D$, we have
    \[w(M)=\mu_w(G\setminus R)\geq \mu_w(G)-\sum_{e\in R} w_e\geq (1-2\varepsilon)\mu_w(G). \]
\end{proof}

\subsection{Weight Reduction Framework for General Graphs}\label{sec:framework:reduction}
We are now ready to show our main result, a deterministic framework with $\poly(1/\eps)$ multiplicative overhead and recourse, which reduces the aspect ratio from $W$ to $\poly(1/\varepsilon)$ for any $(1-\varepsilon)$-approximate dynamic MWM algorithm.

\begin{restatable}{theorem}{Reduction}\label{thm:partial reduction}
  Given a dynamic $(1-\eps)$-approximate MWM algorithm $\A$ that, on input $n$-vertex $m$-edge graph with aspect ratio $W$, has initialization time $\I(n, m, W, \varepsilon)$, amortized/worst-case update time $\U(n, m, W, \varepsilon)$,  amortized/worst-case recourse $\sigma(n,m,W,\varepsilon)$, there is a transformation which produces a dynamic $(1-O(\eps))$-approximate MWM algorithm with initialization time
  \[O(\I(n, m, \Theta(\eps^{-5}), \Theta(\varepsilon))+m\varepsilon^{-1})\]
  time, amortized/worst-case update time
  \[O(\U(n, m, \Theta(\eps^{-5}), \Theta(\varepsilon))+\sigma(n,m,\Theta(\varepsilon^{-5}),\Theta(\varepsilon))\eps^{-1}),\]
  and amortized/worst-case recourse \[O(\sigma(n,m,\Theta(\varepsilon^{-5}),\Theta(\varepsilon))\eps^{-1}).\]
  The transformation is partially dynamic preserving.
\end{restatable}

\begin{algorithm2e}[!ht]
  \caption{Reduction Framework} \label{alg:reduction:partial}
  
  \SetEndCharOfAlgoLine{}

  \SetKwInput{KwData}{Input}
  \SetKwInput{KwResult}{Output}
  \SetKwInOut{State}{global}
  \SetKwProg{KwProc}{function}{}{}
  \SetKwFunction{Initialize}{Initialize}
  \SetKwFunction{Update}{Update}

  \KwData{A dynamic algorithm for $(1-\varepsilon)$-approximate maximum weight matching $\mathcal A$}
  \KwProc{\Initialize{}} {
    $L \gets \lfloor \log_{1/\eps}{W}\rfloor = \widetilde{O}(1)$.\;
    $E_{-1}=E_{L+1}=\emptyset$.\;
    \For{$i = 0, \ldots, L$} {
      $E_i \gets \{e \in E: \lfloor \log_{1/\eps}w(e)\rfloor = i\}$.\;
    }
    \For{$i = 1, \ldots, \lceil (L+1)/3\rceil$} {
      $\ell_i\gets 3i-3,r_i\gets \min(L,3i-1)$.\;
      $\widetilde E_i\gets\bigcup\limits_{j=\ell_i-1}^{r_i+1} E_j$.\;
      $\mathcal A_i\gets$ an independent copy of $\mathcal A$.\;
      Initialize $\mathcal A_{i}$ with $\widetilde E_i$.\;
      Denote $M_{i}$ as the matching maintained by $\mathcal A_{i}$.\;
    }
    $\mathcal C_1,\mathcal C_2\gets$ two independent copies of \cref{alg:locally greedy}.\;
    Initialize $\mathcal C_1$ with $\{M_i\mid i\equiv 1 \pmod 2\land 1\leq i\leq \lceil (L+1)/3\rceil\}$.\;
    Initialize $\mathcal C_2$ with $\{M_i\mid i\equiv 0 \pmod 2\land 1\leq i\leq \lceil (L+1)/3\rceil\}$.\;
    Denote $\Modd$ as the matching maintained by $\mathcal C_1$ and $\Meven$ as the one maintained by $\mathcal C_2$.\;
    $\mathcal M\gets$\cref{alg:degree-two}.\;
    Initialize $\mathcal M$ with $\Modd\cup \Meven$.\;
    \textbf{output} the matching maintained by $\mathcal M$.\;
  }

  \KwProc{\Update{$e$}} {
    $j\gets \lfloor \log_{1/\eps} w(e)\rfloor$.\;
    Update $E_j$ accordingly.\;
    \For{$i:1\leq i\leq \lceil (L+1)/3\rceil\land \ell_i-1\leq j\leq r_i+1$}{
        Update $\widetilde E_i$ based on the update in $E_j$.\;
        Use $\mathcal A_i$ to maintain $M_i$ based on the update in $\widetilde E_i$.\;
        \lIf{$i$ is odd}{
          Use $\mathcal C_1$ to maintain $\Modd$ based on the update in $M_i$.
        }
        \lElse{
          Use $\mathcal C_2$ to maintain $\Meven$ based on the update in $M_i$.
        }
    }
    Feed the updates in $\mathcal C_1$ and $\mathcal C_2$ into $\mathcal M$.\;
    \textbf{output} the matching maintained by $\mathcal M$.\;
  }
\end{algorithm2e}

There are three steps in \cref{alg:reduction:partial}. In the first step, for all $1\leq i\leq \lceil (L+1)/3\rceil$, $M_i$ is maintained by $\mathcal A_i$ and is a $(1-\varepsilon)$-approximation of $\mu_w(\widetilde E_i)$. In the second step, we use the locally greedy census matching \cref{alg:locally greedy} to aggregate $M_i$ for odd $i$ and even $i$ respectively, into $\Modd$ and $\Meven$, with the guarantee from \cref{lemma:local greedy:approximation} that $\Modd$ and $\Meven$ both keep at least a $(1-4\varepsilon)$ fraction of the total weight of the corresponding matchings. Then we use \cref{alg:degree-two} for degree-two graphs to aggregate $\Modd$ and $\Meven$, and \cref{lemma:degree two:approximation and runtime} shows that the final matching output by \cref{alg:reduction:partial} is a $(1-2\varepsilon)$-approximated MWM on $\Modd\cup \Meven$. We will prove that since at each step we lose a $O(\eps)$ fraction, the final matching we output keeps a $(1-O(\eps))$-approximate MWM.

\begin{lemma}\label{lemma:reduction:partial:approximation}
    For $\varepsilon\leq 1/2$, \cref{alg:reduction:partial} maintains a matching $M$ with $\mu_w(M)\geq (1-O(\varepsilon))\mu_w(G)$.
\end{lemma}
\begin{proof}
    \cref{lemma:matching composition lemma} shows that
    \[\mu_w(M_1\cup M_2\cup\dots\cup M_{\lceil(L+1)/3\rceil})\geq (1-O(\eps)) \mu_w(G).\]
    Consider the locally greedy census matching $\Modd$ and $\Meven$. Denote $\Iodd=\{1\leq i\leq\lceil (L+1)/3\rceil:i\text{ is odd}\}$, and $\Ieven=\{1\leq i\leq\lceil (L+1)/3\rceil:i\text{ is even}\}$. \cref{lemma:local greedy:approximation} shows that
    \[w(\Modd)\geq (1-O(\eps))\sum_{i\in \Iodd} w(M_i)\quad\text{and}\quad w(\Meven)\geq (1-O(\eps))\sum_{i\in \Ieven} w(M_i).\]
    Also, we know
    \[\Modd\subseteq \bigcup_{i\in \Iodd} M_i\quad\text{and}\quad \Meven\subseteq \bigcup_{i\in \Ieven} M_i,\]
    thus
    \begin{align*}
        \mu_w(\Modd\cup \Meven)
        &\geq \mu_w\left(\bigcup_{i=1}^{\lceil(L+1)/3\rceil} M_i\right)-\left(\sum_{i\in \Iodd}w(M_i)-w(\Modd)\right)-\left(\sum_{i\in \Ieven}w(M_i)-w(\Meven)\right)\\
        &\geq (1-O(\eps))\mu_w(G)-O(\eps)\cdot \sum_{i\in \Iodd}w(M_i)-O(\eps)\cdot \sum_{i\in \Ieven}w(M_i).
    \end{align*}
    \cref{corollary:weight combination with gap} shows that
    \[\sum_{i\in \Iodd} w(M_i)\leq (1+4\eps)\mu_w(G)\quad\text{and}\quad\sum_{i\in \Ieven} w(M_i)\leq (1+4\eps)\mu_w(G),\]
    and thus
    \[\mu_w(\Modd\cup \Meven)\geq (1-O(\eps))\mu_w(G).\]
    The final matching $M$ we output is a $(1-2\eps)$-approximate MWM on $\Modd\cup \Meven$. Therefore,
    \[\mu(M)\geq (1-O(\eps))\mu_w(G). \]

\end{proof}

\begin{lemma}\label{lemma:reduction:partial:runtime}
    \cref{alg:reduction:partial} initializes in $O(\mathcal I(n,m,\Theta(\varepsilon^{-5}),\Theta(\varepsilon))+m\varepsilon^{-1})$ time, and has update time $O(\mathcal U(n,m,\Theta(\varepsilon^{-5}),\Theta(\varepsilon))+\sigma(n,m,\Theta(\varepsilon^{-5}),\Theta(\varepsilon))\varepsilon^{-1})$ and recourse $O(\sigma(n,m,\Theta(\varepsilon^{-5}),\Theta(\varepsilon))\varepsilon^{-1})$.
\end{lemma}

\begin{proof}
    Each weight class $[\ell,r)$ has $\Theta(\eps^{-5})$ aspect ratio. Thus each edge update $e\in E_j$ causes $\sigma(n,m,\Theta(\eps^{-5}),\Theta(\eps))$ changes in corresponding $M_i$s which take $O(\mathcal U(n,m,\Theta(\eps^{-5}),\Theta(\eps)))$ update time. By \cref{thm:local greedy:runtime}, $\mathcal C_1,\mathcal C_2$ both handle each of these changes in $O(\log_{1/\varepsilon} W/\log n)=O(1)$ time and recourse, thus the update time of $\mathcal M$ would be $O(\sigma(n,m,\Theta(\eps^{-5}),\Theta(\eps))\varepsilon^{-1})$ and the recourse is at most $O(\sigma(n,m,\Theta(\eps^{-5}),\Theta(\eps))\eps^{-1})$. The initialization time follows from that of each subroutine.
\end{proof}

By repeatedly applying~\cref{lemma:matching composition lemma} and~\cref{alg:degree-two}, we can further reduce the aspect ratio.

\begin{restatable}{theorem}{partialreduction}\label{thm:partial reduction:2}
  Given a dynamic $(1-\eps)$-approximate MWM algorithm $\A$ that, on input $n$-vertex $m$-edge graph with aspect ratio $W$, has initialization time $\I(n, m, W, \eps)$, amortized/worst-case update time $\U(n, m, W, \eps)$, and amortized/worst-case recourse $\sigma(n,m,W,\eps)$, there is a transformation that produces a dynamic $(1-O(\eps))$-approximate MWM algorithm that has initialization time
  \[O(\I(n, m, \Theta(\eps^{-2-3\cdot 2^{-d}}), \Theta(\eps))+m\eps^{-1}),\]
  amortized/worst-case update time
  \[O(\U(n, m, \Theta(\eps^{-2-3\cdot 2^{-d}}), \Theta(\eps))+\sigma(n,m,\Theta(\eps^{-2-3\cdot 2^{-d}}),\Theta(\eps))\eps^{-(1+d)}),\]
  and amortized/worst-case recourse
  \[O(\sigma(n,m,\Theta(\eps^{-2-3\cdot 2^{-d}}),\Theta(\eps))\eps^{-(1+d)})\]
  for any integer parameter $d\in\Z_{\geq 0}$.
  The transformation is partially dynamic preserving.
\end{restatable}

\begin{proof}
    For each weight class $[\ell,r)$ with aspect ratio $\Theta(\eps^{-2-3\cdot 2^{-x}})$ (we start with $x=0$, i.e., $\Theta(\eps^{-5})$), denote $m=\sqrt{\ell\cdot r}$. There is a consistent constant $c$ in \cref{lemma:matching composition lemma} such that given two $(1-\eps)$-approximate MWM $M_1$ and $M_2$ on the ``padded'' weight classes $[\ell,m\cdot \eps^{-1})$ and $[m\cdot \eps,r)$ respectively, there is a matching of the weight class $[\ell,r)$ on $M_1\cup M_2$ with approximation ratio 
    \[\left(1-c\cdot\left(\log(1/\eps)/\log({\eps^{-1-3\cdot 2^{-(x+1)}}})\right)\cdot \eps\right)=\left(1-\frac{c}{1+3\cdot 2^{-(x+1)}}\cdot\eps\right)\geq (1-c\cdot \eps).\]
    \cref{alg:degree-two} can maintain a $(1-\eps)$-approximate MWM on the union, thus maintain a $(1-(c+1)\cdot\eps)$-approximate matching, and we reduce the aspect ratio from $\Theta(\eps^{-2-3\cdot 2^{-x}})$ to $\Theta(\eps^{-2-3\cdot 2^{-(x+1)}})$. Repeatedly applying~\cref{lemma:matching composition lemma} and~\cref{alg:degree-two} for $d$ times, we achieve a $d$-depth binary tree representation of the weight reduction. Each inner node of the binary tree is a matching maintained on the union of its offspring. Since for each layer, we lose a $c+1$ factor in the approximation error, the matching maintained at the root has an approximation ratio $1-(c+1)^d\cdot\eps$.
    
    Since there are $2^d=O(1)$ nodes in the binary tree, the initialization takes time
    \[O(\I(n, m, \Theta(\eps^{-2-3\cdot 2^{-d}}), \Theta(\varepsilon))+m\eps^{-1}).\]
    For the update time and recourse, consider the layers in decreasing depths. In the deepest layer with depth $d$, there are $2^d=O(1)$ nodes. The edge change could occur in each of them, so there is an update time $O(\U(n, m, \Theta(\eps^{-2-3\cdot 2^{-d}}), \Theta(\varepsilon)))$ and recourse $O(\sigma(n,m,\Theta(\eps^{-2-3\cdot 2^{-d}}),\Theta(\varepsilon)))$.
    For the layer with depth $d-1$, the number of edge updates in total equals the recourse of the layer with depth $d$, thus both the update time and recourse would be $O(\sigma(n,m,\Theta(\eps^{-2-3\cdot 2^{-d}}),\Theta(\varepsilon)))\eps^{-1}$. Suppose $\eps\leq 1/2$, an easy induction shows that the total update time would be
    \[O(\U(n, m, \Theta(\eps^{-2-3\cdot 2^{-d}}), \Theta(\varepsilon))+\sigma(n,m,\Theta(\eps^{-2-3\cdot 2^{-d}}),\Theta(\varepsilon))\eps^{-(1+d)}),\]
    and the recourse would be
    \[O(\sigma(n,m,\Theta(\eps^{-2-3\cdot 2^{-d}}),\Theta(\varepsilon))\eps^{-(1+d)}).\]
    Combined with~\cref{thm:partial reduction} we finish the proof.
\end{proof}

\subsection{Low-Recourse Transformation}\label{sec:low-recourse}
The update time of our reduction comprises two parts:
the original update time of the algorithm $\A$ and its recourse.
According to~\cref{thm:partial reduction:2},
the multiplicative overhead on the update time is constant
while that on the recourse is $\poly(1/\eps)$.
A high recourse of the algorithm could make it inefficient
when serving as a subroutine. \cite{SolomonS21} provides
a low-recourse transformation that reduces the recourse to worst-case $O(W/\eps)$
for any $\alpha$-approximate dynamic MWM algorithm. In this section, we design a
tailored low-recourse transformation for $(1-\eps)$-approximate
weight matching that reduces the recourse to amortized $O(\poly(\log W)/\eps)$
(see~\cref{thm:our low-recourse}). Besides efficiency, the low-recourse transformation
can be applied to an algorithm that implicitly maintains
a matching as long as it supports the following vertex-match query, relaxing the requirement of explicitly maintaining the matching.
\begin{definition}[Vertex-Match Query]
A dynamic matching algorithm is said to support the vertex-match query in query time $T$ if given any vertex query $v$, it answers in $O(T)$ time either $v$ is unmatched in the maintained matching or the matched vertex of $v$; and it can output all the edges in the maintained matching $M$ in $O(|M|\cdot T)$ time.
\end{definition}

Formally, denote $G_0$ as the initial graph and $G_i$ as the graph after the $i$-th update. The recourse of a dynamic matching algorithm $\A$
measures the changes in the support set of the matching maintained by $\A$,
which is defined as follows.

\begin{definition}[Worst-Case Recourse of a Dynamic Matching Algorithm]\label{definition:recourse:worst-case}
    For a fixed dynamic matching algorithm $\A$ that
    (possibly implicitly) maintains a matching $M_i$ on graph $G_i$,
    the worst-case recourse of $\A$ on $G_0,G_1,\dots,G_k$ is defined as $\max_{i\in[k]}|M_{i}\oplus M_{i-1}|$,
    i.e., the maximum changes in the matching edge set.
\end{definition}
\begin{definition}[Amortized Recourse of a Dynamic Matching Algorithm]\label{definition:recourse}
    For a fixed dynamic matching algorithm $\A$ that
    (possibly implicitly) maintains a matching $M_i$ on graph $G_i$,
    the amortized recourse of $\A$ on $G_0,G_1,\dots,G_k$ is defined as $\frac{1}{k}\sum_{i\in[k]}|M_{i}\oplus M_{i-1}|$,
    i.e., the average changes in the matching edge set.
\end{definition}

We start with designing a transformation between two matchings $\widetilde M_i$ on $G_i$ and $M_j$ on $G_j$ where $i<j$ that builds a large matching on $G_j$ based on $\widetilde M_i$ and $M_j$ with small $|\widetilde M_i\oplus M_j|$. We will use $\Delta$-additive-approximate to represent an additive approximation.
Formally, a matching $M$ on $G$ is \emph{$\Delta$-additive-approximate} if $w(M) \geq \mu_w(G)-\Delta$.
\begin{algorithm2e}[!ht]
  \caption{Direct Transformation between Two Time Points} \label{alg:low-recourse:two time points}
  
  \SetEndCharOfAlgoLine{}

  \SetKwInput{KwData}{Input}
  \SetKwInput{KwResult}{Output}
  \SetKwInOut{State}{global}
  \SetKwProg{KwProc}{function}{}{}

  \KwData{Two matchings $\widetilde M_i$ on $G_i$ and $M_j$ on $G_j$.}
  Consider $P=\widetilde M_i\oplus M_j$.\;
  Denote $U$ as the set of updated edges between time $i+1$ and time $j$.\;
  $D\gets\emptyset$.\;
  \For{any edge $e$ in $U\cap P$ and each of its direction}{
    \If{there exists at least $2/\eps$ edges in that direction in $P$}{
        Denote $E^\prime$ as the closest $2/\eps$ edges.\;
        \If{$E^\prime$ doesn't contain any edge in $U$}{
            Add the edge with minimum weight in $E^\prime\cap M_j$ into $D$.\label{line:two time points:deletion}\;
        }
    }
  }
  For those paths and cycles in $P\setminus D$ that contains edges in $U$, $\widetilde M_j$ picks edges in $M_j$.\;
  For the remaining ones, $\widetilde M_j$ pick edges in $\widetilde M_i$\label{line:two time points:substitution}.\;
  \texttt{return} $\widetilde M_j$.\;
\end{algorithm2e}

\begin{claim}\label{claim:two-time-points}
    Suppose $\widetilde M_i$ is a $\Delta_i$-additive-approximate MWM on $G_i$
    and $M_j$ is a $\Delta_j$-additive-approximate MWM on $G_j$.
    Then, \cref{alg:low-recourse:two time points} outputs a $\left(\Delta_i+\Delta_j+\eps\cdot\mu_w(G_j)\right)$-additive-approximate MWM $\widetilde M_j$
    on $G_j$ such that $|\widetilde M_i\oplus\widetilde M_j| \leq O((j-i)\cdot \eps^{-1})$.
\end{claim}

\begin{proof}
    Starting with $M_j$, the additive approximation is $\Delta_j$.
    It increases by $\eps\cdot\mu_w(G_j)$ during the deletion
    in line \ref{line:two time points:deletion} and $\Delta_i$ during
    the substitution in line \ref{line:two time points:substitution}.
    Since any edge in $\widetilde M_i\oplus\widetilde M_j$ belongs to a path
    or cycle in $P\setminus D$ that contains an edge in $U$, by construction,
    its size is bounded by $O(1/\eps)\cdot |U|=O((j-i)\cdot \eps^{-1})$.
\end{proof}

Now we are ready to introduce the full transformation. By running an independent copy of $\A$ on the unweighted version of $G$ we assume that we have access to a $(1-\eps)$-approximation $\nu_i$ to the size of the MCM in $G_i$.
The full transformation works in multiple phases, where each phase spans a contiguous segment of time points.
Suppose that a phase starts at time $t$.
The transformation reads the entire edge set of the weighted matching $M_{t}$
maintained by $\A$ on $G_{t}$ and sets the length of the phase to be
$\eps\cdot \nu_{t}$.
Then, we define several checkpoints $t_i$ within this phase, where $t_0$ is set to $t$ and the remaining checkpoints are defined iteratively as $t_{i+1} \defeq t_i + \frac{\eps \cdot w(M_{t_i})}{W}$.
The transformation will compute a matching $\widetilde{M}_{t_i}$ for each checkpoint, and this matching will be used as the output from this time point until the next checkpoint.
That is, for each time point $t_i < j < t_{i+1}$, the matching output by the transformation on $G_j$ will simply be $\widetilde{M}_{t_i} \cap G_j$, which is $\widetilde{M}_{t_i}$ with edges deleted in $G_j$ dropped.

\begin{lemma}\label{lemma:sufficiency of checkpoints}
    Suppose on $G_{t_i}$, $\widetilde M_{t_i}$ is a $(1-x\cdot \eps)$-approximate MWM.
    Then, for any $j$ such that $t_i<j<t_{i+1}$, $\widetilde M_{t_i}\cap G_j$ is a
    $(1-(x+2)\cdot \eps)$-approximate MWM on $G_{j}$.
\end{lemma}

\begin{proof}
    Denote $U$ as the edge updates between $t_i$ and $t_{i+1}$ then
    \[w(U)\defeq\sum_{e\in U}w(e)\leq \frac{\eps\cdot w(M_{t_i})}{W}\cdot W
    \leq \eps\cdot\mu_w(G_{t_i}).\]
    The weight of the matching is at least
    \[w(\widetilde M_{t_i}\cap G_j)\geq w(\widetilde M_{t_i})-w(U)
    \geq (1-(x+1)\cdot \eps)\cdot \mu_w(G_{t_i}),\]
    while the MWM in the graph has weight at most
    \[\mu_w(G_j)\leq \mu_w(G_{t_i})+w(U)\leq (1+\eps)\cdot \mu_w(G_{t_i}).\]
    Thus
    \[w(\widetilde M_{t_i}\cap G_j)\geq \frac{1-(x+1)\cdot\eps}{1+\eps}\cdot \mu_w(G_j)
    \geq (1-(x+2)\cdot \eps)\cdot \mu_w(G_j).\]
\end{proof}
\cref{lemma:sufficiency of checkpoints} shows that it suffices to maintain
a good approximate MWM $\widetilde{M}_{t_i}$ at each checkpoint. Below we show that the number of
checkpoints within each phase is bounded by $O(W)$.
\begin{lemma}\label{lemma:number of checkpoints}
    Each phase has at most $\frac W{(1-\eps)^2}$ checkpoints.
\end{lemma}

\begin{proof}
    For a phase starting at time $t_0$, its length is $\eps\cdot \nu_{t_0}$.
    The gap between any checkpoints $t_i$ and $t_{i+1}$ is at least
    \[\frac{\eps\cdot w(M_{t_i})}{W}\geq \frac{\eps(1-\eps)\cdot \mu_w(G_{t_i})}{W}
    \geq \frac{\eps(1-\eps)\cdot \mu(G_{t_i})}{W}
    \geq \frac{\eps(1-\eps)\cdot (\mu(G_{t_0})-\eps\nu_{t_0})}{W}
    \geq \frac{\eps(1-\eps)^2\nu_{t_0}}{W},\]
    thus the number of checkpoints is at most $\frac{W}{(1-\eps)^2}$.
\end{proof}

\cref{alg:low-recourse:two time points} provides a direct transformation between
any two checkpoints. The full transformation will use~\cref{alg:low-recourse:two time points}
as a subroutine. The initial idea is to link the checkpoints in a path-like way,
i.e., the matching maintained by the full transformation $\widetilde M_{t_i}$ on $G_{t_i}$
is the output of~\cref{alg:low-recourse:two time points} on $\widetilde M_{t_{i-1}}$
and $M_{t_i}$, where $M_{t_i}$ is the matching maintained by $\A$ on $G_{t_i}$.
Two issues arise. The first issue is that the guarantee of~\cref{claim:two-time-points}
is an additive approximation, thus $\mu_w(G_{t_{i-1}})$ should not be too much larger
than $\mu_w(G_{t_i})$. The second issue is that the path length is $O(W)$.
Since after one direct transformation, the approximation error accumulates,
suppose we choose a $(1-\delta)$-approximate MWM algorithm, the final approximation
error could reach $O(W\cdot \delta)$.

We solve the above issues in the following way. The first issue can be fixed by only
allowing $\widetilde M_{t_i}$ to be transformed by some checkpoint $t_j$ with
$\mu_w(G_{t_j})\leq 2\cdot \mu_w(G_{t_i})$. The second issue is fixed by linking
the checkpoints in a tree-like way instead of a path-like. Formally, we define
the transformation tree as follows.
\begin{definition}[Transformation Tree]
    The transformation tree is a rooted tree where the nodes represent distinct checkpoints
    and could have \textbf{ordered} children. The degree of a transformation tree
    is the maximum number of children of any node. The depth of a node in the
    transformation tree is the number of edges in the path between the root and that node.
    The depth of the transformation tree is the largest depth of its node.
    The mapping between the checkpoints and the nodes will ensure that the
    preorder traversal of the transformation tree corresponds to a contiguous subarray
    of the checkpoints, i.e., $t_i,t_{i+1},\cdots,t_j$. Further, it ensures that
    for any pair of nodes $t_i,t_j$ such that $t_j$ is an ancestor of $t_i$ in the
    transformation tree, $\mu_w(G_{t_j})\leq O(1)\cdot \mu_w(G_{t_i})$.
\end{definition}

\begin{lemma}\label{lemma:transformation tree guarantee}
    Given a transformation tree with depth $d$ and degree $c$ that corresponds to
    the checkpoints $t_i,t_{t+1},\dots,t_j$ and a dynamic $(1-\eps)$-approximate
    MWM algorithm $\A$ with initialization time $\I(n,m,W,\eps)$, 
    update time $\U(n,m,W,\eps)$ and query time $T(n,m,W,\eps)$,
    there is an algorithm that dynamically and explicitly maintains a
    $(1-O(d\cdot\eps))$-approximate MWM on $G_{t_i},G_{t_{i+1}},\dots,G_{t_j}$
    with initialization time
    \[\I(n,m,W,\eps)+O(\nu_{t_i})\cdot T(n,m,W,\eps),\]
    amortized update time
    \[\U(n,m,W,\eps)+O(c\cdot d\cdot \eps^{-1})\cdot T(n,m,W,\eps),\]
    and ensures that
    \[\frac{1}{t_j-t_i}\sum_{k=i+1}^j|\widetilde M_{t_k}\oplus\widetilde M_{t_{k-1}}|=O(c\cdot d\cdot \eps^{-1}),\]
    where $\widetilde M_{t_k}$ is the matching output by the framework on $G_{t_k}$. The transformation is partially dynamic preserving.
\end{lemma}

\begin{proof}

  We first describe the transformation.
  For the root $t_i$, we set $\widetilde M_{t_i}=M_{t_i}$. For any checkpoint $t_k>t_i$,
  it has a parent node $t_p<t_k$ in the transformation tree.
  We run~\cref{alg:low-recourse:two time points} on $\widetilde M_{t_p}$
  and $M_{t_k}$ to get $\widetilde M_{t_k}$.
  We now establish the guarantee of the transformation.
    
\paragraph{Approximation Error}
\cref{claim:two-time-points} shows that for any checkpoint
$t_k>t_i$ and its parent node $t_p$, the additional approximation error of
$\widetilde M_{t_k}$ increases by $O(\eps)\cdot \mu_w(G_{t_k})$ compared to that of
$\widetilde M_{t_p}$, since $M_{t_k}$ is a $(1-\eps)$-approximate MWM on $G_{t_k}$.
Denote $A_k$ as the set of ancestors of $t_k$,
then the additive approximation error of $\widetilde M_{t_k}$
is at most $O(\eps)\cdot \sum_{t_l\in A_k\cup\{t_k\}}
\mu_w(G_{t_l})=O(d\cdot\eps)\cdot \mu_w(G_{t_k})$,
since the definition of a transformation tree ensures that
for any $t_l\in A_k$, $\mu_w(G_{t_l})\leq O(1)\cdot \mu_w(G_{t_k})$.

\paragraph{Runtime and Recourse}
The additional initialization time is the cost of reading
the edge set of $M_{t_i}$. For the update time and recourse,
consider a fixed checkpoint $t_k$ and its parent $t_p$.
Using the vertex-match query of $\A$, we can find all edges in
$\widetilde M_{t_k}\oplus \widetilde M_{t_p}$ in time
$O((t_k-t_p)\cdot \eps^{-1})\cdot T(n,m,W,\eps)$
by~\cref{claim:two-time-points}. In other words, the cost of a
direct transformation from $t_p$ to $t_k$ can be amortized by
all updates between $t_p$ and $t_k$ and the amortized additional
update time is $O(\eps^{-1})\cdot T(n,m,W,\eps)$ while the amortized recourse
is $O(\eps^{-1})$. It suffices to show that for any fixed update $t$,
there will be at most $O(c\cdot d)$ direct transformation that covers it, i.e.,
the number of pairs $t_k$ and its parent node $t_p$ such that $t_p\leq t\leq t_k$
is at most $O(c\cdot d)$. The definition of the transformation tree ensures that
its subtree also corresponds to a contiguous subarray of checkpoints.
Consider a fixed depth of nodes in the transformation tree. The subtrees with those
nodes as root correspond to disjoint contiguous subarray. Thus $t$ could be included
in at most one of them, i.e., the number of distinct $t_p$ is at most $O(d)$.
Since the degree of each node is $c$, we conclude the proof.
\end{proof}

Below we show an online construction of $O(\log W)$ transformation trees
corresponding to disjoint contiguous subarrays whose union covers the entire phase.

\begin{algorithm2e}[!ht]
  \caption{Online Construction of Transformation Trees within a Phase} \label{alg:low-recourse:transformation trees}
  
  \SetEndCharOfAlgoLine{}

  \SetKwInput{KwData}{Input}
  \SetKwInput{KwResult}{Output}
  \SetKwInOut{State}{global}
  \SetKwProg{KwProc}{function}{}{}
  \SetKwFunction{Initialize}{Initialize}
  \SetKwFunction{Build}{Build}

  \KwData{A set of checkpoints $\{t_0,t_1,\dots,t_l\}$ where $l=O(W)$.}
  Denote $M_{t_i}$ as the matching maintained by $\A$ on $G_{t_i}$.\;
  $\theta \gets \lceil\log \frac{W}{(1-\eps)^2}\rceil=O(\log W)$.\;
  $\mathtt{root}\gets t_0$.\;
  $\mathtt{root.complete}\gets 0$.\;
  $\mathtt{root.depth}\gets 0$.\;
  $\mathtt{cur}\gets t_0$.\;
  \For{$i=1,\dots,l$}{
    \While{$\lfloor \log w(M_{\mathtt{cur}})\rfloor>\lfloor \log w(M_{t_{i}})\rfloor$ and $\mathtt{cur}\neq \mathtt{root}$}{
      $\mathtt{cur}\gets \mathtt{cur.father}$.\;
    }
    \eIf{$\lfloor \log w(M_{\mathtt{cur}})\rfloor>\lfloor \log w(M_{t_{i}})\rfloor$}{
      $\mathtt{root} \gets t_i$.\label{line:new root}\;
      $t_i\mathtt{.depth}\gets 0$.\;
    }{
    $t_i\mathtt{.father}\gets \mathtt{cur}$.\;
    $t_i\mathtt{.depth}\gets \mathtt{cur.depth} + 1$.\;
    }
    \eIf{$t_i\mathtt{.depth}=\theta$}{
      \While{$\mathtt{cur.complete} = 1$}{
        $\mathtt{cur}\gets \mathtt{cur.father}$.\;
      }
      $\mathtt{cur.complete}=1$.\;
    }{
      $\mathtt{cur}\gets t_i$.\;
      $\mathtt{cur.complete}=0$.\;
    }
  }
\end{algorithm2e}

\begin{lemma}\label{lemma:transformation tree construction}
\cref{alg:low-recourse:transformation trees} constructs $O(\log W)$ transformation trees
with depth $O(\log W)$ and degree $O(\log W)$ in amortized $O(1)$ time that
corresponds to contiguous subarrays that are disjoint and their union covers
the entire phase.
\end{lemma}

\begin{proof}
    It is clear that~\cref{alg:low-recourse:transformation trees} constructed
    a set of transformation trees in amortized $O(1)$ time with depth $O(\log W)$
    that are disjoint. Since there are $O(W)$ checkpoints within a phase according to
    \cref{lemma:number of checkpoints} and the choice of $\theta$,
    those transformation tree satisfy the covering property. We will prove that
    there are $O(\log W)$ of them, each with degree $O(\log W)$.

    \cref{alg:low-recourse:transformation trees} creates a new transformation tree
    whenever the root changes, i.e., at line~\ref{line:new root}.
    Thus the new root satisfies that
    $\lfloor\log w(M_{\mathtt{new\_root}})\rfloor<\lfloor\log w(M_{\mathtt{old\_root}})\rfloor$.
    The length of a phase starting with $t_0$ is set to be $\eps\cdot \nu_{t_0}$.
    Thus for any checkpoint $t_i$ within the phase,
    \[(1-\eps)\cdot\nu_{t_0}\leq w(M_{t_i})\leq \frac{1+\eps}{1-\eps}\cdot W\cdot \nu_{t_0},\]
    meaning that the number of different $\lfloor\log w(M_{t_{\cdot}})\rfloor$
    is at most $O(\log W)$.

    For each node $v$ in the constructed tree, $v\mathtt{.complete}$ represents whether
    there is a subtree with one of its children as root that is a complete binary tree
    with depth $\theta-v\mathtt{.depth}$. $v\mathtt{.complete}$ can only be $0$ or $1$
    during the execution. And whenever it changes to $1$, $v$ would continue building
    its subtree with a new child node. Therefore, with the same reason as
    the number of roots, $v$ would have $O(\log W)$ children when $v\mathtt{.complete}=0$ and
    $O(\log W)$ children when $v\mathtt{.complete}=1$, proving that the degree is $O(\log W)$.
\end{proof}

\begin{theorem}\label{thm:our low-recourse}
    Given a dynamic $(1-\eps)$-approximate MWM algorithm $\A$ with initialization time
    $\I(n,m,W,\eps)$, update time $\U(n,m,W,\eps)$ and query time $T(n,m,W,\eps)$, there is a transformation that produces a dynamic algorithm that explicitly maintains a $(1-O(\eps\cdot\log W))$-approximate MWM with initialization time
    \[\I(n,m,W,\eps),\]
    amortized update time
    \[\U(n,m,W,\eps)+O(\log^2 W\cdot\eps^{-1})\cdot T(n,m,W,\eps),\]
    and amortized recourse
    \[O(\log^2 W\cdot \eps^{-1}).\]
    The transformation is partially dynamic preserving.
\end{theorem}

\begin{proof}
    The only initialization time is for $\A$. We would prove the amortized update time
    and recourse within each phase. According to~\cref{lemma:transformation tree construction},
    ~\cref{alg:low-recourse:transformation trees} builds $O(\log W)$ transformation trees with
    $O(\log W)$ depth and degree that correspond to disjoint contiguous subarrays of
    checkpoints and their union covers the entire phase.
    ~\cref{lemma:transformation tree guarantee} shows that for each transformation tree,
    we spend $O(\nu_{t_0})\cdot T(n,m,W,\eps)$ time and $O(\nu_{t_0})$ recourse for
    initialization of the transformation. The initialization for $O(\log W)$ transformation
    trees can be amortized over the entire phase to be
    $O(\log W\cdot\eps^{-1})\cdot T(n,m,W,\eps)$ amortized update time and
    $O(\log W\cdot\eps^{-1})$ recourse.
    Thus the bottleneck is the $O(\log^2 W\cdot \eps^{-1})\cdot T(n,m,W,\eps)$
    amortized update time and $O(\log^2 W\cdot\eps^{-1})$ recourse that each transformation
    tree induces by the checkpoints other than root. 
\end{proof}

Consequently, our framework has a $\poly(1/\eps)$ additive overhead independent of the underlying algorithm. For simplicity, in the further use of this result in this work, we only consider algorithms that explicitly maintain the matching with $T(n,m,W,\eps)=O(1)$.

\begin{corollary}\label{corollary:partial reduction:low-recourse}
  Given a dynamic $(1-\eps)$-approximate MWM algorithm $\A$ that,
  on input $n$-vertex $m$-edge graph with aspect ratio $W$,
  has initialization time $\I(n, m, W, \varepsilon)$, update time $\U(n, m, W, \varepsilon)$ and query time $T(n,m,W,\eps)$,
  there is a transformation that produces a dynamic $(1-O(\eps\log \eps^{-1}))$-approximate MWM algorithm that has initialization time
  \[O(\I(n, m, \Theta(\eps^{-2-3\cdot 2^{-d}}), \Theta(\eps))+m\eps^{-1}),\]
  amortized update time
  \[O(\U(n, m, \Theta(\eps^{-2-3\cdot 2^{-d}}), \Theta(\eps))+\eps^{-2-d}\log^2\eps^{-1}\cdot T(n,m,\Theta(\eps^{-2-3\cdot 2^{-d}}),\eps)),\]
  and amortized recourse
  \[O(\eps^{-2-d}\log^2\eps^{-1})\]
  for any integer parameter $d\in\Z_{\geq 0}$.
  The transformation is partially dynamic preserving.
\end{corollary}

\begin{proof}
    It follows from combining \cref{thm:partial reduction:2} and \cref{thm:our low-recourse}.
\end{proof}

\subsection{Putting Everything Together}\label{sec:putting everything together}

With the help of our low-recourse transformation, we obtain an efficient fully dynamic algorithm on low-degree graphs, which then leads to a $\widetilde O(\poly(1/\eps))$ update time rounding algorithm for weighted fractional matchings. The low-degree algorithm can also serve as an efficient aggregation even if we reduce the aspect ratio to $O(\eps^{-2})$ in our framework, which is the best we can hope for based on~\cref{lemma:matching composition lemma}. Finally, combined with~\cite{BernsteinDL21} we reduce weighted matching algorithms to unweighted ones in bipartite graphs.
                               
\subsubsection{Fully Dynamic Algorithm on Low-Degree Graphs}            
\begin{theorem}
  Given an $n$-vertex $m$-edge graph $G$ with edge weights
  bounded by $W$ that undergoes edge insertions and deletions
  such that the maximum degree of $G$ is bounded by $\Delta$, there is an algorithm with
  $O(\Delta \eps^{-5}\log^2\eps^{-1})$ amortized update time and
  $O(\eps^{-5}\log^2\eps^{-1})$ amortized recourse
  that explicitly maintains a $(1-O(\eps\log\eps^{-1}))$-approximate MWM in $G$.
  \label{thm:low-degree}
\end{theorem}

\begin{proof}
  For a subgraph of $G$, we can use the following standard
  algorithm~\cite{GuptaP13} to maintain a $(1-\eps)$-approximate
  MWM on it.

  \begin{fact}
  There is a fully dynamic $(1-\eps)$-approximate MWM algorithm
  that has amortized update time $O(\Delta W \eps^{-2}\log(\eps^{-1}))$
  and recourse $O(\Delta W \eps^{-1})$ on a graph
  with maximum degree $\Delta$ and aspect ratio $W$.
  \label{fact:low-degree-mwm}
  \end{fact}

  The theorem follows by applying \cref{corollary:partial reduction:low-recourse}
  with $d=3$.
\end{proof}

\subsubsection{Rounding Weighted Fractional Matching}\label{sec:rounding}

We obtain a weighted rounding algorithm with polynomial dependence
on $\eps^{-1}$,
showing that the dynamic fractional matching problem is as hard as the integral
one up to $\poly(1/\eps)$ factors. Formally, a dynamic weighted rounding algorithm is defined as follows. 

\begin{definition}[{see e.g.,~\cite[Definition 3.11]{ChenST23}}]
  A \emph{dynamic rounding algorithm}, for a given $n$-vertex graph $G = (V, E)$, edge weights $\bw \in \N^{E}$ bounded by $W = \poly(n)$, and accuracy parameter $\eps > 0$, initializes with an $\bx \in \M_G$ and must maintain an integral matching $M \subseteq \supp(\bx)$ with $\bw(M) \geq (1-\eps)\bw^\top \bx$ under entry updates to $\bx$ that guarantee $\bx \in \M_G$ after each operation.
  \label{def:rounding}
\end{definition}

We are going to use the given fractional matching to find a sparse subgraph on which there is a large weight matching.

\begin{lemma}[\cite{ChenST23}]
  Given an $\eps > 0$, there is a deterministic algorithm that, on an $m$-edge graph $G$ with edge weights bounded by $W$ and a fractional matching $\bx \in \M_G$, initializes in $\widetilde{O}(m)$ time, supports
  \begin{itemize}
    \item inserting/deleting and edge or changing the value of $\bx_e$ in amortized $\widetilde{O}(W \cdot \eps^{-1})$ time,
  \end{itemize}
  and maintains
  \begin{itemize}
    \item a subgraph $H \subseteq G$ with maximum degree $\widetilde{O}(\eps^{-2})$ on which a fractional matching $\bx^{(H)}$ of weight $\sum_{e \in E}w(e)\bx^{(H)}_e \geq (1-\eps)\sum_{e \in E}w(e)\bx_e$ that satisfies $\bx^{(H)}(v) \leq \bx(v) + O(\eps)$ for all $v \in V$ and $\left|\bx^{(H)}_e - \bx_e\right| \leq O(\eps^2)$ for all $e \in E$.
  \end{itemize}
  \label{lemma:unweighted-sparsifier}
\end{lemma}

\begin{lemma}
  Given a fractional matching $\bx \in \M_G$, for any $\eps > 0$ we can initialize in $\widetilde{O}(m)$ a subgraph $H$ of $G$ and maintain it with $\widetilde{O}(\eps^{-1})$ time per entry update to $\bx$ such that $H$ has maximum degree $\widetilde{O}(\eps^{-2})$ and $\mu_w(H) \geq (1-O(\eps))\bw^\top\bx$ holds.
  \label{lemma:sparsify-to-low-degree}
\end{lemma}

\begin{proof}
  We split the edges into $K = O(\log W)$ classes $E_0 \cup \cdots \cup E_{K}$ such that $E_i$ contain precisely the edges with weights between $2^i$ and $2^{i+1} - 1$ (inclusively).
  Let $G_i$ be the induced subgraph of $E_i$.
  Let $\eps^\prime = \Theta(\eps/K)$.
  We run \cref{lemma:unweighted-sparsifier} on each of the $G_i$ with accuracy $\eps^\prime$ and obtain $H_i$.
  Let $H = H_0 \cup \cdots \cup H_{O(\log W)}$.
  By \cref{lemma:unweighted-sparsifier}, there is a fractional matching $\bx^{(H_i)}$ of weight
  \[ \sum_{e \in E_i}\bw_e\bx^{(H_i)}_e \geq (1-\eps^\prime)\sum_{e \in E_i}\bw_e \bx_e\]
  and therefore letting $\bx^{(H)} \defeq \bx^{(H_1)} + \cdots + \bx^{(H_K)}$ we have $\bw^\top \bx^{(H)} \geq (1-\eps^\prime) \bw^\top \bx$.
  Notice that $\bx^{(H)}(v) \leq \bx(v) + O(\eps^\prime K) \leq 1 + O(\eps)$ for each $v \in V$ and $\bx^{(H)}[B] \leq \bx[B] + O((\eps^\prime)^2 \cdot K) \cdot |B|^2 \leq \bx[B] + O(\eps)|B|$ for all odd sets of size at most $O(1/\eps)$.
  Therefore, we see that if we scale $\bx^{(H)}$ down by a multiplicative $1+O(\eps)$ factor then it will be a feasible fractional matching supported on $H$.
  This proves that $\mu_w(H) \geq (1-\eps)\bw^\top \bx$.
  The update time of the algorithm follows from that of \cref{lemma:unweighted-sparsifier} which is $\widetilde{O}({\eps^\prime}^{-1}) = \widetilde{O}(\eps^{-1})$ since the edge weights in a single $E_i$ are within a factor of two from each other.
\end{proof}

\begin{theorem}
  Given an $m$-edge graph, there is a dynamic rounding algorithm that
  initializes in $\widetilde{O}(m)$ time and handles each entry update
  to $\bx$ in $\widetilde{O}(\eps^{-8})$ time per update.
  \label{thm:weighted-rounding}
\end{theorem}

\begin{proof}
  Given the fractional matching $\bx$,
  we run \cref{lemma:sparsify-to-low-degree} to maintain a
  $\widetilde{O}(\eps^{-2})$-degree subgraph
  $H \subseteq G$ with $\mu_w(G) \geq (1-O(\eps))\bw^\top \bx$.
  We then apply \cref{thm:low-degree} to maintain a
  $(1-O(\eps\log\eps^{-1}))$-approximate matching over $H$.
  By the guarantee of $H$, such a matching will have weight at
  least $(1-O(\eps\log\eps^{-1}))\bw^\top \bx$.
  Since $H$ has maximum degree $\widetilde{O}(\eps^{-2})$,
  \cref{thm:low-degree} handles each update to $H$
  in $\widetilde{O}(\eps^{-7})$.
  The theorem follows as there are
  $\widetilde{O}(\eps^{-1})$ modifications to $H$ per update
  to $\bx$ by \cref{lemma:sparsify-to-low-degree}.
\end{proof}

\subsubsection{Improved Weight Reduction Framework for General Graphs}
The fully dynamic low-degree algorithm in \cref{thm:low-degree} allows us to reduce the aspect ratio to $O(\eps^{-2})$, which is the best we can get using~\cref{lemma:matching composition lemma}.

\UltimateReduction*
\begin{proof}
    According to~\cref{lemma:matching composition lemma}, we consider any $2$-wide weight partition of $G$. Denote $g=\lceil \log(\eps^{-3})\rceil+1$, and for $0<j\leq g-1$, denote $I_j=\{i:i\equiv j\pmod g\}$. The weight gap between neighboring ``padded'' weight classes in $I_j$ is $\Omega(\eps^{-1})$, and we use~\cref{alg:locally greedy} to aggregate matchings in $I_{j}$ and use the low-degree algorithm in~\cref{thm:low-degree} to maintain a $(1-\eps\log \eps^{-1})$-approximate MWM on the union. Following a similar proof of~\cref{lemma:matching composition lemma} itself, the union of the greedy matchings keeps a $1-O(\eps\log\eps^{-1})$ fraction thus so does the output matching.
    
    Any edge would be contained in $O(\log(\eps^{-1}))$ ``padded'' weight classes. For initialization, we use~\cref{thm:DP14} to compute a $(1-\eps)$-approximate MWM on the union, and it takes $m\log(\eps^{-1})\eps^{-1}$ time. The update time and recourse come from~\cref{thm:local greedy:runtime} and~\cref{thm:low-degree}.
\end{proof}

\subsubsection{From Weighted Matching to Unweighted Matching in Bipartite Graphs}\label{sec:framework:weightedtounweighted}
\cite{BernsteinDL21} provides a framework to reduce dynamic weighted matching algorithms to unweighted ones in bipartite graphs. We slightly optimize their algorithm to have better runtime. The description of the algorithm,~\cref{alg:BDL21}, and its proof are deferred to~\cref{sec:appendix:BDL}.
\begin{restatable}[\cite{BernsteinDL21,KaoLST01}]{lemma}{AnalysisBDL}\label{red:WeightedtoUnweighted}
For $\eps\leq 1/6$, given a dynamic algorithm $\A$ that, on input $n$-vertex $m$-edge \textbf{bipartite} graph, initializes in $\I(n, m, \eps)$ time and explicitly maintains an $(1-\eps)$-approximate MCM in $\U(n, m, \varepsilon)$ update time, there is a dynamic algorithm that initializes in
  \[O(\I(nW, mW, \eps)+m\log(\eps^{-1})\eps^{-1})\]
  time and explicitly maintains an $(1-\eps)$-approximate MWM on a bipartite graph with integer edge weights bounded by $W$ in
  \[O(W\cdot \U(nW, mW, \Theta(\varepsilon)))+W\log(\eps^{-1})\eps^{-2})\]
  amortized update time and has amortized recourse
  \[O(W\log(\eps^{-1})\eps^{-2}).\]
  The transformation is partially dynamic preserving. On non-bipartite graphs, the approximation ratio is $\frac{2}{3}-\eps$.
\end{restatable}

\cref{red:WeightedtoUnweighted} requires integer weights.
By a standard scaling and rounding argument, one can reduce a problem
with real weights and $W$ aspect ratio to the same problem with integer weights
and $W\eps^{-1}$ aspect ratio.
Combined~\cref{thm:partial reduction:ultimate},~\cref{red:WeightedtoUnweighted}
and the above fact, we have the following reduction from
weighted matching to unweighted matching.

\begin{restatable}{theorem}{weightedtounweighted}\label{red:WeightedtoUnweighted:new}
  Given a dynamic algorithm $\A$ that,
  on input $n$-vertex $m$-edge unweighted \textbf{bipartite} graph,
  initializes in $\I(n, m, \eps)$ time and explicitly
  maintains an $(1-\eps)$-approximate MCM
  in $\U(n, m, \eps)$ update time, there is a dynamic algorithm that initializes in
  \[\log(\eps^{-1})\cdot O(\I(\Theta(\eps^{-3})n, \Theta(\eps^{-3})m, \Theta(\eps))+m\eps^{-1})\]
  time and explicitly maintains an $(1-O(\eps\log(\eps^{-1})))$-approximate MWM in
  \[\poly(\log(\eps^{-1}))\cdot O(\U(\Theta(\eps^{-3})n, \Theta(\eps^{-3})m,\Theta(\eps))\cdot \eps^{-3}+\eps^{-6})\]
  amortized update time, and has amortized recourse \[\poly(\log(\eps^{-1}))\cdot O(\eps^{-6}).\]
  The transformation is partially dynamic preserving.
  In non-bipartite graphs, the approximation ratio changes to $\frac{2}{3}-O(\eps\log(\eps^{-1}))$ while the runtime and recourse remain the same.
\end{restatable}

The reduction improves on the work of \cite{BernsteinDL21},
whose reduction, when combined with~\cite{GuptaP13},
has an update time of $\U(\eps^{-O(1/\eps)} n, \eps^{-O(1/\eps)} m,\Theta(\eps))\cdot \eps^{-O(1/\eps)}\cdot \log W$.

\section{Applications}\label{sec:applications}

In this section, we discuss the implications of our reduction frameworks for obtaining $(1-\varepsilon)$-approximate maximum weight matching algorithms in various models. 

\subsection{The Dynamic Model}

In the dynamic setting, barring a few exceptions (for example, \cite{Gupta14,BhattacharyaKS23}), much of the focus has been designing algorithms for $(1-\eps)$-approximate MCM. Thus, a lot of the weighted matching results for \emph{bipartite graphs} follow from the reduction of \cite{BernsteinDL21}, and consequently, incur a multiplicative overhead of $\eps^{-O(1/\eps)}$. In this section, we remedy this. A summary of our results is given in \Cref{table:dynamic}. We start with bipartite graphs.

\begin{table}[H]
\centering
	\caption{Summary of Prior and Our Results on Dynamic Weighted Matching}
		\begin{tabular}{c c c c c}
		\hline
		Setting	& Prior Result & Our Result & Reduction\\
		\hline

\makecell{Fully Dynamic\\Bipartite} & \makecell{$\eps^{-O(1/\eps)}\cdot \frac{n}{2^{\Omega(\sqrt{\log n})}}$\\\cite{Liu24}+\cite{BernsteinDL21}} & \makecell{$O(\poly(\eps^{-1})\cdot \frac{n}{2^{\Omega(\sqrt{\log n})}})$\\\Cref{lem:fullydynamic:weighted:Liu24}} & \hyperref[red:WeightedtoUnweighted:new]{Thm \ref*{red:WeightedtoUnweighted:new}}\\
  \hline
\makecell{Incremental\\ Bipartite} & \makecell{$O(m\log n\log^{2}(nW/\eps)\eps^{-2})$\\\cite{BhattacharyaKS23} \\ (fractional)\\$m\cdot \eps^{-O(1/\eps)}\cdot\log W$\\\cite{BlikstadK23}+\cite{BernsteinDL21}} & \makecell{$O((n\eps^{-9}+m\eps^{-8})\cdot \poly(\log(1/\eps))$\\\Cref{lem:incremental}} & \hyperref[red:WeightedtoUnweighted:new]{Thm \ref*{red:WeightedtoUnweighted:new}}\\
\hline
\makecell{Fully Dynamic\\General} & \makecell{$\sqrt{m}\cdot \eps^{-O(1/\eps)}\cdot \log W$\\\cite{GuptaP13}} & \makecell{$O((\sqrt{m}\cdot \eps^{-4}+\eps^{-6})\cdot \poly(\log (1/\eps))$\\\Cref{lem:fullydynamic:weighted:GP13}} & \hyperref[thm:partial reduction:ultimate]{Thm \ref*{thm:partial reduction:ultimate}}\\
\hline
\makecell{Decremental\\General} & \makecell{$m\cdot \poly(\log (nW))\cdot \eps^{-O(1/\eps)}$\\\cite{ChenST23}} & \makecell{$O(m\cdot \poly(\log n,\eps^{-1}))$ \\ \Cref{lem:decremental:weighted}} & \hyperref[thm:partial reduction:ultimate]{Thm \ref*{thm:partial reduction:ultimate}}\\
\hline
\makecell{Fully Dynamic\\ Offline \\ Bipartite} & \makecell{$O(n^{0.58}\cdot \eps^{-O(1/\eps)}\cdot \log W)$  \\ \cite{Liu24}+\cite{BernsteinDL21}} & \makecell{$O(n^{0.58}\cdot \poly(\eps^{-1}))$\\ \Cref{lem:dynamic:offline}} & \makecell{\hyperref[red:WeightedtoUnweighted:new]{Thm \ref*{red:WeightedtoUnweighted:new}}}\\
\hline
\end{tabular}
\label{table:dynamic}
\end{table}

\paragraph{Bipartite Graphs} For bipartite graphs, we show the following three results in the fully dynamic and incremental setting respectively.

\begin{lemma}\label{lem:fullydynamic:weighted:Liu24}
There is a fully dynamic randomized algorithm that maintains a $(1-\varepsilon)$-approximate MWM in a bipartite graph in $O(\poly(\eps^{-1})\cdot \frac{n}{2^{\Omega(\sqrt{\log n})}})$ update time. 
\end{lemma}

\begin{lemma}\label{lem:incremental}
    There is a deterministic incremental algorithm that maintains a $(1-\eps)$-approximate MWM in an incremental bipartite graph in $O(n\eps^{-9}\poly(\log1/\eps)+m\eps^{-8}\poly(\log 1/\eps))$ total update time.
\end{lemma}

\begin{lemma}\label{lem:dynamic:offline}
    There is a randomized algorithm that given an offline sequence of edge insertions and deletions to an $n$-vertex bipartite weighted graph, maintains the edges of a $(1-\eps)$-approximate maximum weight matching in amortized $O(n^{0.58}\poly(\eps^{-1}))$ time with high probability. 
\end{lemma}

We now prove \Cref{lem:fullydynamic:weighted:Liu24}. In order to achieve this, we use the following recent result by \cite{Liu24}.

\begin{lemma}[\cite{Liu24}]\label{lem:unweighted:Liu24}
There is a fully dynamic randomized algorithm that maintains a $(1-\varepsilon)$-approximate MCM in a bipartite graph in $O(\poly(\eps^{-1})\cdot \frac{n}{2^{\Omega(\sqrt{\log n})}})$ update time. 
\end{lemma}

\begin{proof}[Proof of \Cref{lem:fullydynamic:weighted:Liu24}]
    The result follows from \Cref{lem:unweighted:Liu24} and \Cref{red:WeightedtoUnweighted:new}.
\end{proof}

Prior to \cref{lem:fullydynamic:weighted:Liu24}, the best-known result had an update time of $\eps^{-O(1/\eps)}\cdot \frac{n}{2^{-\Omega(\sqrt{\log n})}}$, which was obtained by combining the result of \cite{Liu24} and \cite{BernsteinDL21}. We now show \Cref{lem:incremental}, and for that we need the following recent result by Blikstad and Kiss.

\begin{lemma}[\cite{BlikstadK23}]\label{lem:incremental:unweighted}
 There exists a deterministic incremental algorithm that maintains a $(1-\varepsilon)$-approximate MCM in an incremental bipartite graph in $O(n\eps^{-6}+m\eps^{-5})$ total update time.
\end{lemma}

\begin{proof}[Proof of \Cref{lem:incremental}]
The result follows from the application of \Cref{red:WeightedtoUnweighted:new} to the amortized runtime given by \Cref{lem:incremental:unweighted}.
\end{proof}
Prior to \cref{lem:incremental}, the best known incremental algorithm has an update time of $O(m\cdot \log n\cdot \log^{2}(nW\eps^{-1})\cdot\eps^{-2})$ (see \cite{BhattacharyaKS23}) or $O(\eps^{-O(1/\eps)}\cdot m\cdot \log W)$ (by combining \cite{BlikstadK23} and \cite{GuptaP13}).
We now prove \Cref{lem:dynamic:offline}. In order to achieve this, we use the following recent result by \cite{Liu24}.

\begin{lemma}[\cite{Liu24}]
    There is a randomized algorithm that given an offline sequence of edge insertions and deletions to an $n$-vertex bipartite graph, maintains the edges of $(1-\eps)$-approximate matching in amortized $O(n^{0.58}\poly(\eps^{-1}))$ time with high probability. 
\end{lemma}
\begin{proof}
    The result follows from \Cref{red:WeightedtoUnweighted:new} and \Cref{lem:dynamic:offline}.
\end{proof}

We now show our results for general graphs.

\paragraph{General Graphs} For general graphs, we show two results in the fully dynamic and decremental settings, respectively.

\begin{lemma}\label{lem:fullydynamic:weighted:GP13}
There is a deterministic fully dynamic algorithm that maintains a $(1-\eps)$-approximate MWM in $O((\sqrt{m}\cdot \eps^{-4}+\eps^{-6})\cdot \poly(\log (1/\eps))$ update time.
\end{lemma}

Note that prior to this work, the best-known update time for an algorithm that maintains a $(1-\eps)$-approximate maximum weight matching was $O(\sqrt{m}\cdot \eps^{-O(1/\eps)}\cdot \log W)$. This was obtained by combining the results of \cite{GuptaP13} with their bucketing scheme. 

\begin{lemma}\label{lem:decremental:weighted}
There is a randomized decremental algorithm that maintains a $(1-\eps)$-approximate MWM in a decremental general graph in $O(m\cdot \poly(\log n, \eps^{-1}))$ total update time.
\end{lemma}

Prior to this work, the best known decremental algorithm maintaining a $(1-\eps)$-approximate maximum weight matching in a general graph had $O(\eps^{-O(1/\eps)}\cdot\poly(\log n))$-update time, and was obtained by combining the results of \cite{ChenST23} with the bucketing scheme of \cite{GuptaP13}. We now proceed with showing the proof of \Cref{lem:fullydynamic:weighted:GP13}. For this, we use the following result.

\begin{lemma}[\cite{GuptaP13}]\label{lem:fullydynamic:weighted:GP13:original}
There is a deterministic fully dynamic algorithm that maintains a $(1-\eps)$-approximate MWM in $O(\sqrt{m} W\eps^{-2})$ update time.
\end{lemma}

\begin{proof}[Proof of \Cref{lem:fullydynamic:weighted:GP13}]
The result follows from \Cref{lem:fullydynamic:weighted:GP13:original} and \Cref{thm:partial reduction:ultimate}.
\end{proof}

We now show the proof of \Cref{lem:decremental:weighted}. Our proof uses the following recent result by \cite{ChenST23}.

\begin{lemma}[\cite{ChenST23}]\label{lem:decremental:weighted:CST23}
There is a randomized decremental algorithm that maintains a $(1-\varepsilon)$-approximate MWM in a decremental general graph in $O(m\cdot W \cdot \poly(\log n,\eps^{-1}))$ total update time.
\end{lemma}

\begin{proof}[Proof of \Cref{lem:decremental:weighted}]
The result follows from \Cref{lem:decremental:weighted:CST23} and \Cref{thm:partial reduction:ultimate}.
\end{proof}

\subsection{The Streaming Model}

\paragraph{Model Definition} In the streaming model, the edges of the input $n$-vertex graph $G=(V,E)$ are presented to the algorithm in a stream (in an arbitrary order). A \emph{semi-streaming algorithm} is allowed to make one or a few passes over the stream, use a limited amount of memory $O(n\poly(\log n))$, and at the end output a solution to the problem at hand, say, find an approximate maximum weight matching of $G$. 

\paragraph{Our Results} As in the dynamic case, we obtain two types of reductions. First is an aspect ratio reduction, and the second, is a weighted to unweighted reduction for bipartite graphs. 

\begin{theorem}\label{red:streaming:aspectratio}
    Suppose there is a semi-streaming algorithm for $(1-\eps)$-approximate maximum weight matching in an $n$-node $m$-edge general graph with aspect ratio $W$ that uses $p(n,m,W,\eps)$ passes and space $s(n,m,W,\eps)$, then for any constant $c>0$ there exists a semi-streaming algorithm for $(1-c^{-1}\eps)$-approximate maximum weight matching that uses $p(n,m,\Theta(\eps^{-(2+c)}),\eps)$ passes and space complexity $O(s(n,m,\Theta(\eps^{-(2+c)}),\eps)\cdot \log_{\eps^{-1}} W)$.
\end{theorem}

\begin{theorem}\label{red:streaming:weightedtounweighted}
Suppose there is a semi-streaming algorithm for $(1-\eps)$-approximate MCM in an $n$-node $m$-edge bipartite graph that uses $p(n,m,\eps)$ passes and space $s(n,m,\eps)$, then there exists a semi-streaming algorithm for $(1-O(\eps))$-approximate maximum weight matching in a $n$-node $m$-edge bipartite graph with aspect ratio $W$ that uses $p(\Theta(n\cdot \eps^{-(3+c)}),\Theta(m\cdot \eps^{-(3+c)}),\Theta(\eps))$ passes and space $O(s(\Theta(n\cdot \eps^{-(3+c)}),\Theta(m\cdot \eps^{-(3+c)}),\Theta(\eps))\cdot\log_{\eps^{-1}} W)$.
\end{theorem}

The above reduction has the property that it is a weighted to unweighted reduction that preserves the number of passes, while increasing the space complexity by a factor of $\log W$. 

As a consequence of our reductions, we new results and trade-offs for streaming $(1-\eps)$-approximate bipartite maximum weight matching, which are summarized in the \Cref{table:streaming}. We state them formally thereafter.

\begin{table}[H]
\centering
	\caption{Summary of Results on $(1-\eps)$-approximate Bipartite MWM in Streaming}
    \begin{tabular}{c c c c}
		\hline
        Prior Result & & Our Result ($\forall c>0$ constant) & Reduction Used \\
  \hline
  \makecell{$O(\eps^{-2}\cdot \log(\eps^{-1}))$ passes\\$O(n\eps^{-2}\log W)$ space}
  & \cite{AhnG13} & \makecell{$O(\eps^{-2})$ passes\\$O(n\cdot \eps^{-(3+c)}\log W)$ space\\\Cref{lem:streaming:mwm:ALT21}} & \Cref{red:streaming:weightedtounweighted}\\
  \hline
  \makecell{$O(\eps^{-7}\cdot \log^3(1/\eps))$ passes\\$O(n\cdot \log(1/\eps)\cdot \log W)$ space} & \cite{LiuKK23} &  \makecell{$O(\eps^{-4}\cdot \log^3(1/\eps))$ passes\\$O(n\cdot \log W)$ space\\
  \Cref{lem:mwm:streaming:LiuKK23old}} & \Cref{red:streaming:aspectratio}\\
  \hline
\end{tabular}
\label{table:streaming}
\end{table}

 Our first result is the following lemma. 

 \begin{lemma}\label{lem:streaming:mwm:ALT21}
    For any constant $c>0$, there is a semi-streaming algorithm for $(1-\eps)$-approximate bipartite maximum weight matching that uses $O(n\cdot\eps^{-(3+c)}\cdot \log W)$ space and has a pass complexity of $O(\eps^{-2})$. 
\end{lemma}

Prior to this, the semi-streaming algorithm of \cite{AhnG14} had the best known pass complexity of $O(\eps^{-2}\cdot \log(\eps^{-1}))$. The second result is the following.

\begin{lemma}\label{lem:mwm:streaming:LiuKK23} For any constant $c>0$ there is a semi-streaming algorithm for $(1-\eps)$-approximate bipartite maximum weight matching that uses $O(n\cdot \log \eps^{-1}\cdot \log W)$ space and has a pass complexity of $O(\eps^{-4}\cdot \log^3(1/\eps))$. 
\end{lemma}

This improves on the result of \cite{LiuKK23} which had the same space complexity, but a pass complexity of $O(\eps^{-8})$ passes.

\paragraph{Proofs of Our Streaming Results}
We first begin by stating the proofs of our reductions. We start by proving \Cref{red:streaming:weightedtounweighted}. This will be done using \Cref{red:streaming:aspectratio} and the following result by \cite{BernsteinDL21,KaoLST01}. While their result is stated as being applicable to integral weight graphs, as mentioned before, by standard scaling and rounding techniques, one can reduce the arbitrary weight case to the integral case. We state a modified version of their result incorporating this.

\begin{lemma}[\cite{BernsteinDL21,KaoLST01}]\label{red:weightedtounweighted:unfolding}
Suppose $\mathcal{A}_u$ is a streaming algorithm that computes a $(1-\eps)$-approximation to the MCM in $p(n,m,\eps)$ passes and $s(n,m,\eps)$ space. Then, there is a streaming algorithm $\mathcal{A}_w$ that computes a $(1-\eps)$-approximation to the maximum weight matching in $p(nW\eps^{-1},mW\eps^{-1},\eps )$ passes and $s(n W\eps^{-1},m W\eps^{-1},\eps)$ space, where $W$ is the aspect ratio of the weighted graph.  
\end{lemma}

\begin{proof}[Proof of \Cref{red:streaming:weightedtounweighted}]
    Suppose $\mathcal{A}_u$ is the bipartite unweighted matching algorithm in the premise of the lemma. Then, we can use \Cref{red:weightedtounweighted:unfolding} to get an algorithm $\mathcal{A}_w$ with space complexity $O(s(nW\eps^{-1},mW\eps^{-1},\eps))$ space, and $p(nW,mW,\eps)$ passes. Applying \Cref{red:streaming:aspectratio} to $\mathcal{A}_w$, get an algorithm $\mathcal{A}'_{w}$ with pass complexity $p(\Theta(n\cdot \eps^{-(3+c)}),\Theta(m\cdot \eps^{-(3+c)}),\Theta(\eps))$ and space complexity $O(s(\Theta(n\cdot \eps^{-(3+c)}),\Theta(m\cdot \eps^{-(3+c)}),\Theta(\eps))\cdot\log_{\eps^{-1}} W)$.
\end{proof}

We now show how to implement our reduction in streaming.

\begin{proof}[Proof of \Cref{red:streaming:aspectratio}]
Let $\mathcal{A}$ be the algorithm given in the premise of the lemma. As in \Cref{thm:partial reduction:ultimate}, we consider any $\eps^{-c}$-wide weight partition of $G$, and let $I_j$'s be the set of ``padded" weight classes. Then, by \Cref{lemma:matching composition lemma}, the union of matchings $M_j$ on $I_{j}$ contains a $(1-c^{-1}\eps)$-approximate maximum weight matching of $G$. We run a copy of $\mathcal{A}$ on each of these weight classes $I_j$ and then combine them. Since we run $\log_{\eps^{-1}}W$ copies of $\mathcal{A}$ and the aspect ratio of the weight classes is $\eps^{-(2+c)}$, we have the desired space and pass bound.  
\end{proof}

We now show the proof of \Cref{lem:streaming:mwm:ALT21}. Our proof uses the following result by \cite{AssadiLT21}.

\begin{lemma}[\cite{AssadiLT21}]\label{lem:mcm:streaming:ALT21}
There is a semi-streaming algorithm for $(1-\eps)$-approximate bipartite MCM that uses $O(n)$ space, and has a pass complexity of $O(\eps^{-2})$.
\end{lemma}

\begin{proof}[Proof of \Cref{lem:streaming:mwm:ALT21}]
    Instantiating the reduction of \Cref{red:streaming:weightedtounweighted} with the algorithm of \Cref{lem:mcm:streaming:ALT21}, we obtain the desired $O(n\cdot\eps^{-(3+c)}\cdot \log W)$ space complexity and has a pass complexity of $O(\eps^{-2})$
\end{proof}

For \Cref{lem:mwm:streaming:LiuKK23}, we need the following result. 

\begin{lemma}[\cite{LiuKK23}]\label{lem:mwm:streaming:LiuKK23old}
   There is a semi-streaming algorithm for $(1-\eps)$-approximate bipartite maximum weight matching that has uses $O(n \log(1/\eps))$ space and has a pass complexity of $O(\log^3 W\cdot \eps^{-4})$. By applying the reduction of \cite{GuptaP13}, one can obtain an algorithm that uses $O(n\cdot \log (\eps^{-1})\cdot \log W)$ space and has a pass complexity of $O(\eps^{-7}\log^3(1/\eps))$.
\end{lemma}

We now show \Cref{lem:mwm:streaming:LiuKK23}, which improves the pass \Cref{lem:mwm:streaming:LiuKK23old}, while still achieving a space complexity which has logarithmic dependence on $\frac{1}{\eps}$. 

\begin{proof}[Proof of \Cref{lem:mwm:streaming:LiuKK23}]
    Let $\mathcal{A}$ be the algorithm in \Cref{lem:mwm:streaming:LiuKK23old} which has a space complexity of $O(n\cdot \log(1/\eps))$ and has a pass complexity of $O(\log^3 W\cdot \eps^{-4})$. We instantiate the reduction in \Cref{red:streaming:aspectratio} with $\mathcal{A}$. This yields a semi-streaming algorithm that satisfies the premise of the corollary. 
\end{proof}

\subsection{The MPC Model}

\paragraph{Model Definition} In the MPC Model, there are $p$ machines, each
with a memory of size $s$, such that $p\cdot s=O(m)$. The computation proceeds in synchronous rounds: in each round, each machine performs some local computation and at the end of the round they exchange messages. All messages sent and received by each machine in each round have to fit into the local memory of the machine, and hence their length is bounded by s in each round. At the end, the machines collectively output the solution. In this paper, we work in the \emph{linear memory} model in which, the memory per machine is $s=\Tilde{O}(n)$. We first state our results in this model.

\paragraph{Our Results} As in the dynamic and streaming case, we give the following reductions, the first one being an aspect ratio reduction, and the second, a reduction from weighted to unweighted matching in \emph{bipartite graphs}.

\begin{theorem}\label{red:MPC:aspectratio}
    Suppose there is an MPC algorithm for $(1-\eps)$-approximate maximum weight matching in an $n$-node $m$-edge general graph with aspect ratio $W$ that uses $r(n,m,W,\eps)$ rounds and space $s(n,m,W,\eps)$ per machine, then for any constant $c>0$ there exists an MPC algorithm for $(1-O(\eps))$-maximum weight matching that uses $r(n,m,\Theta(\eps^{-(2+c)}),\eps)$ rounds and $O(s(n,m,\Theta(\eps^{-(2+c)}),\eps)\cdot \log W+n\log W)$ space per machine.
\end{theorem}

\begin{theorem}\label{red:MPC:weightedtounweighted}
Suppose there is an MPC algorithm for $(1-\eps)$-approximate MCM in an $n$-node $m$-edge bipartite graph that uses $r(n,m,\eps)$ passes and space $s(n,m,\eps)$, then there exists an MPC algorithm for $(1-O(\eps))$-approximate maximum weight matching in a $n$-node $m$-edge bipartite graph with aspect ratio $W$ that uses $r(\Theta(n\cdot \eps^{-(3+c)}),\Theta(m\cdot \eps^{-(3+c)}),\Theta(\eps))$ rounds and space 
$O(s(\Theta(n\cdot \eps^{-(3+c)}),\Theta(m\cdot \eps^{-(3+c)}),\Theta(\eps))\cdot\log W+n\log W)$ per machine.   
\end{theorem}

As a consequence of these two reductions, we get the following two results about about $(1-\eps)$-approximate bipartite maximum weight matching, which matches the round complexity of the best known MPC algorithm for unweighted matching by \cite{AssadiLT21}.

\begin{lemma}\label{lem:mpc:improved:alt}
    There is an MPC algorithm for computing a $(1-\eps)$-approximate bipartite maximum weight matching in $O(\log\log (n/\eps)\cdot \eps^{-2})$ rounds and $O(n\cdot \eps^{-(3+c)}\cdot \log_{\eps^{-1}} W)$ space per machine. 
\end{lemma}

The second lemma improves on the result of \cite{LiuKK23}.

\begin{lemma}\label{lem:mpc:improved:liukk}
There is an MPC algorithm for $(1-\eps)$-approximate bipartite maximum weight matching using $O(\log^3(1/\eps)\cdot \log \log n\cdot \eps^{-4})$ rounds and $O(n\log_{\eps^{-1}}W)$ space per machine.
\end{lemma}

We summarize these results in \Cref{table:mpc}.

\begin{table}
\centering
	\caption{Summary of Results on $(1-\eps)$-approximate Bipartite Matching in MPC Model}
		\begin{tabular}{c c c c c}
		\hline
		Rounds & Space & Weighted/Unweighted &  Reference \\  
  \hline
    $O(\eps^{-2} \log \log{n})$ & $O(n)$ & Unweighted & \cite{AssadiLT21}\\
  $O(\eps^{-8} \log \log{n})$ & $O(n \log_{\eps^{-1}}W)$ & Weighted & \cite{LiuKK23}\\ 
 $O(\log\log (n/\eps)\cdot \eps^{-2})$ & $O(n\cdot \eps^{-(3+c)}\cdot \log_{\eps^{-1}} W)$ & Weighted & \Cref{lem:mpc:improved:alt}\\
 $O(\log^3(1/\eps)\cdot \log \log(n/\eps)\cdot \eps^{-4})$ & $O(n \log_{\eps^{-1}} W)$ & Weighted & \Cref{lem:mpc:improved:liukk}\\
  \hline
\end{tabular}
	\label{table:mpc}
	\end{table}

\paragraph{Proofs in the MPC Model}
We first show the proof of our main reductions. We start with the proof of \Cref{red:MPC:weightedtounweighted}, and for that, we need the following theorem, which is implicit from the work of \cite{BernsteinDL21,KaoLST01}.

\begin{lemma}[Implicit in \cite{BernsteinDL21,KaoLST01}]\label{red:weightedtounweighted:unfolding:mpc}
Suppose $\mathcal{A}_u$ is an MPC algorithm that computes a $(1-\eps)$-approximation to the MCM in $r(n,m,\eps)$ rounds and $s(n,m,\eps)$ space per machine. Then, there is an MPC algorithm $\mathcal{A}_w$ that computes a $(1-\eps)$-approximation to the maximum weight matching in $r(nW\eps^{-1},mW\eps^{-1},\eps )$ rounds and $s(n W\eps{^{-1}},m W\eps^{-1},\eps)$ space per machine, where $W$ is the aspect ratio of the weighted graph.  
\end{lemma}

The proof of \Cref{red:MPC:weightedtounweighted} is implied by the above lemma, and \Cref{red:MPC:aspectratio}.

\begin{proof}[Proof of \Cref{red:MPC:weightedtounweighted}]
Suppose $\mathcal{A}_u$ is the bipartite unweighted matching algorithm in the premise of the lemma. Then, we can use \Cref{red:weightedtounweighted:unfolding:mpc} to get an algorithm $\mathcal{A}_w$ with $O(s(nW\eps^{-1},mW\eps^{-1},\eps))$ space per machine, and $r(nW\eps^{-1},mW\eps^{-1},\eps)$ rounds. Applying \Cref{red:MPC:aspectratio} to $\mathcal{A}_w$, get an algorithm $\mathcal{A}'_{w}$ with round complexity $r(\Theta(n\cdot \eps^{-(3+c)}),\Theta(m\cdot \eps^{-(3+c)}),\Theta(\eps))$ rounds and space $O(s(\Theta(n\cdot \eps^{-(3+c)}),\Theta(m\cdot \eps^{-(3+c)}),\Theta(\eps))\cdot\log W+n\log W)$ per machine.   
\end{proof}

We now state the proof our aspect ratio reduction in MPC.

\begin{proof}[Proof of \Cref{red:MPC:aspectratio}]
Let $\mathcal{A}$ be the algorithm given in the premise of the lemma. As in \Cref{thm:partial reduction:ultimate}, we consider any $\eps^{-c}$-wide weight partition of $G$, and let $I_j$'s be the set of ``padded" weight classes. Then, by \Cref{lemma:matching composition lemma}, then, the union of matchings $M_j$ on $I_{j}$ contains a $(1-c^{-1}\eps)$-approximate maximum weight matching of $G$. We run a copy of $\mathcal{A}$ on each of these weight classes $I_j$ and then combine them in a single matching. Since we run $\log W$ copies of $\mathcal{A}$ and the aspect ratio of the weight classes is $\eps^{-(2+c)}$, we have the desired space and pass bound.  
\end{proof}

We now show the proof of \Cref{lem:mpc:improved:alt}. In order to do that, we need the following result.

\begin{lemma}[\cite{AssadiLT21}]\label{lem:mcm:mpc:ALT21}
There is an MPC algorithm for $(1-\eps)$-approximate bipartite matching using $O(\eps^{-2}\cdot \log \log n)$ rounds and $O(n)$ space per machine. 
\end{lemma}

\begin{proof}[Proof of \Cref{lem:mpc:improved:alt}]
    Let $\mathcal{A}$ be the algorithm of \Cref{lem:mcm:mpc:ALT21}. We instantiate the reduction in \Cref{red:MPC:weightedtounweighted} with $\mathcal{A}$ to get an MPC algorithm for $(1-\eps)$-approximate bipartite maximum weight matching that has $O(\log\log (n/\eps)\cdot \eps^{-2})$ round complexity and $O(n\cdot \eps^{-(3+c)}\cdot \log_{\eps^{-1}} W)$ space per machine. 
\end{proof}

Next, we show the proof of \Cref{lem:mpc:improved:liukk}. We need the following result.

\begin{lemma}[\cite{LiuKK23}]\label{lem:mwm:mpc:LKK23}
There is an MPC algorithm for $(1-\eps)$-approximate bipartite maximum weight matching using $O(\log^3(W) \log \log n\cdot \eps^{-4})$ rounds and $O(n)$ space per machine. By applying the reduction of \cite{GuptaP13}, we can obtain an MPC algorithm for the same problem that uses $O(\log \log n \cdot \eps^{-7}\cdot \log^3(1/\eps))$ rounds and $O(n \log_{\eps^{-1}}W)$ space per machine. 
\end{lemma}

\begin{proof}[Proof of \Cref{lem:mpc:improved:liukk}]
Let $\mathcal{A}$ be the algorithm of \Cref{lem:mwm:mpc:LKK23} with space complexity $O(n)$ and round complexity $O(\log^3(W) \log\log n \cdot \eps^{-4})$. Instantiating \Cref{red:MPC:aspectratio} with $\mathcal{A}$, we get a $(1-\eps)$-approximate bipartite maximum weight matching with round complexity $O(\log^3(1/\eps)\cdot \log \log(n/\eps)\cdot \eps^{-4})$ and space $O(n\cdot \log_{\eps^{-1}}W)$ per machine. 
\end{proof}

\subsection{The Parallel Shared-Memory Work-Depth Model}

\paragraph{Model Definition} The parallel shared-memory work-depth model is a parallel model where different processors can process instructions in parallel and read and write from the same shared-memory. In this model, we care about two properties of the algorithm: \emph{work}, which is the total amount of computation done by the algorithm and the \emph{depth}, which is the longest chain of sequential dependencies in the algorithm. Our goal in this section will be to show that our reduction can be implemented in parallel model very efficiently. In particular, we show the following theorems.

\begin{theorem}\label{red:parallel:aspectratio}
    Suppose there is a parallel algorithm that computes a $(1-\eps)$-approximate maximum weight matching on an $n$-node $m$-edge graph with aspect ratio $W$ with $B(n,m,W,\eps)$ work and $D(n,m,W,\eps)$ depth. Then there exists a parallel algorithm that computes a $(1-\eps)$-approximate maximum weight matching in $O(B(n,m,\eps^{-{5}},\Theta(\eps))\cdot \log W+n\log n)$ work and $O(D(n,m,\eps^{-{5}},\Theta(\eps))+ \log W+\log^2 n)$ depth.
\end{theorem}

Since the lack of a parallel implementation of the reduction in~\cite{BernsteinDL21}, we currently cannot reduce the weighted matching problem directly to an unweighted one. Thus we use~\cref{red:parallel:aspectratio} to improve the following weighted parallel algorithm by reducing the weight ranges.

\begin{lemma}[\cite{LiuKK23}]
    There exists a shared-memory parallel algorithm that computes a $(1-\eps)$-approximate maximum weight matching with $O(m \log^3 (W)\eps^{-4})$ work and $O(\log^3(W) \log^2(n)\eps^{-4})$ depth. Using \cite{GuptaP13}, this translates into a parallel algorithm which computes a $(1-\eps)$-approximate maximum weight matching with $O(m\log (W)\log^3(1/\eps)\eps^{-7})$ work and $O(\log (W) \log^2(n)\log^3(1/\eps)\eps^{-7})$ depth. 
\end{lemma}

A consequence of our reduction is the following improvement, in both total work and depth. 

\begin{corollary}
There exists a shared-memory parallel algorithm that computes a $(1-\eps)$-approximate maximum weight matching on an $n$-node $m$-edge graph with aspect ratio $W$ with $O(m \log (W)\eps^{-4})$ work and $O(\log^2 (n) \log (W)\eps^{-4})$ depth.     
\end{corollary}

We now show how to implement our reduction in the parallel model. The most challenging aspect of this implementation is to compute a maximum weight matching on degree two graphs. Such a graph is a collection of paths and cycles. First an observation is in order. Consider two paths $P_1=(v_0,\cdots, v_{|P_1|})$ or $P_1=(e_1,\cdots, e_{|P_1|})$ and $P_2=(u_0,\cdots, u_{|P_2|})$ or $P_2=(e'_{1},\cdots, e_{|P_2|})$.
As in the dynamic program described in \Cref{lemma:dynamic path-cycle maintainer}, we maintain for $P_1$: $f(e_1,x,e_{|P_1|},y)$ for $x,y\in \{0,1\}$. Here, $f(e_1,0,e_{|P_1|},0)$is for example the value of the maximum weight matching on $P_1$ in which $e_1$ and $e_2$ are unmatched. Similarly, for $P_2$, we maintain $f(e'_1,x,e'_{|P_2|},y)$ for $x,y\in \{0,1\}$. Suppose, $P=P_1\bigoplus e\bigoplus P_2$, where $e=(v_{|P_1|},v_0)$, then, we can get the corresponding information for $P$ as follows for all $x,y\in \{0,1\}$.
\begin{align*}
  f(e_1,x,e'_{|P_2|},y)=\max_{\substack{z\in \{0,1\}\\ 0\leqslant s,t\leqslant 1-z}}\left\{ f(e_1,x,e_{|P_1|},s)+w(e)\cdot z+f(e'_1,t,e'_{|P_2|},y)\right\}
\end{align*}
Thus, using $P_1$ and $P_2$, we can get the information for $P$ using a constant amount of work. We now describe our algorithm. Similarly, using depth $1$, and work $|P|$, we can find the corresponding maximum weight matchings for $P$, given the maximum weight matchings for $P_1$ and $P_2$. Analogously, we can also give such a dynamic program for a cycle obtained by concatenating two paths.

\begin{claim}\label{claim:parallel:degtwomwm}
    There exists a parallel algorithm for computing a maximum weight matching on a degree two graph in $O(n\cdot \log n)$ work and $O(\log^2 n)$ depth.
\end{claim}
\begin{proof}
The algorithm proceeds by randomly concatenating paths. At any stage, the algorithm will maintain a collection of paths $\mathcal{P}$. Initially, $\mathcal{P}=\{P(u,u)\mid u\in V\}$. These are just empty paths corresponding to every vertex $u\in V$, and with the endpoints of the paths being $u$. As the algorithm proceeds, $\mathcal{P}$ is updated as follows. Let $\mathcal{V}$ be the collection of all endpoints of a path. Initially, $\mathcal{V}=V$. For all $u\in \mathcal{V}$, toss a coin. We can do this in parallel. Consider any edge $e=uv$ such that $u,v\in \mathcal{V}$. If the results of coin tosses of $u$ and $v$ are opposite, then we combine the paths $P_u=((u_0,u_1)=e_1,\cdots, e_{|P_u|}=(u_{|P_u|-1},u))$ and $P_v=(e'_{1}=(v,v_1),\cdots, e'_{|P_v|}=(v_{|P_v|-1},v_{|P_v|}))$ as follows:
\begin{enumerate}
    \item In the collection of paths, remove $P[u_0,u]$ and $P[v,v_{|P_v|}]$ and add $P[u_0,v_{|P_v|}]$, which is the concatenation of $P_u\oplus e\oplus P_v$. 
    \item We also update $f$ as follows, for all $x,y\in \{0,1\}$,
\begin{align*}
    f(e_1,x,e'_{|P_v|},y)=\max_{\substack{z\in \{0,1\}\\ 0\leqslant s,t\leqslant 1-z}}\left\{ f(e_1,x,e_{|P_u|},s)+w(e)\cdot z+f(e'_1,t,e'_{|P_v|},y)\right\}
\end{align*}
   \item Additionally, in $O(|P|)$ time, we can also do a search version of the above dynamic program to maintain the four candidate matchings which realize $f(e_1,x,e'_{|P_v|},y)$ for $x,y\in \{0,1\}$
\end{enumerate}
Now, want to argue about the depth and work. First, with high probability, we have $\log n$ stages. Additionally, within each stage, with high probability, we will have to contract $O(\log n)$ paths. Thus, total depth is $O(\log^2 n)$. The total work done in each stage is proportional to the total lengths of the paths in $\mathcal{P}$. Thus, the total work is $O(n\cdot \log n)$. 
\end{proof}

\begin{proof}[Proof of \Cref{red:parallel:aspectratio}]
Let $\Tilde{E_i}$ be as defined in \Cref{alg:reduction:partial}. Let $\mathcal{A}$ be the parallel algorithm specified in the premise of the theorem. We consider $G_i=(V,\Tilde{E_i})$ for $i\in [L]$, and run $\mathcal{A}_i$ on $G_i$ to compute $M_i$, which is the $(1-\eps)$-approximate maximum weight matching in $G_i$. Then, for all odd $i\in [L]$ we find a greedy census matching $M_{odd}$. We do the same for all even $i\in [L]$, to get $M_{even}$. Since the aspect ratio in $G_i$ is $\eps^{-5}$, the first step takes work $B(n,m,\eps^{-5},\eps)\cdot \log W$ work and depth $D(n,m,\eps^{-5},\eps)$. The second step can be implemented in total work $B(n,m,\eps^{-5},\eps)\cdot \log W$. The depth for the second step is $\log W$, since we are only greedily combining $\log W$ matchings. 
Finally, we want to compute a maximum matching on the graph $M_{even}\cup M_{odd}$. Thus, we can find a matching in this graph by applying \Cref{claim:parallel:degtwomwm}. Thus, we are able to compute a compute a $(1-\eps)$-approximate maximum weight matching in $O(B(n,m,\eps^{-5},\eps)\cdot \log W+n\log n)$ work and $O(D(n,m,\eps^{-5},\eps)+\log W +\log^2 n)$ depth. 
\end{proof}

\section{Open Problems}
\label{sec:intro-open-problems}
Our reductions go a long way toward showing weighted and unweighted matching have the same complexity in a wide variety of models. By reducing the multiplicative overhead to $\poly(1/\eps)$, we are able to achieve this equivalence even for small approximation parameter $\eps$. There are, however, a few limitations that we need to overcome.
\begin{enumerate}
\item A limitation of \emph{all} existing reductions from weighted to unweighted matching in dynamic graphs is that they incur a large approximation error in non-bipartite graphs: in particular, both our \cref{thm:main-unweighted-informal} and the reduction of \cite{BernsteinDL21} reduce the approximation guarantee by $2/3-\eps$. Achieving a general reduction for non-bipartite graphs that only loses a $(1-\eps)$-factor is probably the main open problem in the area, and would be very interesting even with an update-time overhead that is exponential in $\eps$. A similar open problem is to achieve such a reduction for other models, including a reduction for streaming and MPC that does not increase the number of passes/rounds. Note that both our \cref{thm:main-small-weights-informal} and the reduction of \cite{GuptaP13} already apply to non-bipartite graphs, so we can safely assume that weights are small integers. The only remaining challenge is thus that the framework of \cite{BernsteinDL21} uses an earlier tool called \emph{graph unfolding} (first given by \cite{KaoLST01}) to reduce from small weights to unit weights, but this tool relies on the vertex-cover dual and seems limited to bipartite graphs.  

\item A second limitation is specific to our paper: as discussed in the introduction, our reduction only works for $(1-\eps)$-approximate matching and not for general $\alpha$-approximate matching. Removing this restriction would show that in bipartite graphs at least, unweighted and weighted matching have almost equivalent complexities in a wide variety of computational models.

\end{enumerate}

\section*{Acknowledgements}
We thank the anonymous reviewers for their helpful feedback and discussions.

\bibliography{reference}

@inproceedings{GamlathKMS19,
  author       = {Buddhima Gamlath and
                  Sagar Kale and
                  Slobodan Mitrovic and
                  Ola Svensson},
  title        = {Weighted Matchings via Unweighted Augmentations},
  booktitle    = {Proceedings of the 2019 {ACM} Symposium on Principles of Distributed Computing, {PODC} 2019},
  pages        = {491--500},
  publisher    = {{ACM}},
  year         = {2019},
  note         = {Available at \url{https://arxiv.org/abs/1811.02760}.}
}

@article{GuruswamiO16,
  author       = {Venkatesan Guruswami and
                  Krzysztof Onak},
  title        = {Superlinear Lower Bounds for Multipass Graph Processing},
  journal      = {Algorithmica},
  volume       = {76},
  number       = {3},
  pages        = {654--683},
  year         = {2016},
  url          = {https://doi.org/10.1007/s00453-016-0138-7},
  doi          = {10.1007/S00453-016-0138-7},
  bibsource    = {dblp computer science bibliography, https://dblp.org},
  note         = {Available at \url{https://arxiv.org/abs/1212.6925}.}
}

@inproceedings{Liu24,
  author       = {Yang P. Liu},
  title        = {On Approximate Fully-Dynamic Matching and Online Matrix-Vector Multiplication},
  booktitle    = {65th {IEEE} Annual Symposium on Foundations of Computer Science, {FOCS}
                  2024},
  publisher    = {{IEEE}},
  year         = {2024},
  note         = {Available at \url{https://arxiv.org/abs/2403.02582}.}
}

@inproceedings{AssadiLT21,
  author       = {Sepehr Assadi and
                  S. Cliff Liu and
                  Robert E. Tarjan},
  title        = {An Auction Algorithm for Bipartite Matching in Streaming and Massively
                  Parallel Computation Models},
  booktitle    = {4th Symposium on Simplicity in Algorithms, {SOSA} 2021},
  pages        = {165--171},
  publisher    = {{SIAM}},
  year         = {2021},
  url          = {https://doi.org/10.1137/1.9781611976496.18},
  doi          = {10.1137/1.9781611976496.18},
  bibsource    = {dblp computer science bibliography, https://dblp.org}
}

@inproceedings{AssadiJJST22,
  author       = {Sepehr Assadi and
                  Arun Jambulapati and
                  Yujia Jin and
                  Aaron Sidford and
                  Kevin Tian},
  title        = {Semi-Streaming Bipartite Matching in Fewer Passes and Optimal Space},
  booktitle    = {Proceedings of the 2022 {ACM-SIAM} Symposium on Discrete Algorithms,
                  {SODA} 2022},
  pages        = {627--669},
  publisher    = {{SIAM}},
  year         = {2022},
  url          = {https://doi.org/10.1137/1.9781611977073.29},
  doi          = {10.1137/1.9781611977073.29},
  bibsource    = {dblp computer science bibliography, https://dblp.org},
  note         = {Available at \url{https://arxiv.org/abs/2011.03495}.}
}

@inproceedings{LiuKK23,
  author       = {Quanquan C. Liu and
                  Yiduo Ke and
                  Samir Khuller},
  title        = {Scalable Auction Algorithms for Bipartite Maximum Matching Problems},
  booktitle    = {Approximation, Randomization, and Combinatorial Optimization. Algorithms
                  and Techniques, {APPROX/RANDOM} 2023},
  series       = {LIPIcs},
  volume       = {275},
  pages        = {28:1--28:24},
  year         = {2023},
  url          = {https://doi.org/10.4230/LIPIcs.APPROX/RANDOM.2023.28},
  doi          = {10.4230/LIPICS.APPROX/RANDOM.2023.28},
  bibsource    = {dblp computer science bibliography, https://dblp.org},
  note         = {Available at \url{https://arxiv.org/abs/2307.08979}.}
}

@inproceedings{AhnG14,
  author       = {Kook Jin Ahn and
                  Sudipto Guha},
  title        = {Near Linear Time Approximation Schemes for Uncapacitated and Capacitated
                  b-Matching Problems in Nonbipartite Graphs},
  booktitle    = {Proceedings of the Twenty-Fifth Annual {ACM-SIAM} Symposium on Discrete
                  Algorithms, {SODA} 2014},
  pages        = {239--258},
  publisher    = {{SIAM}},
  year         = {2014},
  url          = {https://doi.org/10.1137/1.9781611973402.18},
  doi          = {10.1137/1.9781611973402.18},
  bibsource    = {dblp computer science bibliography, https://dblp.org},
  note         = {Available at \url{https://arxiv.org/abs/1307.4355}.}
}

@inproceedings{ChenKLPGS22,
  author       = {Li Chen and
                  Rasmus Kyng and
                  Yang P. Liu and
                  Richard Peng and
                  Maximilian Probst Gutenberg and
                  Sushant Sachdeva},
  title        = {Maximum Flow and Minimum-Cost Flow in Almost-Linear Time},
  booktitle    = {63rd {IEEE} Annual Symposium on Foundations of Computer Science, {FOCS} 2022},
  pages        = {612--623},
  year         = {2022},
  url          = {https://doi.org/10.1109/FOCS54457.2022.00064},
  doi          = {10.1109/FOCS54457.2022.00064},
  bibsource    = {dblp computer science bibliography, https://dblp.org},
  note         = {Available at \url{https://arxiv.org/abs/2203.00671}.}
}

@inproceedings{Wajc20,
  author       = {David Wajc},
  title        = {Rounding dynamic matchings against an adaptive adversary},
  booktitle    = {Proceedings of the 52nd Annual {ACM} {SIGACT} Symposium on Theory
                  of Computing, {STOC} 2020},
  pages        = {194--207},
  publisher    = {{ACM}},
  year         = {2020},
  url          = {https://doi.org/10.1145/3357713.3384258},
  doi          = {10.1145/3357713.3384258},
  bibsource    = {dblp computer science bibliography, https://dblp.org},
  note         = {Available at \url{https://arxiv.org/abs/1911.05545}.}
}

@inproceedings{BhattacharyaKSW24,
  author       = {Sayan Bhattacharya and
                  Peter Kiss and
                  Aaron Sidford and
                  David Wajc},
  title        = {Near-Optimal Dynamic Rounding of Fractional Matchings in Bipartite
                  Graphs},
  booktitle    = {Proceedings of the 56th Annual {ACM} Symposium on Theory of Computing,
                  {STOC} 2024},
  pages        = {59--70},
  publisher    = {{ACM}},
  year         = {2024},
  url          = {https://doi.org/10.1145/3618260.3649648},
  doi          = {10.1145/3618260.3649648},
  bibsource    = {dblp computer science bibliography, https://dblp.org},
  note         = {Available at \url{https://arxiv.org/abs/2306.11828}.}
}

@article{DuanP14,
  author       = {Ran Duan and
                  Seth Pettie},
  title        = {Linear-Time Approximation for Maximum Weight Matching},
  journal      = {J. {ACM}},
  volume       = {61},
  number       = {1},
  pages        = {1:1--1:23},
  year         = {2014},
  url          = {https://doi.org/10.1145/2529989},
  doi          = {10.1145/2529989},
  bibsource    = {dblp computer science bibliography, https://dblp.org}
}

@inproceedings{ArarCCSW18,
  author       = {Moab Arar and
                  Shiri Chechik and
                  Sarel Cohen and
                  Cliff Stein and
                  David Wajc},
  title        = {Dynamic Matching: Reducing Integral Algorithms to Approximately-Maximal
                  Fractional Algorithms},
  booktitle    = {Proc. 45th Int. Colloquium on Automata, Languages, and Programming},
  series       = {LIPIcs},
  volume       = {107},
  pages        = {7:1--7:16},
  year         = {2018},
  url          = {https://doi.org/10.4230/LIPIcs.ICALP.2018.7},
  doi          = {10.4230/LIPIcs.ICALP.2018.7},
  bibsource    = {dblp computer science bibliography, https://dblp.org},
  note         = {Available at \url{https://arxiv.org/abs/1711.06625}.}
}

@inproceedings{BernsteinGS20,
  author       = {Aaron Bernstein and
                  Maximilian Probst Gutenberg and
                  Thatchaphol Saranurak},
  title        = {Deterministic Decremental Reachability, SCC, and Shortest Paths via
                  Directed Expanders and Congestion Balancing},
  booktitle    = {61st {IEEE} Annual Symposium on Foundations of Computer Science, {FOCS}
                  2020},
  pages        = {1123--1134},
  publisher    = {{IEEE}},
  year         = {2020},
  url          = {https://doi.org/10.1109/FOCS46700.2020.00108},
  doi          = {10.1109/FOCS46700.2020.00108},
  bibsource    = {dblp computer science bibliography, https://dblp.org},
  note         = {Available at \url{https://arxiv.org/abs/2009.02584}.}
}

@inproceedings{JambulapatiJST22,
  author       = {Arun Jambulapati and
                  Yujia Jin and
                  Aaron Sidford and
                  Kevin Tian},
  title        = {Regularized Box-Simplex Games and Dynamic Decremental Bipartite Matching},
  booktitle    = {49th International Colloquium on Automata, Languages, and Programming,
                  {ICALP} 2022},
  series       = {LIPIcs},
  volume       = {229},
  pages        = {77:1--77:20},
  year         = {2022},
  url          = {https://doi.org/10.4230/LIPIcs.ICALP.2022.77},
  doi          = {10.4230/LIPIcs.ICALP.2022.77},
  bibsource    = {dblp computer science bibliography, https://dblp.org},
  note         = {Available at \url{https://arxiv.org/abs/2204.12721}.}
}

@inproceedings{StubbsW17,
  author       = {Daniel Stubbs and
                  Virginia Vassilevska Williams},
  title        = {Metatheorems for Dynamic Weighted Matching},
  booktitle    = {8th Innovations in Theoretical Computer Science Conference, {ITCS}
                  2017},
  series       = {LIPIcs},
  volume       = {67},
  pages        = {58:1--58:14},
  year         = {2017},
  url          = {https://doi.org/10.4230/LIPIcs.ITCS.2017.58},
  doi          = {10.4230/LIPIcs.ITCS.2017.58},
  bibsource    = {dblp computer science bibliography, https://dblp.org}
}

@article{AhnGuha11a,
  author       = {Kook Jin Ahn and
                  Sudipto Guha},
  title        = {Laminar Families and Metric Embeddings: Non-bipartite Maximum Matching
                  Problem in the Semi-Streaming Model},
  journal      = {CoRR},
  volume       = {abs/1104.4058},
  year         = {2011},
  url          = {http://arxiv.org/abs/1104.4058},
  eprinttype    = {arXiv},
  eprint       = {1104.4058},
  bibsource    = {dblp computer science bibliography, https://dblp.org},
  note = {Available at \url{https://arxiv.org/abs/1104.4058}}
}

@article{AhnG18,
  author       = {Kook Jin Ahn and
                  Sudipto Guha},
  title        = {Access to Data and Number of Iterations: Dual Primal Algorithms for
                  Maximum Matching under Resource Constraints},
  journal      = {{ACM} Trans. Parallel Comput.},
  volume       = {4},
  number       = {4},
  pages        = {17:1--17:40},
  year         = {2018},
  url          = {https://doi.org/10.1145/3154855},
  doi          = {10.1145/3154855},
  bibsource    = {dblp computer science bibliography, https://dblp.org},
  note ={Available at \url{https://arxiv.org/abs/1307.4359}}
}

@inproceedings{Assadi24,
  author       = {Sepehr Assadi},
  title        = {A Simple $(1-\varepsilon)$-Approximation Semi-Streaming
                  Algorithm for Maximum (Weighted) Matching},
  booktitle    = {2024 Symposium on Simplicity in Algorithms, {SOSA} 2024},
  pages        = {337--354},
  publisher    = {{SIAM}},
  year         = {2024},
  url          = {https://doi.org/10.1137/1.9781611977936.31},
  doi          = {10.1137/1.9781611977936.31},
  bibsource    = {dblp computer science bibliography, https://dblp.org},
  note = {Available at \url{https://arxiv.org/abs/2307.02968}}
}

@article{KaoLST01,
  author       = {Ming{-}Yang Kao and
                  Tak Wah Lam and
                  Wing{-}Kin Sung and
                  Hing{-}Fung Ting},
  title        = {A Decomposition Theorem for Maximum Weight Bipartite Matchings},
  journal      = {{SIAM} J. Comput.},
  volume       = {31},
  number       = {1},
  pages        = {18--26},
  year         = {2001},
  url          = {https://doi.org/10.1137/S0097539799361208},
  doi          = {10.1137/S0097539799361208},
  bibsource    = {dblp computer science bibliography, https://dblp.org}
}

@inproceedings{GuptaP13,
  author       = {Manoj Gupta and
                  Richard Peng},
  title        = {Fully Dynamic $(1+\epsilon)$-Approximate Matchings},
  booktitle    = {54th Annual {IEEE} Symposium on Foundations of Computer Science, {FOCS}
                  2013},
  pages        = {548--557},
  publisher    = {{IEEE} Computer Society},
  year         = {2013},
  url          = {https://doi.org/10.1109/FOCS.2013.65},
  doi          = {10.1109/FOCS.2013.65},
  bibsource    = {dblp computer science bibliography, https://dblp.org},
  note         = {Available at \url{https://arxiv.org/abs/1304.0378}.}
}

@inproceedings{BhattacharyaKSW23dynamic2,
  author       = {Sayan Bhattacharya and
                  Peter Kiss and
                  Thatchaphol Saranurak and
                  David Wajc},
  title        = {Dynamic Matching with Better-than-2 Approximation in Polylogarithmic
                  Update Time},
  booktitle    = {Proceedings of the 2023 {ACM-SIAM} Symposium on Discrete Algorithms,
                  {SODA} 2023},
  pages        = {100--128},
  publisher    = {{SIAM}},
  year         = {2023},
  url          = {https://doi.org/10.1137/1.9781611977554.ch5},
  doi          = {10.1137/1.9781611977554.CH5},
  bibsource    = {dblp computer science bibliography, https://dblp.org},
  note         = {Available at \url{https://arxiv.org/abs/2207.07438}.}
}

@inproceedings{Behnezhad23,
  author       = {Soheil Behnezhad},
  title        = {Dynamic Algorithms for Maximum Matching Size},
  booktitle    = {Proceedings of the 2023 {ACM-SIAM} Symposium on Discrete Algorithms,
                  {SODA} 2023},
  pages        = {129--162},
  publisher    = {{SIAM}},
  year         = {2023},
  url          = {https://doi.org/10.1137/1.9781611977554.ch6},
  doi          = {10.1137/1.9781611977554.CH6},
  bibsource    = {dblp computer science bibliography, https://dblp.org},
  note         = {Available at \url{https://arxiv.org/abs/2207.07607}.}
}

@inproceedings{BhattacharyaKS23dynamic1,
  author       = {Sayan Bhattacharya and
                  Peter Kiss and
                  Thatchaphol Saranurak},
  title        = {Dynamic $(1+\epsilon)$-Approximate Matching Size in Truly Sublinear
                  Update Time},
  booktitle    = {64th {IEEE} Annual Symposium on Foundations of Computer Science, {FOCS}
                  2023},
  pages        = {1563--1588},
  publisher    = {{IEEE}},
  year         = {2023},
  url          = {https://doi.org/10.1109/FOCS57990.2023.00095},
  doi          = {10.1109/FOCS57990.2023.00095},
  bibsource    = {dblp computer science bibliography, https://dblp.org},
  note         = {Available at \url{https://arxiv.org/abs/2302.05030}.}
}

@inproceedings{BhattacharyaKS23,
  author       = {Sayan Bhattacharya and
                  Peter Kiss and
                  Thatchaphol Saranurak},
  title        = {Dynamic Algorithms for Packing-Covering LPs via Multiplicative Weight
                  Updates},
  booktitle    = {Proceedings of the 2023 {ACM-SIAM} Symposium on Discrete Algorithms,
                  {SODA} 2023},
  pages        = {1--47},
  publisher    = {{SIAM}},
  year         = {2023},
  url          = {https://doi.org/10.1137/1.9781611977554.ch1},
  doi          = {10.1137/1.9781611977554.CH1},
  bibsource    = {dblp computer science bibliography, https://dblp.org},
  note         = {Available at \url{https://arxiv.org/abs/2207.07519}.}
}

@inproceedings{Gupta14,
  author       = {Manoj Gupta},
  title        = {Maintaining Approximate Maximum Matching in an Incremental Bipartite
                  Graph in Polylogarithmic Update Time},
  booktitle    = {34th International Conference on Foundation of Software Technology
                  and Theoretical Computer Science, {FSTTCS} 2014},
  series       = {LIPIcs},
  volume       = {29},
  pages        = {227--239},
  year         = {2014},
  url          = {https://doi.org/10.4230/LIPIcs.FSTTCS.2014.227},
  doi          = {10.4230/LIPICS.FSTTCS.2014.227},
  bibsource    = {dblp computer science bibliography, https://dblp.org}
}

@inproceedings{GrandoniLSSS19,
  author       = {Fabrizio Grandoni and
                  Stefano Leonardi and
                  Piotr Sankowski and
                  Chris Schwiegelshohn and
                  Shay Solomon},
  title        = {$(1+\epsilon)$-Approximate Incremental Matching in Constant
                  Deterministic Amortized Time},
  booktitle    = {Proceedings of the Thirtieth Annual {ACM-SIAM} Symposium on Discrete
                  Algorithms, {SODA} 2019},
  pages        = {1886--1898},
  publisher    = {{SIAM}},
  year         = {2019},
  url          = {https://doi.org/10.1137/1.9781611975482.114},
  doi          = {10.1137/1.9781611975482.114},
  bibsource    = {dblp computer science bibliography, https://dblp.org}
}

@inproceedings{BlikstadK23,
  author       = {Joakim Blikstad and
                  Peter Kiss},
  title        = {Incremental $(1-\epsilon)$-Approximate Dynamic Matching in ${O}(\mathrm{poly}(1/\epsilon))$
                  Update Time},
  booktitle    = {31st Annual European Symposium on Algorithms, {ESA} 2023},
  series       = {LIPIcs},
  volume       = {274},
  pages        = {22:1--22:19},
  publisher    = {Schloss Dagstuhl - Leibniz-Zentrum f{\"{u}}r Informatik},
  year         = {2023},
  url          = {https://doi.org/10.4230/LIPIcs.ESA.2023.22},
  doi          = {10.4230/LIPICS.ESA.2023.22},
  bibsource    = {dblp computer science bibliography, https://dblp.org},
  note         = {Available at \url{https://arxiv.org/abs/2302.08432}.}
}

@inproceedings{BernsteinS15,
  author       = {Aaron Bernstein and
                  Cliff Stein},
  title        = {Fully Dynamic Matching in Bipartite Graphs},
  booktitle    = {Automata, Languages, and Programming - 42nd International Colloquium,
                  {ICALP} 2015},
  series       = {Lecture Notes in Computer Science},
  volume       = {9134},
  pages        = {167--179},
  publisher    = {Springer},
  year         = {2015},
  url          = {https://doi.org/10.1007/978-3-662-47672-7\_14},
  doi          = {10.1007/978-3-662-47672-7\_14},
  bibsource    = {dblp computer science bibliography, https://dblp.org},
  note         = {Available at \url{https://arxiv.org/abs/1506.07076}.}
}

@inproceedings{BernsteinS16,
  author       = {Aaron Bernstein and
                  Cliff Stein},
  title        = {Faster Fully Dynamic Matchings with Small Approximation Ratios},
  booktitle    = {Proceedings of the Twenty-Seventh Annual {ACM-SIAM} Symposium on Discrete
                  Algorithms, {SODA} 2016},
  pages        = {692--711},
  publisher    = {{SIAM}},
  year         = {2016},
  url          = {https://doi.org/10.1137/1.9781611974331.ch50},
  doi          = {10.1137/1.9781611974331.CH50},
  bibsource    = {dblp computer science bibliography, https://dblp.org}
}

@inproceedings{HenzingerKNS15,
  author       = {Monika Henzinger and
                  Sebastian Krinninger and
                  Danupon Nanongkai and
                  Thatchaphol Saranurak},
  title        = {Unifying and Strengthening Hardness for Dynamic Problems via the Online
                  Matrix-Vector Multiplication Conjecture},
  booktitle    = {Proceedings of the 47th Annual {ACM} Symposium on Theory
                  of Computing, {STOC} 2015},
  pages        = {21--30},
  year         = {2015},
  url          = {https://doi.org/10.1145/2746539.2746609},
  doi          = {10.1145/2746539.2746609},
  bibsource    = {dblp computer science bibliography, https://dblp.org},
  note         = {Available at \url{https://arxiv.org/abs/1511.06773}.}
}

@inproceedings{CharikarS18,
  author       = {Moses Charikar and
                  Shay Solomon},
  title        = {Fully Dynamic Almost-Maximal Matching: Breaking the Polynomial Worst-Case
                  Time Barrier},
  booktitle    = {45th International Colloquium on Automata, Languages, and Programming,
                  {ICALP} 2018},
  series       = {LIPIcs},
  volume       = {107},
  pages        = {33:1--33:14},
  year         = {2018},
  url          = {https://doi.org/10.4230/LIPIcs.ICALP.2018.33},
  doi          = {10.4230/LIPICS.ICALP.2018.33},
  bibsource    = {dblp computer science bibliography, https://dblp.org},
  note         = {Available at \url{https://arxiv.org/abs/1711.06883}.}
}

@inproceedings{BhattacharyaK21,
  author       = {Sayan Bhattacharya and
                  Peter Kiss},
  title        = {Deterministic Rounding of Dynamic Fractional Matchings},
  booktitle    = {48th International Colloquium on Automata, Languages, and Programming,
                  {ICALP} 2021},
  series       = {LIPIcs},
  volume       = {198},
  pages        = {27:1--27:14},
  year         = {2021},
  url          = {https://doi.org/10.4230/LIPIcs.ICALP.2021.27},
  doi          = {10.4230/LIPICS.ICALP.2021.27},
  bibsource    = {dblp computer science bibliography, https://dblp.org},
  note         = {Available at \url{https://arxiv.org/abs/2105.01615}.}
}

@inproceedings{AzarmehrBR24,
  author       = {Amir Azarmehr and
                  Soheil Behnezhad and
                  Mohammad Roghani},
  title        = {Fully Dynamic Matching: $(2-\sqrt{2})$-Approximation in Polylog
                  Update Time},
  booktitle    = {Proceedings of the 35th {ACM-SIAM} Symposium on Discrete Algorithms,
                  {SODA} 2024},
  year         = {2024},
  note         = {Available at \url{https://arxiv.org/abs/2307.08772}.}
}

@inproceedings{BernsteinDL21,
  author       = {Aaron Bernstein and
                  Aditi Dudeja and
                  Zachary Langley},
  title        = {A framework for dynamic matching in weighted graphs},
  booktitle    = {Proceedings of the 53rd Annual {ACM} Symposium on Theory of Computing, {STOC} 2021},
  pages        = {668--681},
  publisher    = {{ACM}},
  year         = {2021},
  url          = {https://doi.org/10.1145/3406325.3451113},
  doi          = {10.1145/3406325.3451113},
  bibsource    = {dblp computer science bibliography, https://dblp.org}
}

@inproceedings{BrandCLPGSS23,
  author       = {Jan van den Brand and
                  Li Chen and
                  Richard Peng and
                  Rasmus Kyng and
                  Yang P. Liu and
                  Maximilian Probst Gutenberg and
                  Sushant Sachdeva and
                  Aaron Sidford},
  title        = {A Deterministic Almost-Linear Time Algorithm for Minimum-Cost Flow},
  booktitle    = {64th {IEEE} Annual Symposium on Foundations of Computer Science, {FOCS}
                  2023},
  pages        = {503--514},
  year         = {2023},
  url          = {https://doi.org/10.1109/FOCS57990.2023.00037},
  doi          = {10.1109/FOCS57990.2023.00037},
  bibsource    = {dblp computer science bibliography, https://dblp.org},
  note         = {Available at \url{https://arxiv.org/abs/2309.16629}.}
}

@inproceedings{ChenST23,
  author       = {Jiale Chen and
                  Aaron Sidford and
                  Ta{-}Wei Tu},
  title        = {Entropy Regularization and Faster Decremental Matching in General
                  Graphs},
  booktitle    = {arXiv Preprint},
  year         = {2023},
  note         = {Available at \url{https://arxiv.org/abs/2312.09077}.},
}

@inproceedings{SolomonS21,
  author       = {Noam Solomon and
                  Shay Solomon},
  title        = {A Generalized Matching Reconfiguration Problem},
  booktitle    = {12th Innovations in Theoretical Computer Science Conference, {ITCS}
                  2021},
  series       = {LIPIcs},
  volume       = {185},
  pages        = {57:1--57:20},
  year         = {2021},
  url          = {https://doi.org/10.4230/LIPIcs.ITCS.2021.57},
  doi          = {10.4230/LIPICS.ITCS.2021.57},
  bibsource    = {dblp computer science bibliography, https://dblp.org},
  note         = {Available at \url{https://arxiv.org/abs/1803.05825}.}
}

@inproceedings{AnandBGS12,
  author       = {Abhash Anand and
                  Surender Baswana and
                  Manoj Gupta and
                  Sandeep Sen},
  title        = {Maintaining Approximate Maximum Weighted Matching in Fully Dynamic
                  Graphs},
  booktitle    = {{IARCS} Annual Conference on Foundations of Software Technology and
                  Theoretical Computer Science, {FSTTCS} 2012},
  series       = {LIPIcs},
  volume       = {18},
  pages        = {257--266},
  publisher    = {Schloss Dagstuhl - Leibniz-Zentrum f{\"{u}}r Informatik},
  year         = {2012},
  url          = {https://doi.org/10.4230/LIPIcs.FSTTCS.2012.257},
  doi          = {10.4230/LIPICS.FSTTCS.2012.257},
  bibsource    = {dblp computer science bibliography, https://dblp.org},
  note         = {Available at \url{https://arxiv.org/abs/1207.3976}.}
}

@inproceedings{Dudeja24,
  author       = {Aditi Dudeja},
  booktitle    = {arXiv Preprint},
  title        = {A Note on Rounding Matchings in General Graphs},
  year         = {2024},
  note         = {Available at \url{https://arxiv.org/abs/2402.03068}.}
}

@inproceedings{Kapralov21,
  author       = {Michael Kapralov},
  title        = {Space Lower Bounds for Approximating Maximum Matching in the Edge
                  Arrival Model},
  booktitle    = {Proceedings of the 2021 {ACM-SIAM} Symposium on Discrete Algorithms,
                  {SODA} 2021},
  pages        = {1874--1893},
  publisher    = {{SIAM}},
  year         = {2021},
  url          = {https://doi.org/10.1137/1.9781611976465.112},
  doi          = {10.1137/1.9781611976465.112},
  bibsource    = {dblp computer science bibliography, https://dblp.org},
  note         = {Available at \url{https://arxiv.org/abs/2103.11669}.}
}

@article{AhnG13,
  author       = {Kook Jin Ahn and
                  Sudipto Guha},
  title        = {Linear programming in the semi-streaming model with application to
                  the maximum matching problem},
  journal      = {Inf. Comput.},
  volume       = {222},
  pages        = {59--79},
  year         = {2013},
  url          = {https://doi.org/10.1016/j.ic.2012.10.006},
  doi          = {10.1016/J.IC.2012.10.006},
  bibsource    = {dblp computer science bibliography, https://dblp.org}
}

\appendix
\section{Analysis of \texorpdfstring{\cite{BernsteinDL21}}{[BDL21]}}\label{sec:appendix:BDL}
We first give the definitions and notations used in the statement of the algorithm and proof.

\begin{definition}[\cite{KaoLST01}]
Let $G$ be a graph with integer edge weights in $[W]$. The unfolded graph $\phi(G)$ is an unweighted graph defined as follows: For each vertex $u\in G$, there are $W$ copies of $u$, $\{u^1,u^2,\dots,u^W\}$, in $\phi(G)$. Corresponding to each edge $uv$ in $G$ there are $w_{uv}$ edges $\left\{u^{i} v^{w_{uv}-i+1}\right\}_{i\in [w_{uv}]}$ in $\phi(G)$.
\end{definition}

A simple consequence of the above definition is the following observation.

\begin{observation}\label{obs:unfolding:graphsize}
Let $G$ be any weighted bipartite graph, and suppose $W$ is the ratio between the maximum and minimum edge weights, then $|V(\phi(G))|=W\cdot n$ and $|E(\phi(G))|=W\cdot m$. 
\end{observation}

\begin{fact}[\cite{KaoLST01}]
Let $G$ be a weighted bipartite graph, and suppose $M$ is the maximum weight matching of $G$ and let $M_{\phi}$ be the MCM of $\phi(G)$. Then, $w(M)=|M_{\phi}|$.
\end{fact}

\begin{definition}
    Let $G$ be a weighted graph, and let $H\subseteq \phi(G)$. The refolded graph $\mathcal{R}(H)$ has vertex set $V(G)$, and edges $E(\mathcal{R}(H))=\left\{uv \in G\mid u^i v^j\in H\text{ for }i+j+1=w(uv)\right\}$. 
\end{definition}

\begin{fact}[\cite{BernsteinDL21}]\label{fact:BDL:approximation ratio}
    Let $G$ be a weighted graph with weight function $w$ and let $M$ be an $\alpha$-approximate matching of $\phi(G)$.
    If $G$ is bipartite, then $\mu_w(\mathcal{R}(M)) \ge \alpha \mu_w(G).$
    If $G$ is not bipartite, then $\mu_w(\mathcal{R}(M)) \ge \frac{2}{3} \alpha \mu_w(G).$
\end{fact}

Now we modify Algorithm 1 of~\cite{BernsteinDL21} to achieve a better amortized update time.

\begin{algorithm2e}[!ht]
  \caption{Bipartite Reduction} \label{alg:BDL21}
  
  \SetEndCharOfAlgoLine{}

  \SetKwInput{KwData}{Input}
  \SetKwInput{KwResult}{Output}
  \SetKwInOut{State}{global}
  \SetKwProg{KwProc}{function}{}{}
  \SetKwFunction{Initialize}{Initialize}
  \SetKwFunction{Update}{Update}
  \SetKwFunction{Rebuild}{Rebuild}

  \KwData{A dynamic algorithm for $(1-\eps)$-approximate MCM in bipartite graphs $\mathcal A_u$}
  
  \KwProc{\Initialize{}} {
    Initialize $\mathcal A_u$ with the unfolded graph $\phi(G)$.\;
    Denote $M_u$ as the matching maintained by $\mathcal A_u$.\;
    $M\gets\Rebuild{}$.\;
    \textbf{output} $M$.\;
  }
  \KwProc{\Update{$(uv)$}} {
    Update $(u^i,v^{w_{uv}-i})$ in $\phi(G)$ for $i\in [w_{uv}]$ accordingly.\;
    Use $\mathcal A_u$ to maintain a matching $M_u$ of $\phi(G)$.\;
    $c\gets c+1$.\;
    \eIf{$c< \eps\cdot W^*/W$}{
      $M\gets M\setminus uv$.\;
    }{
      $M\gets\Rebuild{}$.\;
    }
    \textbf{output} $M$.\;
  }
  \KwProc{\Rebuild{}}{
    $M\gets$ a $(1-\eps)$-approximate MWM on the refolded graph $\mathcal R(M_u)$.\;
    $c\gets 0$, $W^*\gets w(M)$.\;
    \textbf{output} M.\;
  }
\end{algorithm2e}

\AnalysisBDL*
\begin{proof}
    We first prove that \cref{alg:BDL21} maintains a $(1-O(\eps))$-approximate MWM on a bipartite $G$. Since $M_u$ is a $(1-\eps)$-approximate matching, thus by~\cref{fact:BDL:approximation ratio}, $\mu_w(\mathcal R(M_u))\geq (1-\eps)\mu_w(G)$. After each $\Rebuild$, $M$ is a $(1-\eps)$-approximate MWM on $\mathcal R(M_u)$, thus $W^*=w(M)\geq (1-\eps)\cdot \mu_w(R(M_u))\geq (1-2\eps)\cdot \mu_w(G)$. Between $\Rebuild$, there are at most $\eps\cdot W^*/W$ edge updates, thus $w(M)\geq (1-\eps) W^*$ and $\mu_w(G)\leq W^*/(1-2\eps)+\eps\cdot W^*\leq (1+4\eps)W^*$, and $w(M)\geq (1-5\eps) \mu_w(G)$. On a general graph, according to~\cref{fact:BDL:approximation ratio}, $\mu_w(\mathcal R(M_u))\geq \frac{2}{3}(1-\eps)\mu_w(G)$ and similarly we can prove that $w(M)\geq (\frac{2}{3}-O(\eps))\mu_w(G)$.

    Now we analyze the running time of \cref{alg:BDL21} for both bipartite and non-bipartite graphs. By~\cref{thm:DP14}, the initialization can be implemented in time 
    \[O(\I(nW, mW, \eps)+m\log(\eps^{-1})\eps^{-1}).\]
    The running time of $\mathcal A_u$ is $W\cdot \U(nW,mW,\Theta(\eps))$.
    Since 
    \[|\mathcal R(M_u)|\leq |M_u|\leq \mu(\phi(G))=\mu_w(G)\leq O(1)\cdot W^*,\]
    each \Rebuild takes total time
    $O(|\mathcal R(M_u)|\log(\eps^{-1})\eps^{-1})=O(W^*\log (\eps^{-1})\eps^{-1})$ by~\cref{thm:DP14}. Thus the amortized cost of \Rebuild and the amortized recourse of the algorithm is
    \[O(W^*\log(\eps^{-1})\eps^{-1})/(\eps\cdot W^*/W)=O(W\log (\eps^{-1})\eps^{-2}).\]
\end{proof}

\section{Counterexample to \texorpdfstring{\cref{conjecture:weight partition}}{Question 3.1}}\label{sec:appendix:counterexample}

To answer \cref{conjecture:weight partition} in the negative, we prove the following claim:

\begin{claim}\label{claim:counterexample of conjecture}
    There is a graph $G$ such that for \emph{any} weight partition of $G$, $[\ell_1,r_1),[\ell_2,r_2),\dots,[\ell_k,r_k)\subseteq \R$, the following holds: if all possible choices of MWM $M_i$ of  $G_{[\ell_i,r_i)}$ satisfy
    \[\mu_w(M_1\cup M_2\cup \dots\cup M_k)\geq (1-\delta)\cdot \mu_w(G),\]
    then there is a weight class $[\ell_i,r_i)$ such that $r_i\geq \ell_i\cdot \exp(\Omega(\delta^{-1}))$.
\end{claim}

To prove~\cref{claim:counterexample of conjecture}, we will use the following gadget to explicitly build the graph $G$.

\begin{definition}[Level-$i$ Gadget]
A level-$i$ Gadget is a path with three edges $a_i,b_i,c_i$ such that the edge weights of $a_i$ and $b_i$ are $1.5^i$, and the edge weight of $c_i$ is $1.5^{i+1}$.
\end{definition}

\paragraph{Proof Intuition:} The MWM on a level-$i$ gadget is $1.5^i + 1.5^{i+1} = 1.5^{i}\cdot 2.5$ by choosing $a_i$ and $c_i$. But now let us consider what happens if the gadget is ``broken", meaning that it is partitioned into two different weight classes. More concretely, say that $a_i,b_i \in [\ell_j, r_j)$, while $c_i \in [\ell_{j+1},r_{j+1})$. Then, one valid MWM of weight class $[\ell_j, r_j)$ is $M_j = \{b_i\}$, and clearly we have $M_{j+1} = \{c_i\}$. As a result, $\mu_w(M_j \cup M_{j+1}) = w(c_i) = 1.5^{i+1} \leq \frac{3}{5}\cdot \mu_W(\{a_i,b_i,c_i\})$. In other words, the loss incurred by a broken gadget is a constant fraction of the weight of the gadget. Intuitively, we will have gadgets on different levels, and any partition into $k$ weight classes will break $k-1$ of the gadgets.
 Since we can only afford a total loss of only $\delta \mu_w(G)$, the average weight class $[\ell_i, r_i]$ must contain at least $\Omega(\delta^{-1})$ non-broken gadgets to make up for the loss of the broken ones. Since each gadget is $1.5$-wide, this implies that the average weight class must be $1.5^{\Omega(\delta^{-1})}$-wide. We now formally prove~\cref{claim:counterexample of conjecture}.

\begin{proof}
    We first build the graph $G$ using our gadgets. For $i=0,1,\dots,N$, where $N=\lfloor \log_{1.5} W\rfloor$, $G$ contains $1.5^{N-i}$ level-$i$ gadgets. Therefore $\mu_w(G)=1.5^N\cdot 2.5\cdot N$ since the total MWM of each level is $1.5^{N-i}\cdot 1.5^i\cdot 2.5=1.5^{N}\cdot 2.5$. Now consider any weight partition $[\ell_1,r_1),\dots[\ell_k,r_k)$. We assume that for any weight class $[\ell_i,r_i)$, there is some integer $j$ such that $\ell_i\leq 1.5^j<r_i$. This is w.l.o.g.\ since otherwise that weight class does not contain any edge and we can merge it with one of its neighboring weight classes.

    $k-1$ levels of broken gadgets can be found in this graph. For $i=0,1,2,\dots,k-1$, consider the largest $j$ such that $1.5^j<r_i$ and level-$j$ gadgets will be broken. On a broken level-$j$ gadget, the MWM will be $1.5^{j+1}$ instead of $1.5^j\cdot 2.5$. Since there are $1.5^{N-j}$ level-$j$ gadgets, the weight loss of all level-$j$ gadgets is $1.5^N$. Therefore the total weight loss is $1.5^N\cdot (k-1)$.
    
    Suppose the MWM on the union has weight at least $(1-\delta)\cdot \mu_w(G)=(1-\delta)\cdot 1.5^N\cdot2.5\cdot N.$ Then we have $1.5^N\cdot (k-1)\leq \delta\cdot 1.5^N\cdot 2.5\cdot N$ meaning $k\leq 2.5\cdot \delta\cdot N+1$. Therefore the would be a weight class $[\ell_i,r_i)$ with $r_i\geq \ell_i\cdot 1.5^{\frac{N}{2.5\cdot \delta\cdot N+1}}=\ell_i\cdot \exp(\Omega(\delta^{-1}))$.
\end{proof}

\section{Reduction of \texorpdfstring{$\alpha$-Approximation}{a-Approximation} Requires Exponential Width}
\label{appendix:lower-bound}

\begin{claim}
    For any constant $\frac{1}{2}<\alpha<1$, there is a graph $G$ such that for \emph{any} set of weight classes (not necessarily a weight partition), $[\ell_1,r_1),[\ell_2,r_2),\dots,[\ell_k,r_k)\subseteq \R$, the following holds: if all possible choices of $\alpha$-approximate MWM $M_i$ of  $G_{[\ell_i,r_i)}$ satisfy
    \[\mu_w(M_1\cup M_2\cup \dots\cup M_k)\geq (\alpha-\delta)\cdot \mu_w(G),\]
    then there is a weight class $[\ell_i,r_i)$ such that $r_i\geq \ell_i\cdot \exp(\Omega(\delta^{-1}))$.
\end{claim}

We use a similar gadget for $\alpha$-approximation as before.
\begin{definition}[Level-$i$ Gadget for $\alpha$-Approximation]
A level-$i$ Gadget for $\alpha$-approximation is a path with three edges $a_i,b_i,c_i$ such that the edge weights of $a_i$ and $b_i$ are $\beta^i$, and the edge weight of $c_i$ is $\beta^{i+1}$, where $\frac{\beta}{\beta+1}=\alpha$.
\end{definition}

\begin{proof}
We construct a graph $G$ that contains $\beta^{N-i}$ number of level-$i$ gadgets for each $0\leq i\leq N-1$, where $N=\lfloor\frac{\alpha-\alpha^2-(1-\alpha)^2}{\delta}-1\rfloor$. Thus $\mu_w(G)=N\cdot \beta^N\cdot (1+\beta)$. We will construct a sparsifier $S\subseteq G$ such that for any weight class $[\ell_i,r_i)$ that doesn't contain the entire graph, there is an $\alpha$-approximate matching $M_i$ of $G_{[\ell_i,r_i)}$ in $S$. That means if no weight classes in $[\ell_1,r_1),[\ell_2,r_2),\cdots,[\ell_k,r_k)\subseteq \R$ contain the entire graph, there is a choice of matchings such that $M_1\cup M_2\cup\cdots\cup M_k\subseteq S$. On the other hand, the construction of $S$ will ensure that $\mu_w(S)< (\alpha-\delta)\cdot\mu_w(G)$ suggesting that at least one weight class contains the entire graph and thus has $\exp(\Omega(\delta^{-1}))$ width.

Now we explicitly construct $S$. For all $0\leq i\leq N-2$, $S$ contains all $b_i$ and $c_i$ edges in level-$i$ gadgets. Also, $S$ contains all $b_{N-1}$ edges and $\alpha$ fraction of $c_{N-1}$ edges in level-$(N-1)$ gadgets. Consider any weight class $[l,r)$.
\begin{enumerate}
    \item Suppose $r<\beta^{N}$. Let $j=\lfloor\log_{\beta} r\rfloor$. By picking all $c_{i}$ edges in any corresponding level $i<j$ intersecting with the weight class and all $b_j$ edges in level $j$, there is an $\alpha$-approximate MWM.
    \item Suppose $r\geq \beta^{N}$. If $l>\beta^{N-1}$, then $S$ clearly contains an $\alpha$-approximation. Otherwise, if $l>1$, by picking all $c_i$ edges in any level $i<N-1$, there is an $\alpha$-approximate MWM.
\end{enumerate}
So far, we have shown that if none of the weight classes contains $G$, $S$ can consistently output an $\alpha$-approximate MWM. However,
\[\mu_w(S)=(N-1)\cdot \beta^{N+1}+\alpha\cdot \beta^{N+1}+(1-\alpha)\cdot \beta^N,\]
and the largest matching in $S$ has approximation ratio
\[\frac{(N-1)\cdot \beta^{N+1}+\alpha\cdot \beta^{N+1}+(1-\alpha)\cdot \beta^N}{N\cdot \beta^N\cdot (1+\beta)}=\alpha-\frac{1}{N}\left(\alpha-\alpha^2-(1-\alpha)^2\right)<\alpha-\delta.\]
\end{proof}

\end{document}